\documentclass[12pt]{article}

\usepackage{amsmath,amsfonts,amssymb,amsthm,amsbsy}
\usepackage{authblk}
\usepackage{graphicx}
\usepackage{pdfpages}
\usepackage{subfig}
\usepackage{adjustbox}
\usepackage{caption}
\usepackage[ruled,vlined]{algorithm2e}
\usepackage{bbm}
\usepackage{breakcites}
\usepackage{tabularx}
\usepackage{geometry}
\usepackage{enumitem}
\usepackage{natbib}
\usepackage{authblk}
\usepackage{graphicx}
\usepackage{mathtools}
\usepackage{xcolor}
\usepackage{tcolorbox}
\usepackage{setspace}
\usepackage{lipsum}
\usepackage{hyperref}
\usepackage[normalem]{ulem}

\providecommand{\keywords}[1]
{
  \small	
  \textbf{{Keywords:}} #1
}

\theoremstyle{definition}
\newtheorem{definition}{Definition}[section]
\newtheorem{assumption}{Assumption}
\newtheorem{theorem}{Theorem}[section]
\newtheorem{corollary}{Corollary}[theorem]
\newtheorem{lemma}[theorem]{Lemma}
\newtheorem{guideline}{Guideline}
\newtheorem{remark}{\textbf{Remark}}[section]

\newcommand{\independent}{\perp \!\!\! \perp}

\usepackage{color}

\title{Benefits and costs of matching prior to a Difference in Differences analysis when parallel trends does not hold\thanks{We thank the authors of \citet{principal_turnover} for supplying their data and replication source code. This work was supported by the U.S. Department of Education, Institute for  Education Sciences, through Grant R305D200010. The opinions expressed are those of the authors and do not represent views of the Institute or the U.S. Department of Education.}}

\author[1]{Dae Woong Ham}
\author[2]{Luke Miratrix}
\affil[1]{\textit{Department of Statistics, Harvard}}
\affil[2]{\textit{Harvard Graduate School of Education}}

\date{}
\begin{document}
\maketitle
\vspace{-0.5cm}
\begin{abstract}
The consequence of a change in school leadership (e.g., principal turnover) on student achievement has important implications for education policy.
The impact of such an event can be estimated via the popular Difference in Difference (DiD) estimator, where those schools with a turnover event are compared to a selected set of schools that did not have such an event.
The strength of this comparison depends on the plausibility of the ``parallel trends'' assumption that the ``treated group'' of those schools which had leadership turnover, absent such turnover, would have changed ``similarly'' to those which did not.
To bolster such a claim, one might generate a comparison group, via matching, that is similar to the treated group with respect to pre-treatment outcomes and/or pre-treatment covariates. 
Unfortunately, as has been previously pointed out, this intuitively appealing approach also has a cost in terms of bias.
To assess the trade-offs of matching in our application, we first characterize the bias of matching prior to a DiD analysis under a linear structural model that allows for time-invariant observed and unobserved confounders with time-varying effects on the outcome.
Given our framework, we verify that matching on baseline covariates generally reduces bias.
We further show how additionally matching on pre-treatment outcomes has both cost and benefit.
First, matching on pre-treatment outcomes partially balances unobserved confounders, which mitigates some bias.
This reduction is proportional to the outcome's reliability, a measure of how coupled the outcomes are with the latent covariates.
Offsetting these gains, matching also injects bias into the final estimate by undermining the second difference in the DiD via a regression-to-the-mean effect.
Consequently, we provide heuristic guidelines for determining to what degree the bias reduction of matching is likely to outweigh the bias cost.
We illustrate our guidelines by reanalyzing a principal turnover study that used matching prior to a DiD analysis and find that matching on both the pre-treatment outcomes and observed covariates makes the estimated treatment effect more credible.
\end{abstract}


\newpage

\section{Introduction}
Principal turnover rates in the United States are high, with some estimates at 18\% nationally, which has led to much policy attention and research to understand the effect of principal turnover on student achievement \citep{principal_rate, principal_research1, principal_research2}.
In particular, a recent empirical study \citep{principal_turnover} found statistically significant negative effects of principal turnover on student achievement through a difference in difference (DiD) estimator.
The DiD estimator is a popular method in causal inference for measuring treatment effects in observational panel data \citep{abad:05, DiDPopularity1, DiDpopularity2}.
The DiD estimator uses the change of a control group to estimate how much change a targeted treated group would have experienced over time, absent treatment.
In such an approach, researchers assume that the control group's change over time is equivalent to the treatment group's change over time absent any treatment effect.
This assumption is often referred to as the ``parallel trend assumption'' \citep{abad:05} and is used as the main identifying assumption justifying the DiD estimator.
One core concern with this assumption is whether the chosen comparison group is suitable.

To bolster the plausibility of the parallel trends assumption, \citeauthor{principal_turnover} performed matching prior to the DiD analysis to select a control group where parallel trends seems more plausible.
Matching methods are a class of methods in causal inference that allows direct comparison of the treated and control groups by selecting units from a pool of possible control units that resemble the treatment group in terms of baseline, observable characteristics \citep{matching1, matching2}.
The hope is, if our observed characteristics are informative enough, that the control units would be, other than the receipt of treatment, just like the treatment units in other aspects.
For the DiD case, in particular, we would hope our matched controls would change just like the treated units would have, absent treatment.
Given this intuition, many other previous empirical works have also used matching prior to their DiD analysis \citep{DiDmatch_empirical1, DiDmatch_empirical2, DiDmatch_empirical3_kosuke}; in these studies, the authors believe that the key parallel trend assumption is more likely to hold conditional on similar observed covariates and lagged outcomes.
This assumption is often referred as ``conditional parallel trends'': the evolution of treated and control units would be the same, absent treatment, for units with the same observable characteristics.
Although conditional parallel trends may be more likely to hold in certain scenarios, \citet{matchDiDSim} show that when unconditional parallel trends hold, matching on pre-treatment outcomes can actually \emph{induce bias} into a perfectly unbiased estimator.
We are thus faced with a dilemma: matching intuitively gives us a more plausible control group, but the act of matching injects bias into our estimator.
This is the dilemma we explore in this paper. 

Engaging with this dilemma is not new.
Many works analyze matching prior to a DiD analysis and show that matching may actually hurt or help depending on different scenarios \citep{matchDiDSim, ding_bracketing, DiDMatching2, DiDmatching4, chabesim, chabe_2017, gain_scores1, gain_scores2}.
These existing works can be roughly grouped into three categories.
The first are those that characterize the bias when either parallel trends or conditional parallel trends perfectly holds \citep{ding_bracketing, matchDiDSim}.
The second are those that characterize the bias in more general settings via simulations \citep{DiDMatching2, DiDmatching4}.
The third give sufficient conditions for when matching combined with DiD gives perfectly unbiased estimates, as opposed to comparing which of the biased estimates may lead to the minimum bias \citep{chabe_2017, chabesim}.\footnote{The work by this author also considers a different data generating model than ours; we discuss this further in Section~\ref{subsection:LSEM}.}
 We seek to provide exact mathematical characterizations of the bias under a general framework that allows for both imperfect conditional and unconditional parallel trends driven by unobserved and observed time-invariant confounders with time-varying relationships to the outcome.
In other words, all estimators, regardless of matching, are biased in our setting as neither parallel trends nor conditional parallel trends necessarily holds, and we explore how to identify which strategy produce the \textit{least} (as opposed to zero) bias.

In particular, we characterize the bias of different matching DiD estimators in a linear structural equation model (detailed in Section~\ref{section:setup}) as a function of the underlying relationship between covariates, pre-treatment outcomes, breakage in parallel trends, and amount of confounding.
In Section~\ref{section:main_results}, we present the main results of the bias for the simple unmatched DiD estimator, the DiD estimator after matching on observed covariates, and the DiD estimator after matching additionally on the pre-treatment outcome.
We find that matching on observed covariates generally leads to a reduction in bias.
This improvement is guaranteed unless unlikely scenarios occur such as pre-existing biases cancelling each other out.
This result is also consistent with the ``pre-treatment criterion'' that suggests practitioners always control for observed covariates when estimating a causal effect \citep{rose:02b, rose:rubi:83, shpi:vand:robi:10}, although our additional results do undermine this principle.

In contrast to matching on our baseline covariates, matching on the pre-treatment outcome exhibits a trade off.
On one hand, matching on the pre-treatment outcome does indirectly partially match on the latent confounders, thus mitigating their bias contribution.
The amount in reduction is proportional to a key quantity known as reliability, a measure of how coupled the outcomes are with these latent covariates \citep{reliability1, reliability2}.
On the other hand, matching on the pre-treatment outcome undermines the second ``difference'' in the DiD estimator by forcing the treated and control group's pre-treatment outcomes to be equal \citep{gain_scores1, gain_scores2}.
Therefore, if parallel trends originally held, matching on the pre-treatment outcome breaks an initially unbiased estimator.
On balance, as we show in Lemma~\ref{lemma:suffcond_match}, if the reliability is higher than the breakage of parallel trends, it is better to match on the pre-treatment outcome.

If we view pre-treatment outcomes as just another baseline covariate, the above results connect to matching on covariates measured with error \cite{webb2017imputation, lenis2017doubly, rudolph2018using}.
In particular, when we match on such variables, we only imperfectly match on the underlying, latent, variable directly connected to our post-treatment outcomes, leaving residual imbalance.
The DiD estimator can then exacerbate this residual imbalance by, as described in \cite{matchDiDSim}, via a regression to the mean effect.
These results show the ``pretreatment criterion'' does not always apply for difference-in-difference contexts.

For clarity of exposition we initially focus on the simple pre-post (or two time point) DiD.
After we present these results, however, we extend to a more general setting with multivariate covariates and multiple pre-treatment outcomes in Section~\ref{section:generalization}.
We find that matching on multiple pre-treatment outcomes helps further recover the latent confounder, leading to a greater reliability and reduction in bias.
In addition to characterizing the bias of the DiD after matching, we also provide heuristic guidelines in Section~\ref{section:sensitivity_analysis} to practitioners on whether matching prior to their DiD analysis is sensible, along with a method for roughly estimating the reduction in bias.
Finally in Section~\ref{section:application}, we revisit our principal turnover application, where the authors found a statistically negative treatment effect of principal turnover on student achievement.
We show that matching on both observed covariates and pre-treatment outcomes helped reduce the bias of the estimate, providing further support for the authors' justification to match, and lending further credibility for their reported treatment effect.

\section{Empirical Application - Impact of Principal Turnover}
\label{section:empirical_example}
Many studies suggest principals play a crucial role in building a learning climate, supporting teacher improvement, and instructional leadership \citep{principal_relevance1, principal_relevance2, principal_relevance3, principal_relevance4}; we might believe that turnover would undermine such benefits.
In this vein, a recent study by \citep{principal_turnover} indeed found short term negative effects of principal turnover on public schools in Missouri and Tennessee.
The Missouri data, initially obtained by the Department of Elementary and Secondary Education, contains 2,400 public schools in 565 districts, and span the years of 1991 to 2016.
These data record whether there was a turnover in any given year, and also contain outcomes such as school-average math scores on statewide exams in Grades 3 to 8.
The original authors also drew on school demographics from the Common Core of Data (CCD) that contained information about student enrollment size, proportion of Black and Hispanic students, and the proportion of students receiving free lunch \citep{principal_turnover}.
The main years of analysis were between 2001-2015 to allow for multiple lagged outcomes.
As our running example, we use these Missouri data, generously provided by the initial authors.

In our analysis, the main treatment variable is a binary principal turnover variable, which takes value of 1 if a principal in a school in year $t$ is not the principal in the school in year $t +1$ and 0 otherwise.
We focus on estimating the immediate impact of principal turnover on student's math scores.
To answer this, the initial authors estimated a causal effect using a DiD analysis after matching. 
\citeauthor{principal_turnover} matched prior to DiD because ``even conditional on school fixed effects and the other controls in the model ... parallel trends do not hold between treatment and comparison schools.''
The authors argue that since principals could be leaving due to factors that are also negatively impacting the schools, the schools that have principals departing most likely already have a declining student achievement trend in the years leading to the principal turnover. 

To address this challenge, the authors construct a matched comparison group, matching on school demographics and pre-treatment outcomes.
Then they obtain an estimated treatment effect via a DiD analysis through regression.
We replicated their analysis, with results on Figure~\ref{fig:math_results}.
Details of our replication of this matching step and the model we used for the DiD analysis are in Appendix~\ref{Appendix:application_details}.
The difference in Year 1 (first post-treatment year) between the treated and control schools indicate  a substantial negative impact of principal turnover.
Furthermore, matching provided a control group (in triangles) with a similar downward trend prior to the principal turnover as we see for the treatment group (in circles).
Consequently, the authors argue that matching on the pre-treatment outcomes was essential in forming a valid control and treatment group comparison.

\begin{figure}[t!]
\begin{center}
\includegraphics[width=9cm, height = 5cm]{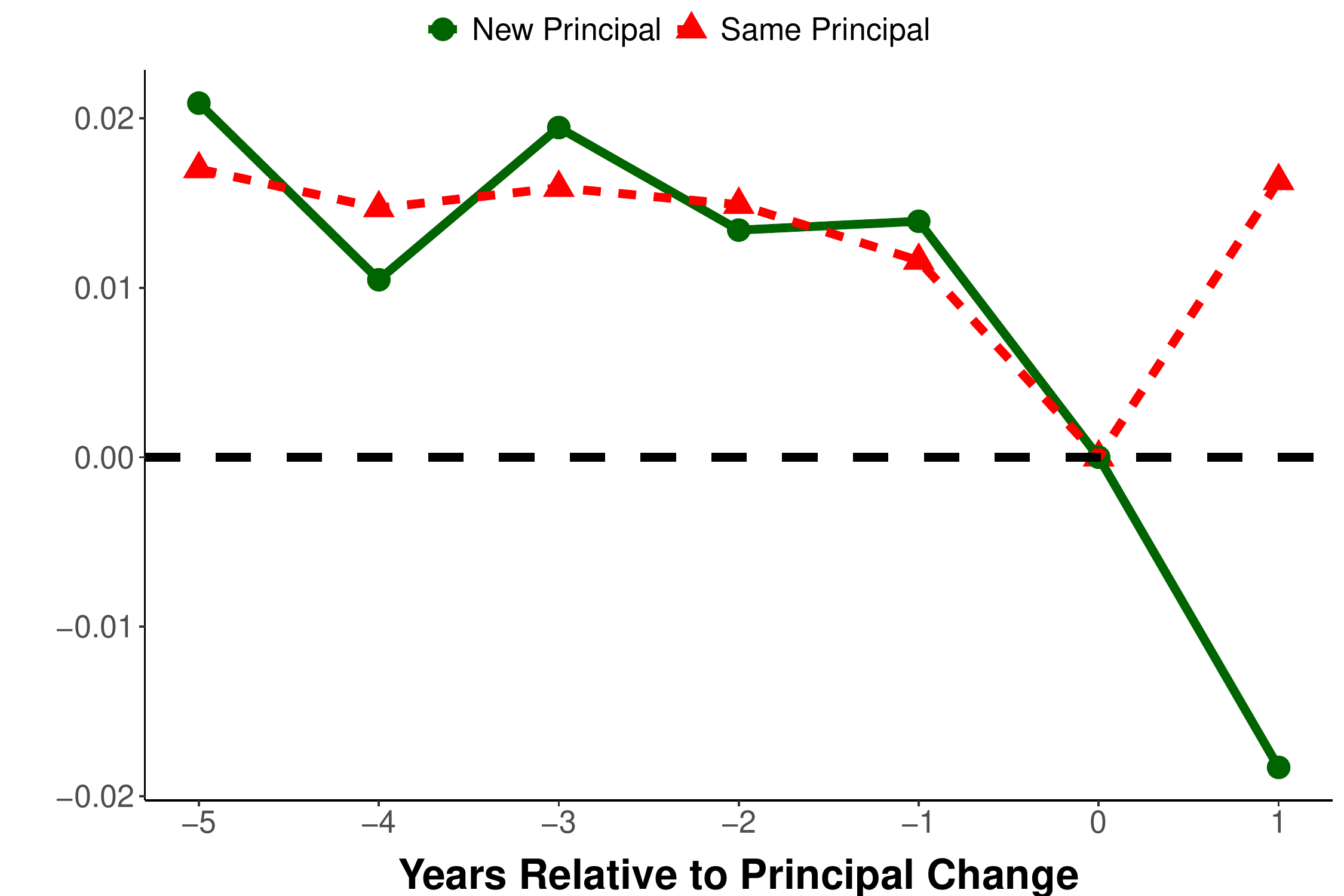}
\caption{This event-study style plot shows the estimated margins for treatment and comparison schools relative to year of a principal change (Year 0) for Missouri. The estimates are obtained via a DiD after obtaining a comparable control group through matching on both covariates and lagged outcomes.}
\label{fig:math_results}%
\end{center}
\end{figure} 

In this paper, we work to formalize this choice.
On one hand, \cite{matchDiDSim} shows matching can break parallel trends, but on the other, parallel trends is not initially holding.
We investigate this bias-bias trade-off in Section~\ref{section:application} with our practical guidelines that we present in Section~\ref{section:sensitivity_analysis}.
As part of this we also calculate rough estimates of the reduction in bias obtained from matching due to the covariates, and due to the covariates along with pre-treatment outcomes.
We find that the authors' decision to match on observed covariates and pre-treatment outcomes was indeed likely to have reduced overall bias, providing a more credible treatment effect estimate.

\section{A Working Model for Evaluating Matching Prior to DiD}
\label{section:setup}

\subsection{A Linear Structural Equation Model}
\label{subsection:LSEM}
We present our main results assuming a linear structural equation model following similar studies related to DiD and matching (see \citep{DiDmatching4, matchDiDSim} for examples).
Although our model assumes a specific parametric setting, we believe it nevertheless serves as a useful starting point for understanding the main tradeoffs at play.

Let $Y_{i, t}$ be the measured outcome of individual $i$ at time $t$.
We first focus on the classic DiD setting with one pre-treatment period at $t = 0$ and one post-treatment period at $t = 1$. 
In this pre-post setting, $Y_{i, 1}$ is our post-treatment outcome and $Y_{i,0}$ is the pre-treatment outcome.
No one has received treatment at $t  = 0$, and some of our units have received treatment at $t=1$.
We extend our results to multiple pre-treatment periods in Section~\ref{section:generalization}.

Denote $Z_i$ as a binary treatment indicator for individual $i$, with $Z_i = 1$ indicating membership in the treatment group.
We have an observed covariate $X_i$ for each individual $i$ and an unobserved latent covariate $\theta_i$.
We first focus on the case when both $X_i$ and $\theta_i$ are one-dimensional and extend our results to a multivariate setting in Section~\ref{section:generalization}.
We assume a model of both the latent and observed covariates with the following conditional multivariate Gaussian distribution,
\begin{equation}
\begin{aligned}
  \begin{pmatrix} \theta_i \\ X_i \end{pmatrix} \mid Z_i = z &\sim N\left(\begin{pmatrix} \mu_{\theta, z} \\ \mu_{x, z}  \end{pmatrix} ,  \begin{pmatrix} \sigma_{\theta}^2 & \rho \sigma_{\theta}\sigma_x \\ \rho \sigma_{\theta}\sigma_x & \sigma_{x}^2    \end{pmatrix}\right) \quad z = 0, 1,  \\
  Z_i & \sim \text{Bern}(p).
\end{aligned}
\label{eq:simplethetax}
\end{equation}
Both $\theta$ and $X$ are imbalanced across treatment groups when $\mu_{\theta, 1} \neq \mu_{\theta, 0}$ and $\mu_{x, 1} \neq \mu_{x, 0}$. 

The above model implies a selection mechanism of 
\begin{equation}
P(Z_i = 1\mid X_i = x, \theta_i = \theta) = \frac{pN_1(\theta, x)} {pN_1(\theta,x) + (1-p)N_0(\theta, x)},
\label{eq:selection_mech}
\end{equation}
where $N_z(\theta, x)$ is the probability density of $(\theta, X)$ in Equation~\eqref{eq:simplethetax} given $Z_i = z$. 
I.e., our model is equivalent to first generating $(\theta_i, X_i)$ pairs from some distribution and then assigning treatment to the units with the above probabilities (see Remark~\ref{remark:assignment_mech} for alternate assignment mechanisms found in the literature). 

We use the potential outcomes framework, and assume the Stable Unit Treatment Value Assumption, i.e., no interference between units and no multiple forms of treatment (see \citet{SplawaNeyman:1990} for their historical introduction and \citet{Rosenbaum:DesignObsStudy} for an overview).
In particular, we generate the potential outcomes $Y_{i,t}(z)$, which denotes the outcome we would see at time $t$ had unit $i$ received treatment $z$, and the observed outcome $Y_{i,t}$, with the following linear model, 
\begin{equation}
\begin{aligned}
  Y_{i,t}(0) &= \beta_{0, t} + \beta_{\theta, t} \theta_{i}  + \beta_{x, t} X_{i} + \epsilon_{i, t}, \\
  Y_{i,t}(1) &= Y_{i,t}(0) + \tau_i \mathbf{1}(t = 1), \\
  Y_{i,t} &= Z_i Y_{i,t}(1)  + (1- Z_i) Y_{i,t}(0),
\end{aligned}
\label{eq:simpleresponsemodel}
\end{equation}
where $\epsilon_{i, t} \sim N(0, \sigma_{E}^2) \independent (\theta_i, X_i, Z_i)$ (we assume homoscedastic errors across time), $\epsilon_{i, t} \independent \epsilon_{i, t'}$ for all $t, t'$, and  $\beta_{0, t}, \beta_{\theta, t}, \beta_{x,t}$ are fixed constants that denote the intercept at time period $t$ and the slopes of $\theta$ and $X$ at time $t$, respectively.
The vector of random variables $(Z_i, \theta_i, X_i, Y_{i,t}(0), Y_{i,t}(1),Y_{i,t})$ are independently and identically distributed according to the data generating process described in Equations~\ref{eq:simplethetax} and \ref{eq:simpleresponsemodel} across individuals.
While both $\theta$ and $X$ have time varying effects on the outcome as represented by the $\beta_{*,t}$, $\theta$ and $X$ themselves are time invariant, i.e., they do not change over time (although the pre-treatment outcome, viewed as a covariate, does in the case of multiple pre-treatment periods).
See Remark~\ref{remark:time_varying} for further details and connection to time varying confounders  and Section~\ref{subsection:PT} for connections to the principal turnover motivating example.

We let individual units have individual treatment effects $\tau_i = Y_{i, 1}(1) - Y_{i, 1}(0)$, with our estimand being the average treatment effect on the treated (ATT) of 
\begin{equation}
\tau := E(Y_{i, 1}(1) - Y_{i, 1}(0) \mid Z_i = 1) . \label{eq:ATT_definition}
\end{equation}
We assume no anticipation of treatment effect, with $Y_{i,t}(0) = Y_{i,t}(1)$ for $t = 0$; this is ensured by the above indicator $\mathbf{1}(t = 1)$. 
%Equation~\eqref{eq:simpleresponsemodel} shows $\tau$ is our true average causal effect and the third line of Equation~\eqref{eq:simpleresponsemodel} assumes the consistency of potential outcomes \citep{rubin:imbens}, where the indicators are added so we observe control potential outcomes for all units in the pre-period.
%Equation~\eqref{eq:simpleresponsemodel} implicitly assumes the common SUTVA assumption that a unit's potential outcome is only a function of its treatment.

We critically allow the slopes of $\theta, X$ ($\beta_{\theta, t}, \beta_{x,t}$) to differ across time, which allows the confounders $\theta$ and $X$ to have a time varying effect on the response, thus breaking the parallel trends assumption.
To capture the breakage in parallel trends succinctly, define the imbalance of $\theta$ and the time varying change in the effect of $\theta$ as $\delta_{\theta}$ and $\Delta_{\theta}$, respectively, where,  
\begin{align*}
\delta_{\theta} & := \mu_{\theta, 1} - \mu_{\theta, 0}   \qquad \mbox{imbalance} ,\\
\Delta_{\theta} & := \beta_{\theta, 1} - \beta_{\theta, 0} \qquad  \mbox{time variation of effect.}
\end{align*}
We define $\delta_x$ and  $\Delta_x$ similarly for the imbalance and time varying effect of $X$. Finally, define $\Delta_0 = \beta_{0,1} - \beta_{0,0}$ as the baseline growth.
Under this notation, the growth of a unit $i$ from pre to post is a function of the observed and latent covariates, with an expected growth of $\Delta_0 + \Delta_\theta \theta_i + \Delta_X X_i$; if the treatment and control groups have different covariate distributions, they can have systematically different growth rates if $\Delta_\theta$ or $\Delta_x$ are nonzero.

Although we simplify our model by not including interactions between $\theta$ and $X$, we analyze the interaction case in Appendix~\ref{Appendix:interaction}; we find that it does not offer more insight than that already provided in the no-interaction case. 
We end this subsection with two remarks connecting our setup to existing literature.

\begin{remark}[Time varying confounders]
\label{remark:time_varying}
The outcome model we consider in Equation~\eqref{eq:simpleresponsemodel} is nearly equivalent to the one considered in \citep{DiDmatching4} except our covariates $(\theta_i, X_i)$ are time invariant while the aforementioned authors consider time varying ($\theta_{it}, X_{it}$) covariates.
We do not consider the time varying covariate case (other than the pre-treatment outcome) for three reasons:
1) As also noted by \citet{DiDmatching4}, the time invariant covariate case is interesting in itself as it allows imperfect parallel trends and leads to surprising and illustrative results when some covariates are unobserved, as shown in Section~\ref{section:main_results}.
2) By focusing on time-invariant covariates, we examining DiD in a context where we are interested in the ideal of matching units with similar characteristics; time varying covariates are then a consequence of these stable characteristics, and we represent them via the changing coefficients.
3) The time varying covariate case would require more structural assumptions on how the covariates evolve over time. 
\citet{DiDmatching4} detail that time varying covariates may evolve either a) completely unrelated to the treatment (such as with an AR(1) process), b) related to the treatment, or c) with a combination of the treatment and other factors. 
All three cases would require structurally modeling the time varying confounder, where settings b) and c) poses thorny issues due to post-treatment matching.
For these reasons, we choose to present our paper with the more simplified, but still widely used, time invariant covariates that have time varying effects as a starting point.
That said, also see discussion in Section~\ref{section:discussion} for further discussion of time-varying settings.
%We further briefly detail how our setup can also mathematically be re-parameterized to account for simple time varying confounders in Section~\ref{section:discussion}.
\end{remark}

\begin{remark}[Alternate assignment mechanisms]
\label{remark:assignment_mech}
The assignment mechanism we consider in Equation~\eqref{eq:selection_mech} presumes that the units select themselves into treatment or control as a function of $(X_i, \theta_i)$, their baseline characteristics (both observed and latent), as opposed to the actual pre-period outcome $Y_{i, 0}$.
Other evaluations have used alternative selection mechanisms; for example, \citep{chabe_2017, chabesim, DiD_new} instead consider assignment mechanisms that are purely a function of the pre-period outcome with no covariates.
In other words, this alternate mechanism can assign units to treatment due to chance deviations in the pre-period (this is the mechanism that would generate, for example, Ashenfelter’s dip \citep{ashen_dip}). We leave the connection and extension of our results for this setting to future work.
\end{remark}

\subsection{Parallel Trends}
\label{subsection:PT}
The key assumption behind DiD is the parallel trend assumption.
The parallel trends assumption states that the change between the pre-treatment and post-treatment control potential outcomes for treatment units is the same, on average, as that of the control units.
More formally,
\begin{equation}
 \begin{aligned}
    \underbrace{E(Y_{i, 1}(0) - Y_{i, 0}(0) \mid Z_i = 1)}_{\text{Change in treatment group}}  &=   \underbrace{E(Y_{i, 1}(0) - Y_{i, 0}(0) \mid Z_i = 0)}_{\text{Change in control group}}. \\
\end{aligned} 
\label{eq:paralleltrends}
\end{equation}
Under our linear model in Equation~\eqref{eq:simpleresponsemodel}, Equation~\eqref{eq:paralleltrends} is equivalent to,
$$\beta_{0,1} - \beta_{0,0} + \Delta_{\theta}\mu_{\theta, 1} + \Delta_{x}\mu_{x, 1} =   \beta_{0,1} - \beta_{0,0} + \Delta_{\theta}\mu_{\theta, 0} + \Delta_{x}\mu_{x, 0}, $$
or
$$ \Delta_{\theta}(\mu_{\theta, 1} - \mu_{\theta, 0}) + \Delta_{x}( \mu_{x, 1} -  \mu_{x, 0}) = \Delta_{\theta}\delta_{\theta} + \Delta_x\delta_x = 0 .$$
Since the pre-treatment period's potential outcome is unaffected by the treatment status, the potential outcome $Y_{i, 0}(0)$ is observed for both the treatment and control group. The post-treatment potential outcome, $Y_{i, 1}(0)$, is also observed for the control group.
However, $Y_{i, 1}(0)$ for the treatment group is unobserved, making it the main quantity to identify in a DiD analysis.
Equation~\eqref{eq:paralleltrends} allows us to impute this quantity to obtain an estimate of the ATT (Equation~\ref{eq:ATT_definition}).

Any departure from Equation~\eqref{eq:paralleltrends}, i.e., if 
$$\Delta_{\text{PT}} := E(Y_{i, 1}(0) - Y_{i, 0}(0) \mid Z_i = 1)  -   E(Y_{i, 1}(0) - Y_{i, 0}(0) \mid Z_i = 0) \neq 0,$$
leads to a biased DiD estimator.
Parallel trends holds despite the presence of observed or unobserved imbalance between the covariates as long as the expected \textit{change across time} in both the treatment and control is still same.
Consequently, \citet{DiDmatching4} state that typical confounders , i.e., covariates with different means within a single time point, are not necessarily confounders in a DiD study.
Instead, they say covariates ``that differ by treatment group and are associated with outcome trends are confounders'' (pg. 2).
In our model, not only are $(\theta, X)$ typical confounders but also (time invariant) confounders in a DiD study since $(\theta, X)$ can break parallel trends due to their time varying effects ($\beta_{\theta, t}, \beta_{x, t}$) on the response.
Therefore, we refer to $(\theta, X)$ as confounders throughout this paper.

For concreteness, consider the principal turnover application.
Say we observe the proportion of students on Free and Reduced Price Lunch for each school (this is an indicator of the proportion of the student body that live in lower income families).
This may be imbalanced since schools that have more principals departing may also be schools with higher or lower proportion of FRPL students, thus $\delta_x \neq 0$.
Now say there was increased state-wide government sponsored funding for schools above some threshold of FRPL in the post-treatment years.
This could change the relationship between the proportion of FRPL students and overall student achievement from the pre-treatment to post-treatment years, i.e., $\Delta_x \neq 0$.
An example of an unobserved confounder $\theta$ could be how well the school is able to adapt to statewide changes in the curriculum.
Now, if the state testing were changing, this could change the relationship of our outcome and $\theta$ across time.
Furthermore, we would expect $\theta$ to be correlated, $\rho \neq 0$, with many observed covariates $X$ (e.g., some underfunded schools may be less able to adapt to changing curriculum).

\subsection{DiD and Matching}

We now formally introduce the relevant DiD and matching estimators. For the rest of this section, we work with the expected values of each estimator (as opposed to the finite-sample estimators) to focus on the bias.
We first introduce the classic unmatched DiD estimator $\hat\tau_{DiD}$ \citep{abad:05}.
The expected value of $\hat\tau_{DiD}$ is,
\begin{equation}
 E\big[\hat{\tau}_{DiD}\big] = \{E(Y_{i, 1}\mid Z_{i} = 1) -  E(Y_{i, 1}\mid Z_{i} = 0)\} - \{ E(Y_{i, 0}\mid Z_{i} = 1) -  E(Y_{i, 0}\mid Z_{i} = 0)\},
 \label{eq:naiveDiD}
\end{equation}
\noindent where the expectation is over the joint distribution of $(X_i, \theta_i, \epsilon_{i, 0}, \epsilon_{i, 1})$. We refer to $\hat\tau_{DiD}$ as the na\"ive DiD estimator as it does not use matching at all. The na\"ive DiD estimator simply takes the difference of the differences between each control and treated group. It is unbiased if the parallel trend assumption in Equation~\eqref{eq:paralleltrends} is satisfied.

For a matched DiD estimator, we would first find a control unit for each treatment unit that shares the same value of $X$ (or both $X$ and $Y_0$, if matching on both).
We do this for each treatment unit and fit the resulting DiD estimator to this matched data.
For the purpose of the paper, we assume arbitrarily close-to-perfect matching.
In other words, if we are matching on $X$, we assume that for every treated individual $i$, we obtain a control individual $j$ such that $X_{i} = X_{j}$.
Although perfect matching rarely occurs in practice, it allows us to isolate bias from matching as an approach vs. its implementation.
It can also be viewed from an asymptotic argument, if the matching variable has finite support; see, e.g., \citep{perfect_matching_asymptotics}.
Additionally, many empirical works that assume conditional parallel trends also implicitly assume they have achieved perfect matching since conditional parallel trend requires parallel trends to hold given both the treated and control group have the \textit{exact} same values of the matched variable(s) \citep{DiDmatch_empirical3_kosuke}.
We leave fully characterizing the bias under imperfect matches to future work. 

Before introducing the matched DiD estimators, we define some additional notation to formally denote the matching step.
Suppose we are matching a treatment unit with value $X_i = x$.
Denote the expected observed outcome of the matched control unit as $E(Y_{i, 1} \mid Z_i = 0, X_i = x)$.
Now, we will have a control unit for each treated unit, so the average of these control units would be the average of the above over the distribution of $X_i$ in the treated group, giving an overall expected average outcome of 
$$E_{x \mid Z_i = 1}[E(Y_{i, 1}\mid Z_{i} = 0, X_i = x)] ,$$
where $E_{x \mid Z_i = 1}(\cdot)$ denotes the integral over the distribution of $X_i$ conditioned on $Z_i = 1$.

We have three estimators, depending on how we match: matching on $X_i$, matching on pre-treatment outcomes, and matching on both.
We found matching only on the pre-treatment outcome and not on an available covariate $X$ was both practically unsound and also unhelpful for providing insight to the benefits of matching.
Consequently, we relegate the analysis for this estimator to Appendix~\ref{Appendix:only_pre_treatment}, although we do analyze the case of matching on a pre-treatment outcome when there is no $X$ below.

The expected outcomes of the two remaining matching estimators (assuming perfect matches) of matching on the observed covariates $X$ and matching on both $X$ and the pre-treatment outcome $Y_{0}$ are then:
\begin{equation}
\begin{aligned}
E\big[\hat{\tau}_{DiD}^{X} \big] &= E(Y_{i, 1}\mid Z_{i} = 1) -  E(Y_{i, 0}\mid Z_{i} = 1)  \\
 & \quad - (E_{x \mid Z_i = 1}[E(Y_{i, 1}\mid Z_{i} = 0, X_i = x)] -  E_{x \mid Z_i = 1}[E(Y_{i, 0}\mid Z_{i} = 0, X_i = x)]) ,\\
 E\big[\hat{\tau}_{DiD}^{X, Y_{0}} \big] &= E(Y_{i, 1}\mid Z_{i} = 1) -  E(Y_{i, 0}\mid Z_{i} = 1) \\
 &\quad - (  E_{(x,y) \mid Z_i = 1}[E(Y_{i, 1}\mid Z_{i} = 0, X_i = x, Y_{i, 0} = y)] \\
 &\quad -  E_{(x,y) \mid Z_i = 1}[E(Y_{i, 0}\mid Z_{i} = 0, X_i = x, Y_{i, 0} = y)]) ,\\
 &=  E(Y_{i, 1}\mid Z_{i} = 1) - E_{(x,y) \mid Z_i = 1}[E(Y_{i, 1}\mid Z_{i} = 0, X_i = x, Y_{i, 0} = y)] .
\end{aligned}
\label{eq:matchingDiD_estimators}
\end{equation}
Importantly, the expression for $E\big[\hat{\tau}_{DiD}^{X, Y_{0}} \big]$ reveals one consequence of matching: we no longer have a second difference correction term because we have matched it away \citep[see][for a futher discussion]{gain_scores2, gain_scores1}.
We will later see that if the matching did not achieve true balance on the latent aspects, this lack of correction term leads to bias.

As a final baseline of comparison, we also consider the simple difference in means estimator $E\big[\hat{\tau}_{DiM}\big] = E(Y_{i, 1}\mid Z_i = 1) - E(Y_{i, 1}\mid Z_i = 0)$, which takes the difference in means for the treated and control using only the post-treatment outcomes. We will refer to this estimator as the simple DiM estimator. We now present the main results of when matching prior to a DiD analysis may help or hurt. 

\section{The Biases of the Estimators}
\label{section:main_results}
We first present the bias results in full generality and then consider special cases to gain intuition on what is driving the bias. 
\begin{theorem}[Bias of DiD and Matching DiD estimators]
If $(Z_i, X_i, \theta_i, Y_{i,t})$ are independently and identically drawn from the data generating process as shown in Equations~\ref{eq:simplethetax}-~\ref{eq:simpleresponsemodel}, then the bias of our estimators are the following,
\begin{align*}
E\big[\hat{\tau}_{DiM}\big] - \tau &= \beta_{\theta, 1}\delta_{\theta} + \beta_{x, 1}\delta_x, \\
E\big[\hat{\tau}_{DiD}\big]- \tau &= \Delta_{\theta}\delta_{\theta} + \Delta_x\delta_x ,\\
E\big[\hat{\tau}_{DiD}^X\big]- \tau &=  \Delta_{\theta}\bigg[ \delta_{\theta}  - \rho \frac{\sigma_{\theta}}{\sigma_x}\delta_x \bigg],\\
E\big[\hat{\tau}_{DiD}^{X, Y_{0}}\big]- \tau &=  \beta_{\theta, 1}\bigg[(1 -  r_{ \theta \mid x}) \bigg(\delta_{\theta} - \rho \frac{\sigma_{\theta}}{\sigma_x}\delta_x  \bigg)\bigg],
\end{align*}
where 
\begin{align*}
    r_{\theta \mid x} &= 1 - \frac{Var(Y_{i, 0} \mid Z_i = 0, X_i, \theta_i) }{Var(Y_{i, 0} \mid Z_i = 0, X_i) } 
    = \frac{(\beta_{\theta, 0})^2 \sigma_{\theta}^2(1-\rho^2) }{(\beta_{\theta, 0})^2 \sigma_{\theta}^2(1-\rho^2) + \sigma_{E}^2}.
\end{align*}
\label{theorem:mainresults}
\end{theorem}
The proof is provided in Appendix~\ref{appendix:proof_mainresult}.
Before moving to simpler cases, we make a few remarks.
First, the bias of the na\"ive DiD estimator is $\Delta_{\text{PT}} =  \Delta_{\theta}\delta_{\theta} + \Delta_x\delta_x$, i.e. is the degree to which parallel trend is broken.
Second, the simple DiM estimator is biased proportionally to the imbalance in $\theta$ and $X$ ($\delta_{\theta}, \delta_x$, respectively) and to how connected these are to the outcome ($\beta_{\theta,1}, \beta_{X,1}$) post treatment.
Lastly, the $r_{\theta \mid x}$ term is related to the reliability or consistency of a measure \citep{reliability1, reliability2}.
The reliability of a measure $Y$ captures how much $Y$ is a function of what it is measuring ($\theta$) over noise ($\epsilon$).
More formally stated:
\begin{definition}[Reliability]
The reliability $r_{\theta}$ of a random variable $Y$ as a measure of a random variable $\theta$ is
$$0 \leq r_{\theta} = 1 - \frac{Var(Y \mid \theta)}{Var(Y)} \leq 1.$$
\label{def:reliability}
\end{definition}
In particular, we use $r_{\theta \mid x}$, the (conditional) reliability of the pre-treatment outcome with respect to $\theta$ within the control group after controlling for $X$.
In our linear framework, $r_{\theta \mid x}$ can be interpreted as the population $R$-squared statistic if we were able to regress the pre-treatment outcome $Y_{i, 0}$ on $\theta_i$ after accounting for $X_i$ within the control group.
Alternatively put, $r_{\theta \mid x}$ can be interpreted as the square of the correlation between $\theta$ and the pre-treatment response after accounting for $X$ within the control group: $r_{\theta \mid x} = \text{Cor}^2(Y_{i,0}, \theta_i \mid X_i, Z_i = 0)$.
Throughout this paper, we often refer to $r_{\theta \mid x}$ as simply the reliability term.
This reliability cannot generally be directly estimated from the data and may differ greatly depending on the specific application and context.
We provide a heuristic way to estimate the reliability with multiple pre-treatment time points under some further assumptions in Section~\ref{section:sensitivity_analysis}.

Because it is not immediately obvious from the expressions in Theorem~\ref{theorem:mainresults} how matching on each variables is helping/hurting the DiD estimator, we will build intuition for these results by going through simple sub-cases first.
We first build intuition when there exist no covariate $X$ in Section~\ref{subsection:nocov}. We then bring back covariate $X$ but force a zero correlation with $\theta$ ($\rho = 0$) in Section~\ref{subsection:generalcase_uncorrelated}, and finally return to the general correlated case in Section~\ref{subsection:full_general_case}.

\subsection{No Covariate Case}
\label{subsection:nocov}
If there is no observed $X$ (equivalently, if all $\beta_{x,t}=0$), then the bias expressions of matching only on $Y_{0}$ prior to a DiD analysis simplify.

\begin{corollary}[Bias with no covariates]
Assume the same conditions in Theorem~\ref{theorem:mainresults}, but with no effect from $X$, i.e.,  $\beta_{x, 0} = \beta_{x, 1} = \rho = 0$. The bias expressions in Theorem~\ref{theorem:mainresults} then reduce to the following:
\begin{align*}
E\big[\hat{\tau}_{DiM}\big] - \tau &= \beta_{\theta, 1}\delta_{\theta} ,\\ 
E\big[\hat{\tau}_{DiD}\big] - \tau &= \Delta_{\theta} \delta_{\theta},\\
E\big[\hat{\tau}_{DiD}^{Y_{0}}\big]  - \tau&=\beta_{\theta, 1}\delta_{\theta}(1-r_{\theta}) = \beta_{\theta, 1}\delta_{\theta}\bigg[1-\frac{ \beta_{\theta, 0}^2 \sigma_{\theta}^2 }{\beta_{\theta, 0}^2 \sigma_{\theta}^2  + \sigma_{E}^2} \bigg],
\end{align*}
where $r_{\theta} = 1 - Var(Y|\theta) / Var(Y)$.
\label{corollary:nocov_results}
\end{corollary}
The na\"ive DiD, as explained in Section~\ref{subsection:PT}, has two interpretable terms contributing to the bias.
The first term,  $\delta_{\theta}$, is how imbalanced $\theta$ is.
The second term, $\Delta_{\theta}$, represents a potentially time varying relationship with the outcome that can break parallel trends. 
The matched estimator, if the reliability is greater than zero, uniformly dominates (has less bias than) the simple DiM estimator, but the comparison to DiD is less direct.

The key intuition of why matching on pre-treatment outcome before DiD could give more credible causal estimates than simple DiD is that matching should reduce imbalance between the treatment and control groups by reducing the differences driven by $\delta_{\theta}$.
Unfortunately, $\theta$ is not directly observed, but we hope that matching on the pre-treatment outcome will indirectly match on $\theta$.
The reliability term,  $r_{\theta}$, is a natural measure of how good the pre-treatment outcome proxies $\theta$.

To compare matching with DiD to DiD without matching, we rewrite the bias of $E\big[\hat{\tau}_{DiD}^{Y_{0}}\big]$ as
$$E\big[\hat{\tau}_{DiD}^{Y_{0}}\big]  - \tau =
   \beta_{\theta, 1}\delta_{\theta}(1-r_{\theta}) = 
   \Delta_{\theta} \delta_{\theta} - \underbrace{\Delta_{\theta} \delta_{\theta}r_{\theta}}_{\text{Reduction in Bias}} + \underbrace{ \beta_{\theta, 0}\delta_{\theta}(1-r_{\theta})}_{\text{Amplification in Bias}} . $$
This expression confirms that matching helps reduce the bias proportionally to how high the reliability $r_{\theta}$ is (note the bias offsetting term of $\Delta_{\theta} \delta_{\theta}r_{\theta}$).
It also reveals an additional bias term due to matching, $\beta_{\theta, 0}\delta_{\theta}(1-r_{\theta})$, that grows larger as reliability declines.

 Indeed, if $Y_{0}$ is completely dominated by noise, i.e., $r_{\theta} \approx 0$, implying that $Y_0 \approx \epsilon$ (assuming no observed covariates $X$), then the matching estimator achieves the maximum possible bias since one is matching on noise.
On the other hand, if the pre-treatment outcome perfectly proxies $\theta$, i.e., $r_{\theta} = 1$, implying that $Y_0 = \beta_{\theta} \theta$, then our bias is exactly zero as matching on $Y_0$ is equivalent to the ideal of matching directly on $\theta$.
Finally, if parallel trends perfectly held ($\Delta_\theta = 0$), then matching on the pre-treatment outcome would result in a non-zero bias of $\beta_{\theta, 0}\delta_{\theta}(1-r_{\theta})$; this is the regression to the mean phenomenon discussed in \citet{matchDiDSim}.

On one hand, matching on the pre-treament outcome reduces the bias compared to the na\"ive DiD by recovering the latent confounder $\theta$ proportional to the reliability  $r_{\theta}$.
On the other hand, matching also incurs additional bias by effectively removing the second ``difference'' in the DiD estimator in terms of $\theta$.
We next provide, for this simple no-covariate case, necessary and sufficient conditions to determine when it is better to match on the pre-treatment outcomes as compared to not matching at all. 

\begin{lemma}[Sufficient and necessary condition to match on pre-treatment outcome in the no-covariate case]
Under the assumptions of Corollary~\ref{corollary:nocov_results}, the absolute bias of matching on the pre-treatment outcomes is smaller or equal to the absolute bias of the na\"ive DiD, i.e., $|E\big[\hat{\tau}_{DiD}^{X, Y_{0}}\big] - \tau| \leq |E\big[\hat{\tau}_{DiD}\big] - \tau|$, if and only if
\begin{equation}
   r_{\theta} \geq 1 - \left|1 - s\right| \mbox{ with } s = \frac{\beta_{\theta, 0}}{\beta_{\theta, 1}}.
\label{eq:suffcond}
\end{equation}
\label{lemma:suffcond_match}
\end{lemma}
The proof comes from an algebraic simplification of the main result.
Intuitively, we should expect the trade off to depend on the reliability vs. how well parallel trends is satisfied.
The ratio of the pre-treatment outcome and post-treatment outcome slopes, $s$, is a scale invariant measure of how non-parallel the outcome-covariate relationship is before and after treatment.
When parallel trends hold, with $s=1$, the condition reduces to $r_{\theta} \geq 1$, showing that matching on the pre-treatment outcome is never recommended in this case. 

Parallel trends can break in two different ways.
The first is when the post-treatment slope is bigger than the pre-treatment slope $\beta_{\theta,1} > \beta_{\theta, 0}$, and they have the same sign, i.e., $0 < s < 1$.
In this case, the absolute value in Equation~\eqref{eq:suffcond} is irrelevant and the condition reduces to $r_{\theta} > s$, directly illustrating the trade off between reliability and the breakage in parallel trends.
For example, if the breakage in parallel trends is very minimal, i.e., $s = 0.99$ (close to 1) then the reliability has to be at least 99\% for the matching estimator to have lower bias than the na\"ive DiD.
In the second case when $s > 1$, there is a similar interpretation.
For example, if $s = 1.01$, we need our reliability to be at least 99\% to have a lower bias. 
Lastly, if the pre and post slopes of $\theta$ have differing signs, i.e., $s < 0$, or when $s > 2$, then Equation~\eqref{eq:suffcond} always holds because the right hand side is less than zero and reliability is always greater than zero.

We summarize the above relationship in Figure~\ref{fig:reliability_PT}, which shows the two-dimensional regions of $(s, r_{\theta})$ such that Lemma~\ref{lemma:suffcond_match} holds. 
As expected, when parallel trends is almost satisfied (near one), one would generally not want to match (the lighter region) unless reliably is very high.
On the other hand, when parallel trends is far from perfect (far from one), we benefit from matching more often, even with modest reliability.
\begin{figure}
\begin{center}
\includegraphics[width = 7cm, height = 5cm]{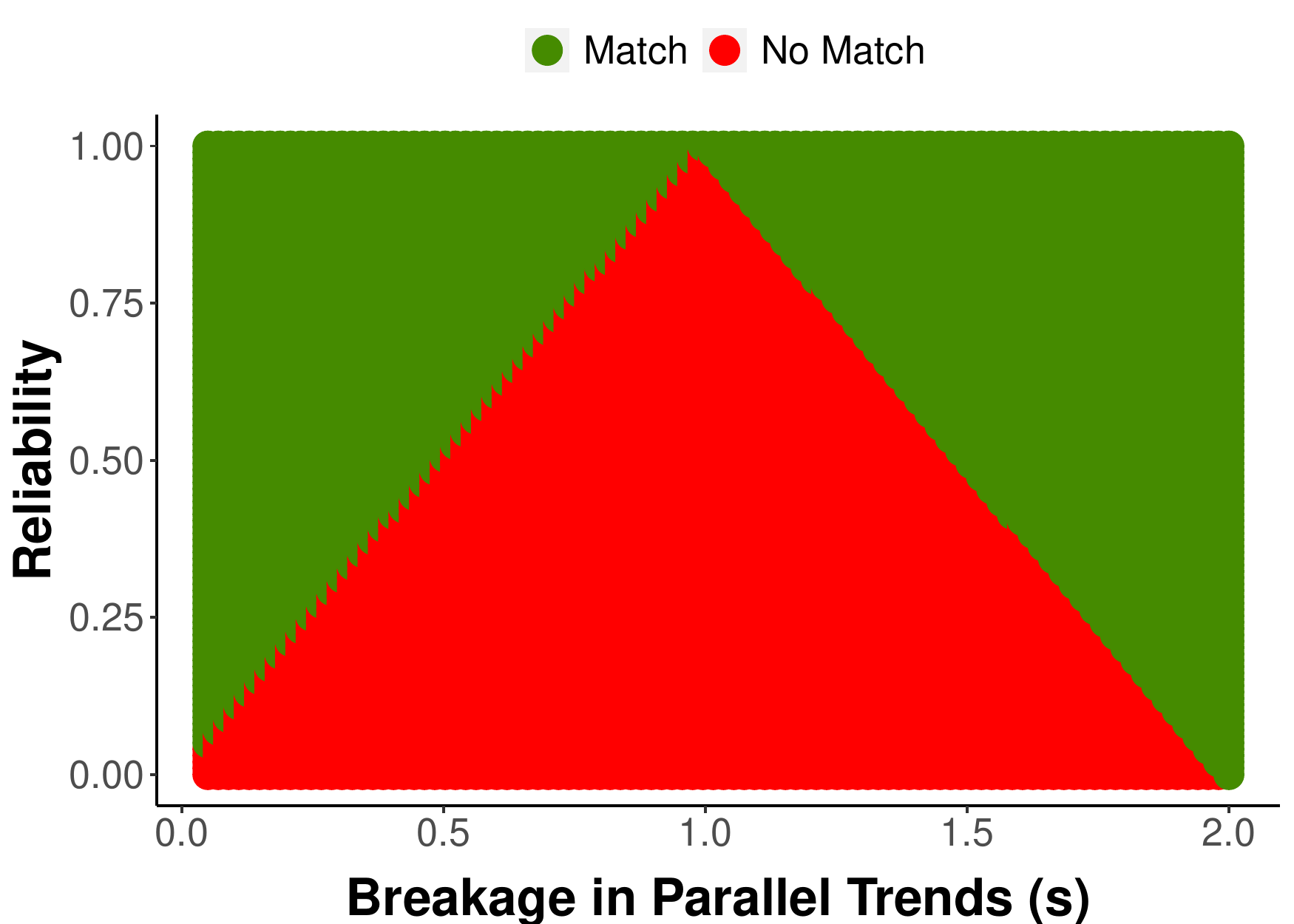}
\caption{Values of the breakage in parallel trends ($s$, on $x$-axis) and reliability ($r_{\theta}$, on $y$-axis) such that matching on pre-treatment outcomes before DiD leads to less bias (darker areas) or not (lighter areas).  See Lemma~\ref{lemma:suffcond_match}.}
\label{fig:reliability_PT}
\end{center}
\end{figure}

\subsection{Uncorrelated Case with Covariate}
\label{subsection:generalcase_uncorrelated}
We next consider the more general situation when the observed covariates do affect the outcome, i.e., $\beta_{x, 0} \neq 0$ and $\beta_{x, 1} \neq 0$, but the covariate is uncorrelated with the latent confounder ($\rho = 0$).  
This model is quite general as we could define the latent confounder as only those unexplained differences beyond what can be explained by $X$.

\begin{corollary}[Bias with zero correlation]
Assume the same conditions in Theorem~\ref{theorem:mainresults} but the observed covariate is uncorrelated with $\theta$ within the treatment groups, i.e., $\rho = 0$.
The bias expressions in Theorem~\ref{theorem:mainresults} then reduce to the following:
\begin{align*}
E\big[\hat{\tau}_{DiM}\big]  - \tau&= \beta_{\theta, 1}\delta_{\theta} + \beta_{x, 1}\delta_x ,\\
E\big[\hat{\tau}_{DiD}\big] - \tau &= \Delta_{\theta}\delta_{\theta} + \Delta_x\delta_x, \\
E\big[\hat{\tau}_{DiD}^X\big]  - \tau&=  \Delta_{\theta}\delta_{\theta}, \\
E\big[\hat{\tau}_{DiD}^{X, Y_{0}}\big]  - \tau&=  \beta_{\theta, 1}\delta_{\theta}(1-r_{\theta \mid x}) = \beta_{\theta, 1}\delta_{\theta}\bigg[1-\frac{ \beta_{\theta, 0}^2 \sigma_{\theta}^2 }{\beta_{\theta, 0}^2 \sigma_{\theta}^2  + \sigma_{E}^2} \bigg].
\end{align*}
\label{corollary:zerocorrelation}
\end{corollary}

Unlike matching on $Y_0$, matching on $X$ directly gets rid of the bias terms related to $X$ ($\Delta_x \delta_x$). Although this is generally an improvement from the na\"ive DiD, the na\"ive DiD bias could still be lower than the bias from matching on $X$ if preexisting bias from $X$ and $\theta$ were cancelling each other.
For example, if the time varying effects of $\theta$ and $X$ are the same ($\Delta_x = \Delta_{\theta})$ but the imbalance of $\theta$ and $X$ were in opposite direction ($\delta_x = -\delta_{\theta})$, then the bias contribution of $X$ and $\theta$ would perfectly cancel without matching, and  matching on $X$ would result in non-zero bias.
Lastly, the bias of matching on both $X$ and pre-treatment outcome reduces to the same bias expression when matching only on pre-treatment outcome without covariates in Section~\ref{subsection:nocov}. This is not a coincidence: matching on $X$ removes any effect $X$ had on the response and the situation simplifies to the no covariate case.

\subsection{General Correlated Case}
\label{subsection:full_general_case}
Given the results explained in Section~\ref{subsection:nocov}-~\ref{subsection:generalcase_uncorrelated}, we return to understanding the results presented in Theorem~\ref{theorem:mainresults}, where there exists a correlation between $\theta$ and $X$ ($\rho \neq 0$).
Theorem~\ref{theorem:mainresults} has two additional complications as compared to  Corollary~\ref{corollary:zerocorrelation}.
The first is regarding the reliability $r_{\theta \mid x}$, i.e., the proportion of variance of the pre-treatment outcome explained by $\theta$ after accounting for $X$.
Because accounting for $X$ via conditioning gives information about $\theta$ due to the non-zero correlation, it consequently reduces the total conditional variance of the pre-treatment outcome, $Var(Y_{i, 0} \mid X_i, Z_i = 0)$.
The consequence of this ``extra'' information is captured by the $(1 - \rho^2)$ term in the expression for $r_{\theta \mid x}$.
This means that as $X$ is increasingly correlated with $\theta$, the conditional reliability decreases, suggesting an increase in the cost of matching additionally on lagged outcomes.

The second complication is a countervailing force where the more correlated $X$ and $\theta$ are, the smaller (in general) the post-matched imbalance in $\theta$ because as we match on $X$ we are also likely bringing $\delta_\theta$ towards zero.
This is not guaranteed, however, as we can illustrate with the maximal $\rho = 1$ case.
As a reminder, $\rho = 1$ does not imply that $\theta$ and $X$ are perfectly correlated but instead implies that $\theta$ and $X$ are perfectly correlated \textit{within} each of the treatment and control groups.
Therefore, even if $\rho = 1$, matching on the same value of $X$ will not necessarily lead to the same value of $\theta$ for both the treatment and control group if the group averages of $\theta$ differ.
Consequently, the usefulness of the extra information we gain from matching on a correlated $X$ depends not only on the strength of the correlation $\rho$ but also on how much the imbalance of $X$ ($\delta_x$) is similar to the imbalance of $\theta$ ($\delta_{\theta}$) (after scale adjustment with a $\frac{\sigma_{\theta}}{\sigma_x}$ term to put $X$ on the scale of $\theta$).
In other words, if the imbalances $\delta_\theta$ and $\delta_X$ are the same (once $\delta_X$ is rescaled by $\frac{\sigma_{\theta}}{\sigma_x}$) then the bias from matching on a perfectly correlated $X$ would be zero, as stated by Theorem~\ref{theorem:mainresults}.

% \sout{To understand how matching on $X$ exactly impacts $\theta$, we show in Appendix~\ref{appendix:proof_mainresult} that the expected value of $\theta_i$ after matching a control unit $i$ to a treated unit $X$ is
% \begin{equation}
% E_{x \mid Z_i = 1}[E(\theta_i \mid Z_i = 0, X_i = x)] = \mu_{\theta, 0} + \rho \frac{\sigma_{\theta}}{\sigma_x} \delta_x
% \label{eq:expected_theta}
% \end{equation}
% We would hope this expectation to be equal to the mean of $\theta$ in the treated group, i.e., $\mu_{\theta, 1}$. 
% Unfortunately, even $\rho = 1$ (matching on a perfectly correlated covariate $X$) is not a sufficient condition for Equation~\eqref{eq:expected_theta} to equal $\mu_{\theta, 1}$. Although a perfect correlation in theory allows one to recover $\theta$ (up to a constant shift) after conditioning on $X$ \textit{within} the control and treatment group, this is not what matching is doing. Matching systematically chooses units with similar values \textit{between} control and treatment group, i.e., matching chooses a control unit such that it is equal to a treated unit's matched variable value. Therefore, the matching step impacts the expected value of $\theta$ proportional to the imbalance of $X$ ($\delta_x$) as evidenced by Equation~\eqref{eq:expected_theta}.}

Overall, the correlation term can decrease the conditional reliability, and thus increase the bias as discussed initially, but it also can allow matching on $X$ to indirectly match on $\theta$.
We generally find that the more correlated $X$ and $\theta$, the less benefit of matching additionally on pre-treatment outcomes.
In the extreme, Theorem~\ref{theorem:mainresults} shows that the bias will always increase if $\rho \delta_x$ is in the opposite direction of $\delta_{\theta}$, i.e., $sign(\rho \delta_x) \neq sign(\delta_{\theta})$, where $sign(.)$ is the sign function.
In such a case, the correlation actually harms our matching DiD estimator and fails to recover $\theta$ by pushing it further (in the opposite direction) from the desired value.
%Although this is possible, we believe it is unlikely in practice for two confounders' imbalance, $(\delta_{\theta}, \delta_x)$, to be in the opposite directions than the direction of their correlation, although this would potentially depend on specific context.
Additionally, even if the confounding effects were in the same direction as the correlation, if $\delta_x \gg \delta_{\theta}$ then $X$ can ``over-correct'' the confounding effect of $\theta$, leading to possibly greater bias.
We formalize the necessary conditions that removes these ``edge cases'' in Section~\ref{subsection:rule_matching_X}, where we give a rule-of-thumb guide for determining whether to additionally match on lagged outcome.

\subsection{Takeaways}
\label{subsection:takeaways}
There is unfortunately no simple condition such as Lemma~\ref{lemma:suffcond_match} that guarantees that it is always better to match on $X$ or both $X$ and the pre-treatment outcome under a more complicated general setting.
We instead visually illustrate our findings by plotting the bias of the na\"ive DiD estimator and the two matching estimators across three panels in Figure~\ref{fig:bias_complications}.
To produce the plots, we initially fix $\beta_{\theta, 1} = \beta_{x, 1} = 1.5$, $\beta_{\theta, 0} = \beta_{x, 0} = \delta_{\theta} = \delta_x = \sigma_{\theta} = \sigma_x = 1$, $r_{\theta \mid x} = 0.5$, and $\rho = 0$.
We then vary one of the fixed parameters for each of the three panels in Figure~\ref{fig:bias_complications}. For example, in the left panel we vary $\delta_x$. Although the results of the plot are sensitive to the initial parameter values (which were chosen based off our application in Section~\ref{section:application}), we show the plots to visually summarize the main findings and show the complex ways the biases can differ.
\begin{figure}[t]
\begin{center}
\includegraphics[width=\textwidth]{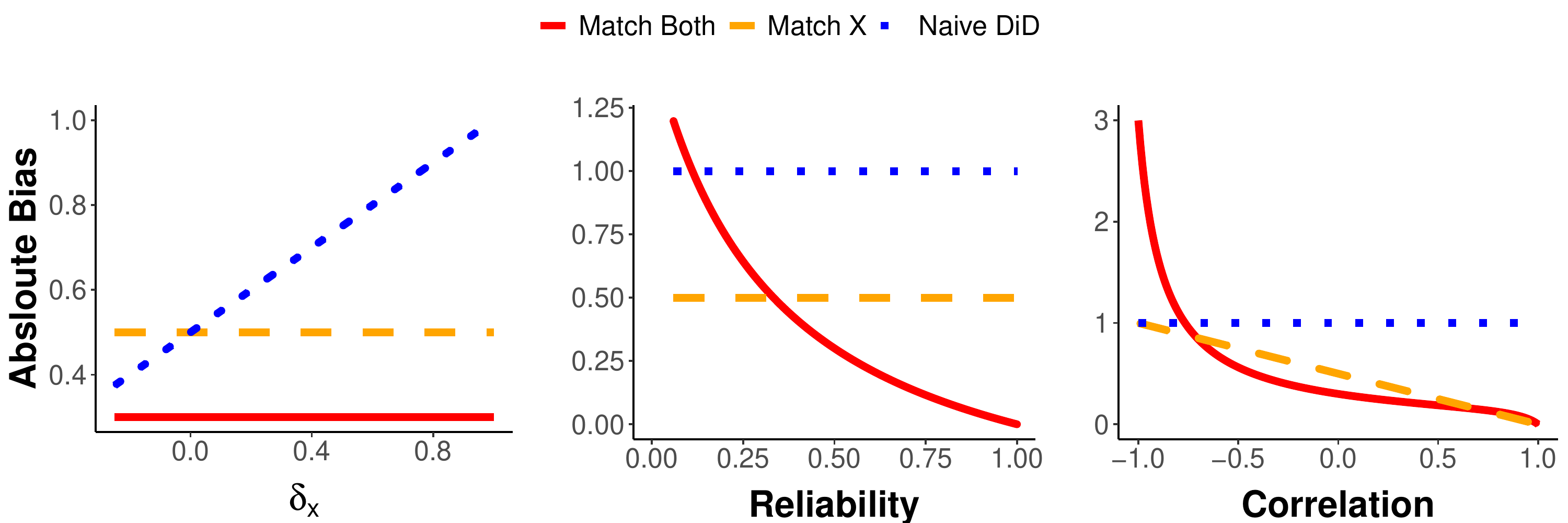}
\caption{Bias of the na\"ive DiD estimator (dotted line), matching on $X$ DiD estimator (dashed line), and matching on both $X$ and the pre-treatment outcome (solid line) according to the results in Theorem~\ref{theorem:mainresults}. We initially fix $\beta_{\theta, 1} = \beta_{x, 1} = 1.5$, $\beta_{\theta, 0} = \beta_{x, 0} = \delta_{\theta} = \delta_x = \sigma_{\theta} = \sigma_x = 1$, $r_{\theta \mid x} = 0.5$, and $\rho = 0$. We then vary one of the fixed parameters for each of the three panels in Figure~\ref{fig:bias_complications}.  }
\label{fig:bias_complications}
\end{center}
\end{figure}

The left panel in Figure~\ref{fig:bias_complications} shows that matching on $X$ generally reduces the bias relative to doing nothing since the bias of the na\"ive DiD estimator (dotted line) increases as the imbalance in $X$ increases, while both matching estimators are unaffected by this imbalance (in this uncorrelated case).
This confirms the guideline on why it is generally better to match on observed covariates.
However, as mentioned in Section~\ref{subsection:nocov}, matching on $X$ is not uniformly better than the na\"ive DiD: when $\delta_x$ is negative, the bias of $X$ cancels with the positive bias of $\theta$ to allow the na\"ive DiD (dotted line) to have a lower bias than that of the matching estimator that matches on $X$ only (dashed line). 

The second panel shows that as the reliability grows, it is better to additionally match on pre-treatment outcome as illustrated by the solid bias curve decreasing to zero.

The third panel shows how matching on X improves as the correlation of $X$ and $\theta$ increases.
As expected, the bias of both matching estimators is exactly zero when the correlation is one since in this scenario the imbalance of $\theta, X$ are the same, allowing matching on $X$ to perfectly recover $\theta$.
When the correlation is negative, however, the bias increases for both matching estimators since matching on $X$ actually pushes $\theta$ in the opposite direction, as explained in Section~\ref{subsection:full_general_case}.

Overall, given the above, we find that, regardless of whether parallel trends hold or not, matching on observable covariates $X$ is generally always advised for the following reasons:
\begin{enumerate}
    \item Matching on $X$ directly gets rid of all biases and confounding that results directly from $X$ (i.e., it removes the $\Delta_x\delta_x$ term).
    \item Matching on $X$ may further help recover the unobserved confounder $\theta$ when they are correlated. This shows up in the bias by offsetting the effect of $\delta_{\theta}$ by $\rho \frac{\sigma_{\theta}}{\sigma_x}\delta_x$.
\end{enumerate}
That being said, matching on $X$ does not always reduce the bias since pre-existing bias from $\theta$ and $X$ may have been cancelling each other out.
Furthermore, a non-zero correlation between $\theta$ and $X$ does not always help reduce the bias  since the correlation $\rho$ and the imbalances, $\delta_x, \delta_{\theta}$, may act in opposite directions or matching on $X$ may ``over-correct'' the bias.

Matching on the pre-treatment outcome $Y_{0}$ has a more complex trade-off that depends on the following two quantities:
\begin{enumerate}
    \item The reliability of the pre-treatment outcome for latent variation beyond what is explained by $X$. If the conditional reliability is high it will generally favor matching on $Y_{0}$ since it helps us recover the latent confounder $\theta$.
    \item The amount of breakage in parallel trends of the unobserved confounder $\theta$. Because matching on the pre-treatment outcome erodes the second ``difference'' in the DiD estimator, matching on the pre-treatment outcome is only favored if parallel trends does not hold. 
\end{enumerate}
This means that the decision to additionally match on pre-treatment $Y_{0}$ depends on the size of these quantities.
Given these tensions, we provide some guidance on how to decide when to match on what variables in Section~\ref{section:sensitivity_analysis}, after extending the above to multiple pre-treatment time periods.

\section{Multiple Time Periods and Multivariate Confounders}
\label{section:generalization}
So far we have focused on the setting when there is only one pre-treatment time point and univariate confounders $X, \theta$.
It is common, however, to have outcome measures for more than one pre-treatment time point and multiple covariates.
 For example, in our motivating empirical examples introduced in Section~\ref{section:empirical_example}, \citeauthor{principal_turnover} matched on seven observed covariates and six pre-treatment outcomes.
In cases such as these, we may want to match on all available pre-treatment outcomes and multiple covariates prior to the DiD analysis. 

We generalize our linear structural equation model introduced in Section~\ref{subsection:LSEM} to account for $T$ pre-treatment outcomes and multivariate confounders with the following model,

\begin{equation}
\begin{aligned}
Y_{i, t}(0) &= \beta_{0,t} + \vec{\beta}_{\theta, t}^\top \boldsymbol{\theta}_i + \vec{\beta}_{x, t}^\top \mathbf{X}_i + \epsilon_{i ,t}, \\
  Y_{i,t}(1) &= Y_{i,t}(0) +  \tau_{i} \mathbf{1}(t = T),\\
   Y_{i,t} &= Z_i Y_{i,t}(1)  + (1- Z_i) Y_{i,t}(0) ,
\end{aligned}
\label{eq:multivariate_X_multiple_T} 
\end{equation}
We index $0 \leq t \leq T - 1$ for the $T$ periods prior to post-treatment period $t = T$. Additionally, $\vec{\beta}_{x,t}$ is a $p$-dimensional column of slopes at time $t$ for $\mathbf{X}_i = (X_{i ,1}, \dots, X_{i, p})$, where $X_{i, j}$ represents individual $i$'s value of covariate $j$ for up to $p$ observed covariates. We define $\vec{\beta}_{\theta, t}, \boldsymbol{\theta}_i$ similarly for a $q$-dimensional number of latent confounders. We keep the data generating model for our confounders similar to Equation~\eqref{eq:simplethetax} but allow arbitrary covariances across the multivariate confounders,
\begin{equation}
\begin{pmatrix} \boldsymbol{\theta}_i \\ \mathbf{X}_i \end{pmatrix} \mid Z_i = z \sim N\left(\begin{pmatrix} \vec{\mu}_{\theta, z} \\ \vec{\mu}_{x, z}   \end{pmatrix} ,  \begin{pmatrix} \Sigma_{\theta \theta} & \Sigma_{\theta X} \\ \Sigma_{X \theta} & \Sigma_{XX}    \end{pmatrix}\right) \quad z = 0, 1  ,
\label{eq:generalthetax}
\end{equation}
where $\Sigma_{XX} \in \mathbb{R}^{p \times p}$ and $\Sigma_{\theta \theta } \in \mathbb{R}^{q \times q}$ are the covariance matrix for $\mathbf{X}_i$ and $\boldsymbol{\theta}_i$ within each treatment and control group respectively, and  $\Sigma_{\theta X} \in \mathbb{R}^{q \times p}$ is the covariance matrix between $\boldsymbol{\theta}_i$ and $\mathbf{X}_i$. 
As before, each individual's $(\boldsymbol{\theta}_i, \mathbf{X}_i)$ is drawn independently and identically according to Equation~\eqref{eq:generalthetax}, where $\vec{\mu}_{\theta,z}$ is a $q$-dimensional column vector for the average of $\boldsymbol{\theta}$ in treatment group $z$ and $\vec{\mu}_{x,z}$ is defined similarly. Lastly, we denote $\mathbf{Y}_i^{T} = (Y_{i, 0}, \dots, Y_{i, T-1})$ to represent the vector of $T$ pre-treatment period outcomes. 

It is common to use a linear regression framework to estimate $\hat\tau$ when there are multiple time periods \citep{DiD_multiple_timeperiods1, DiD_multiple_timeperiods3, DiD_multiple_applied1, DiD_multiple_applied2}.
In particular, the following two-way fixed effects linear regression is especially prevalent:
\begin{equation}
\label{eq:gDiD_regression}
    Y_{i, t} = \alpha_t + \gamma_i + \beta_Z Z_i + \beta_W W_t + \beta_{ZW} Z_i W_t + \epsilon_{i, t},
\end{equation}
where $W_t$ is an indicator that is 1 if in the post-treatment period ($t = T$) and zero otherwise, $\alpha_t, \gamma_i$ are fixed effects for time and unit, respectively, and $\epsilon_{i, t}$ is the error.
The estimate $\hat \beta_{ZW}$ for the interaction term would then be taken as the DiD estimate of the ATT.

It is well known that in a balanced panel data with multiple pre-treatment outcomes, $\hat\beta_{ZW}$ is equivalent to the estimate one would obtain using a classical two-period DiD using the \textit{average} of all the $T$ pre-treatment outcomes, $\overline{Y}_{i, 0:(T-1)} = \sum_{t = 0}^{T - 1} Y_{i, t}/T$, for the pre-treatment measure \citep{DiD_multiple_timeperiods1, DiD_multiple_timeperiods3}.
We will use this balanced property to derive the costs and benefits of matching before estimating impacts with Model~\ref{eq:gDiD_regression}.

Given our model and the balanced panel setting, the expected values of the generalized DiD estimators we would obtain, if using the above two-way fixed effects model for estimation, is then: 
\begin{equation*}
\begin{aligned}
 E\big[\hat{\tau}_{gDiD}\big] &= \{E(Y_{i, T}\mid Z_{i} = 1) -  E(Y_{i, T}\mid Z_{i} = 0)\} \\
     & \qquad  - \{ E(\overline{Y}_{i, 0:(T-1)} \mid Z_{i} = 1) -  E(\overline{Y}_{i, 0:(T-1)} \mid Z_{i} = 0)\} ,\\
E\big[\hat{\tau}_{gDiD}^{\mathbf{X}}\big] &= E(Y_{i, T}\mid Z_{i} = 1) -  E(\overline{Y}_{i, 0:(T-1)} \mid Z_{i} = 1) \\
    & \qquad - (E_{\mathbf{x} \mid Z_i = 1}[E(Y_{i, T}\mid Z_{i} = 0, \mathbf{X}_{i} = \mathbf{x})] -  E_{\mathbf{x} \mid Z_i = 1}[E( \overline{Y}_{i, 0:(T-1)} \mid Z_{i} = 0, \mathbf{X}_{i} = \mathbf{x})]) ,\\
E\big[\hat{\tau}_{gDiD}^{\mathbf{X}, \mathbf{Y}^{T}}\big] &= E(Y_{i, T}\mid Z_{i} = 1) -  E_{(\mathbf{x}, \mathbf{y^T}) \mid Z_i = 1}[E(Y_{i, T}\mid Z_{i} = 0, \mathbf{X}_{i} = \mathbf{x}, \mathbf{Y_i^T} = \mathbf{y^t} )] ,
\end{aligned}
\end{equation*}
where $\mathbf{Y_i^T} = (Y_{i, 0}, Y_{i, 1}, \dots, Y_{i, T-1})$ represents the collection of all pre-treatment outcomes.
Under our model, 

$$\overline{Y}_{i, 0:(T-1)} = \frac{\sum_{t = 0}^{T - 1} \beta_{0, t}}{T} + \frac{\sum_{t = 0}^{T - 1} \beta_{\theta, t}^\top \boldsymbol{\theta_i} }{T} + \frac{\sum_{t = 0}^{T - 1} \beta_{x, t}^\top \mathbf{X_i} }{T} + \frac{\sum_{t = 0}^{T - 1} \epsilon_{i, t} }{T} . $$
The biases of our estimators are then as stated in the following theorem:

\begin{theorem}[Bias under multiple time periods]
If $(Z_i, \mathbf{X}_i, \boldsymbol{\theta}_i, Y_{i,t})$ are independently and identically drawn from the data generating process in Equations~\ref{eq:multivariate_X_multiple_T} and \ref{eq:generalthetax}, then the bias of our estimators are the following:
\begin{align*}
E\big[\hat{\tau}_{gDiD}\big] - \tau &= \vec{\Delta}_{\theta}^\top \vec{\delta}_{\theta} + \vec{\Delta}_x^\top\vec{\delta}_x ,\\
E\big[\hat{\tau}_{gDiD}^{\mathbf{X}}\big]  - \tau&=   \vec{\Delta}_{\theta}^\top [\vec{\delta}_{\theta} -  \Sigma_{\theta X} \Sigma_{XX}^{-1} \vec{\delta}_x] ,\\
E\big[\hat{\tau}_{gDiD}^{\mathbf{X}, \mathbf{Y}^{T}}\big] - \tau&= \vec{\beta}_{\theta, T}^\top \Bigg[\vec{\delta}_{\theta} - \begin{pmatrix}
\Sigma_{\theta X} & \Sigma_{\theta Y_{T}}
\end{pmatrix} \begin{pmatrix}
\Sigma_{XX} & \Sigma_{X Y_{T}} \\ \Sigma_{Y_{T} X} & \Sigma_{Y_{T} Y_{T}}
\end{pmatrix}^{-1} \begin{pmatrix}
\vec{\delta}_x \\ \vec{\beta}_{\theta, 0 }^\top \vec{\delta}_{\theta} + \vec{\beta}_{x, 0}^\top \vec{\delta}_x \\ \vdots \\ \vec{\beta}_{\theta, T-1}^\top \vec{\delta}_{\theta} + \vec{\beta}_{x, T-1}^\top \vec{\delta}_x
\end{pmatrix} \Bigg] ,
\end{align*}
with vector analogs to baseline imbalance and breakage of parallel trends of
\begin{align*}
\vec{\delta}_{\theta} & := \vec{\mu}_{\theta, 1} - \vec{\mu}_{\theta, 0}  \mbox{ and, }\\
\vec{\Delta}_{\theta} & := \vec{\beta}_{\theta, T} - \frac{1}{T}\sum_{t = 0}^{T - 1} \beta_{\theta, t}  ,
\end{align*}
with $\vec{\delta}_x, \vec{\Delta}_x$ defined similarly, and $\Sigma_{Y_{T} Y_{T}} \in \mathbb{R}^{T \times T}$ the covariance matrix of all pre-treatment periods $Y_{0}, \dots, Y_{T- 1}$, with $\Sigma_{\theta Y_{T}} \in \mathbb{R}^{q \times T}, \Sigma_{X Y_{T}} \in \mathbb{R}^{p \times T}$ defined similarly. 
\label{theorem:most_general}
\end{theorem}
\noindent The proof is provided in Appendix~\ref{Appendix:generalization_proofs}. The interpretation for the na\"ive DiD and the DiD estimator that matches only on $\mathbf{X}$ remains similar to that already discussed in Section~\ref{section:main_results}. For example, the na\"ive DiD is biased again by the amount the multivariate confounders $\boldsymbol{\theta}, \mathbf{X}$ are breaking the parallel trends. Matching on $\mathbf{X}$ similarly reduces bias contribution from $\mathbf{X}$ and offsets the confounding effect of the latent confounder via the correlation in $\Sigma_{\theta X}$ and the imbalance $\vec{\delta}_x$. However, it is unclear from the expression above how matching on the $T$ additional pre-treatment outcomes reduces/harms the bias. To explore this further, we apply Theorem~\ref{theorem:most_general} to a simpler case that allows us to simplify the bias expression.

\subsection{No Covariate Case}
\label{subsection:multiple_time_nocov}
To understand how matching on multiple pre-treatment outcome may impact bias, we first assume a univariate unobserved confounder $\theta$, i.e., we let $q = 1$ and, similar to Section~\ref{subsection:nocov}, assume no observed covariates.
Unless one has strong subject-matter knowledge regarding missing confounders, one loses little generality by representing all latent confounding with a single confounder $\theta$ with arbitrary breakage in parallel trends ($\Delta_{\theta}$) and imbalance ($\delta_{\theta}$).

Since  we have no observed covariates $\mathbf{X}$, we state the simplified bias result only for the matching estimator that matches additionally on the $T$ pre-treatment outcome. 

\begin{theorem}[Bias]
If $(Z_i, \theta_i, Y_{i,t})$ are independently and identically drawn from the data generating process in Equations~\ref{eq:multivariate_X_multiple_T} and \ref{eq:generalthetax}  with $q = 1$ and $p = 0$ (no observed covariates), then the bias of matching on $T$ pre-period outcomes is then
\begin{align*}
E\big[\hat{\tau}_{gDiD}^{\mathbf{X}, \mathbf{Y}^{T}}\big]  - \tau &= \beta_{\theta, T} \delta_{\theta} (1 - r_{\theta}^T),
\end{align*}
where 
\begin{align*}
    r_{\theta}^T = \frac{T \bar{\beta}_{\theta, \text{pre}}^2 \sigma_{\theta}^2}{T \bar{\beta}_{\theta, \text{pre}}^2 \sigma_{\theta}^2 + \sigma_{E}^2} , \quad  \bar{\beta}_{\theta, \text{pre}}^2 := \frac{1}{T} \sum_{t = 0}^{T-1} \beta_{\theta, t}^2.
\end{align*}
\label{theorem:bias_multiple_T_noX}
\end{theorem}

With a slight abuse of notation, we define $\bar{\beta}_{\theta, \text{pre}}^2$ as the average of the squared coefficients as opposed to the square of the average. The proof is provided in Appendix~\ref{Appendix:generalization_proofs}. The result is similar to Corollary~\ref{corollary:zerocorrelation} (and exactly equivalent when $T = 1$), except our new ``reliability'' term, $r_{\theta}^T$, increases as a function of $T$, assuming the average $\bar{\beta}_{\theta,\text{pre}}$ does not shrink as $T$ grows.
Therefore, the more pre-treatment period outcomes we match on, the more we decrease bias resulting from our latent confounder $\theta$

To illustrate how much our reliability increases, suppose $r_{\theta}^{T = 1} = 0.5$.
Then if $T = 2$, i.e., we have one more additional pre-treatment period, $r_{\theta}^{T = 2} \approx 0.67$ (assuming the average $\bar\beta^2_{\theta,pre}$ remains the same).
If $T = 3$, then $r_{\theta}^{T = 3} = 0.75$, which is a 50\% increase of reliability with only two additional time periods. 
In practice, when the effects of $\theta$ vary in the pre-treatment periods, using more pre-period measures will only help if the average $\beta_\theta$ coefficient does not shrink so much as to offset the gain from $T$ in the reliability expression.

Matching on multiple pre-treatment time points has close ties to synthetic controls, where one would construct a synthetic comparison unit as a weighted average of ``donor'' (control) units such that the synthetic unit closely matches the measured characteristics (in particular pre-treatment outcomes) of a target treated unit.
In particular, the regression to the mean phenomenon shown in Theorem~\ref{theorem:mainresults} and ~\ref{theorem:bias_multiple_T_noX} is also present in synthetic controls with multiple pre-period outcomes \citep{SC1, SC2, Synth_RTM}.
Our findings would suggest, then, that one should attend to reliability of the pre-treatment outcome as a measure of latent characteristics in the synthetic control context as well.
Further, it also shows that when the number of pre-period outcomes in synthetic controls are few, the reliability will be lower and thus the bias induced by the weighting of units larger.
Of course, as \citet{damour_overlap} shows, perfect matching becomes impractical when the number of periods grows due to the curse of dimensionality.
Further adjustment \citep{benm:fell:roth:19} may avoid some of these difficulties.
We leave the exploration of how to estimate reliability here, as well as this tension, to future work.

We can also extend Lemma~\ref{lemma:suffcond_match}, which provided the changeover point between matching and not matching in the no-covariate case, to this more general setting of our multiple time-point:
\begin{lemma}[Sufficient and necessary condition to match on $T$ pre-treatment outcomes]
Suppose a data matrix of $D_i = (Y_{i, 0}, \dots, Y_{i, T}, Z_i)$ for $i = 1,2, \dots, n$ follows the same setting as that listed in Theorem~\ref{theorem:bias_multiple_T_noX}. 
Then the absolute bias of matching on the $T$ pre-treatment outcomes is smaller or equal to the absolute bias of the na\"ive DiD, i.e., $|E\big[\hat{\tau}_{gDiD}^{\mathbf{X}, \mathbf{Y}^{T}}\big] - \tau| \leq |E\big[\hat{\tau}_{gDiD} - \tau|$, if and only if
\begin{equation}
r_{\theta}^T > 1 - \left|1 - \frac{ \bar{\beta}_{\theta, \text{pre}}}{\beta_{\theta, T}}\right|, \label{eq:guideline_two}	
\end{equation}
where $r_{\theta}^T$ is as defined in Theorem~\ref{theorem:bias_multiple_T_noX} and 
$$ \bar{\beta}_{\theta, \text{pre}} = \frac{1}{T}\sum_{t = 0}^{T-1} \beta_{\theta, t}.$$
In other words, when Equation~\eqref{eq:guideline_two} holds, matching on $\mathbf{Y^T}$ has less bias than not matching.
\label{lemma:general_suffcond_match}
\end{lemma}

The proof is in Appendix~\ref{appendix:proof_guidelineY}.

\subsection{Stable Pretreatment Case With Covariates}
\label{subsection:multiple_time_simple}

% To understand how matching on multiple pre-treatment outcome may impact bias, we first assume a univariate unobserved confounder $\theta$, i.e., we let $q = 1$. 
In Section~\ref{subsection:multiple_time_nocov}, we assumed away all observed covariates $\mathbf{X}$ to gain further intuition on the bias expression in the no-covariate case.
In this section, we keep univariate $\theta$ and bring back all observed covariates but make a different simplifying assumption in order to gain intuition on the bias when additionally matching on the observed covariates $\mathbf{X}$. In particular, we assume $\vec{\beta}_{x,t}$ and $\vec{\beta}_{\theta,t}$ are unchanging across the pre-treatment period outcomes, i.e., parallel trends hold in the pre-treatment periods, as given in the following assumption:
\begin{assumption}[Unconditional parallel trends for pre-treatment outcomes]
    For all pre-treatment periods, $0 \leq t, t' \leq T - 1$, assume there is no time varying effects of either confounders $\theta$ or $X$. i.e., $\beta_{\theta, t} = \beta_{\theta, t'}$ and $\vec{\beta}_{x,t} = \vec{\beta}_{x,t'}$ for all $t, t'$ such that  $0 \leq t, t' \leq T - 1$.
\label{assumption:multiple_time_periods}
\end{assumption}
This assumption effectively states that we have $T$ independent and identical measurements of pre-treatment outcome $Y$ available, i.e., for all the pre-treatment outcomes the parallel trends hold perfectly (unconditional on $\mathbf{X}$ or $\theta$). 

We emphasize that we use Assumption~\ref{assumption:multiple_time_periods} only in this section for pedagogical reasons to give more interpretable bias expression that include observed covariates $\mathbf{X}$.

\begin{theorem}[Bias with Multiple Time Periods under Stability Assumption]
If $(Z_i, \mathbf{X}_i, \boldsymbol{\theta}_i, Y_{i,t})$ are independently and identically drawn from the data generating process in Equations~\ref{eq:multivariate_X_multiple_T} and \ref{eq:generalthetax} and Assumption~\ref{assumption:multiple_time_periods} holds with $q = 1$, then the biases of our matching estimators are then
\begin{align*}
E\big[\hat{\tau}_{gDiD}^{\mathbf{X}}\big]  - \tau &=   \Delta_{\theta}\tilde{\delta}_{\theta}, \\
E\big[\hat{\tau}_{gDiD}^{\mathbf{X}, \mathbf{Y}^{T}}\big]  - \tau &= \beta_{\theta, T} \tilde{\delta}_{\theta} (1 - r_{\theta \mid x}^T),
\end{align*}
where 
\begin{align*}
    \tilde{\delta}_{\theta} &= \delta_{\theta} - \Sigma_{\theta X}\Sigma_{XX}^{-1} \vec{\delta}_{x}, \\
    r_{\theta \mid x}^T &= \frac{T\beta^2_{\theta, T - 1} \tilde{\sigma}_{\theta}^2}{T\beta^2_{\theta, T - 1} \tilde{\sigma}_{\theta}^2 + \sigma_{E}^2} ,\\
    \tilde{\sigma}_{\theta}^2  &= \sigma_{\theta}^2 - \Sigma_{\theta X}\Sigma_{XX}^{-1} \Sigma_{X \theta} .
\end{align*}
\label{theorem:bias_multiple_T}
\end{theorem}
The proof is provided in Appendix~\ref{Appendix:generalization_proofs}.
The result is similar to Theorem~\ref{theorem:bias_multiple_T_noX}. Our new ``reliability'' term, $r_{\theta \mid x}^T$, again increases as a function of $T$ similar to that in Theorem~\ref{theorem:bias_multiple_T_noX}.\footnote{Strictly speaking $r_{\theta \mid x}^T$ is not the reliability of our outcome as defined in Definition~\ref{def:reliability} However, it can still be roughly interpreted as the ratio of how much total variance from the $T$ pre-period outcome is explained by $\theta$.} One difference, however, is that the average of the pre-period slopes, $\bar{\beta}_{\theta, \text{pre}}$, simplifies to $\beta_{\theta, T-1}$ under Assumption~\ref{assumption:multiple_time_periods}. The second difference arises from matching on the observed covariates $\mathbf{X}$.
The tildes in the expressions in Theorem~\ref{theorem:bias_multiple_T} represent a ``residualization'' of the outcome and $\theta$ by $\mathbf{X}$, driven by the correlation of $\mathbf{X}$ and $\theta$.
Namely, $\delta_{\theta}$ and $\sigma_{\theta}^2$ are adjusted by the extra information gained from matching on correlated $\mathbf{X}$ (see Section~\ref{subsection:full_general_case} for detailed explanation).
In other words, the more tightly coupled $\mathbf{X}$ and $\theta$, the less the potential benefit for additionally matching on pre-treatment outcomes beyond just matching on $\mathbf{X}$, as represented by a smaller $\tilde{\sigma}_{\theta}^2$ and generally lower $r_{\theta \mid x}^T$.

Note that if $p = 1$, then $\Sigma_{\theta X} = \rho \sigma_{\theta} \sigma_x$, $\tilde{\delta}_{\theta} = \delta_{\theta} - \rho \frac{\sigma_{\theta}}{\sigma_x}\delta_x$, and $\tilde{\sigma}_{\theta}^2 = \sigma_{\theta}^2(1 - \rho^2)$, which recovers Theorem~\ref{theorem:mainresults}.

 The ``when to match'' Lemma~\ref{lemma:general_suffcond_match} directly extends here by replacing $r_{\theta}^T$ with $r_{\theta \mid x}^T$ and $\bar{\beta}_{\theta, \text{pre}}$ with $\beta_{\theta, T-1}$. The intuition is that by first taking out the predictive element of $\mathbf{X}$ both directly and through its correlation with $\theta$, we can reduce our covariate case to the no-covariate case, and then follow the ideas in Section~\ref{subsection:nocov}, with $\delta_{\theta}$ replaced with $\tilde{\delta}_{\theta}$, the imbalance of the ``residualized'' latent confounder. We later leverage this result to motivate our guideline for determining when to additionally match for the pre-period outcomes.

\section{Determining When to Match}
\label{section:sensitivity_analysis}

Using our model as a working approximation, we now use the theoretical results in Sections~\ref{section:main_results} and \ref{section:generalization} to provide heuristic guidance on what to match on, along with a means of roughly estimating the reduction (or increase) in bias due to matching.
We also provide a publicly available script to run our proposed guidelines.\footnote{See \url{https://github.com/daewoongham97/DiDMatching}.}
We assume a univariate $\theta$, and allow for arbitrary breakage in parallel trends and degree of imbalance.

\subsection{Guidance for Matching on Covariates} 
\label{subsection:rule_matching_X}

\begin{guideline}[Matching on $\mathbf{X}$ guideline]
Always match on $\mathbf{X}$.

One can estimate the reduction in bias (relative to the na\"ive DiD) from matching on $\mathbf{X}$ as:
$$\hat\Delta_{\tau_x} := \left|\hat{\vec{\delta}}_x^\top \hat{\vec{\Delta}}_x\right| = \left|(\hat{\vec{\delta}}_x)^\top \left(\hat{\vec{\beta}}_{x, T} - \frac{1}{T}\sum_{t = 0}^{T - 1} \hat{\vec{\beta}}_{x, t}\right)\right| ,$$
where $\hat{\vec{\delta}}_x$ is the difference in means of $\mathbf{X}$ between the treated and control group and $\hat{\vec{\beta}}_{x, t}$ are the $p$-dimensional regression coefficients from linear regressions of $Y_{i,t}$ on $\mathbf{X}_i$ within the control group, one regression for each time point.
\label{guideline:matchX}
\end{guideline}

This advice is consistent with the current advice on how practitioners should generally account for as many observed covariates as possible \citep{rose:02b, rose:rubi:83, shpi:vand:robi:10, Mbias}, which is often referred to as the ``pre-treatment criterion.''

$\hat\Delta_{\tau_x}$ is a rough estimate of the degree of bias reduction (relative to the na\"ive DiD). More formally $\hat\Delta_{\tau_x}$ estimates $\Delta_{\tau_x}$, where
$$\Delta_{\tau_x} :=  \left|E\big[\hat{\tau}_{gDiD}\big] - \tau \right| - \left|E\big[\hat{\tau}_{gDiD}^{\mathbf{X}}\big] - \tau \right|.$$

Section~\ref{subsection:generalcase_uncorrelated} shows that, when $\mathbf{X}$ and $\theta$ are uncorrelated, matching directly on the observed covariates $X$ reduces the bias contribution from $X$, i.e., gets rid of the $\Delta_x \delta_x$ term in the bias.
As long as biases are not cancelling, this term would then be the bias reduction.
It turns out that, even if $\Sigma_{\theta X} \neq 0$, na\"ively taking the difference of the \emph{estimated} slopes from a linear regression still accounts for how the correlation affects the bias, because the correlation with $\theta$ gets picked up by the estimated slope coefficients themselves being biased.
We formalize how this bias estimate works in the following theorem:

\begin{theorem}[Conditions for matching on $\mathbf{X}$ reducing bias] \ \\
Suppose a data matrix of $D_i = (\mathbf{X}_i, Y_{i, 0}, \dots, Y_{i, T}, Z_i)$ for $i = 1,2, \dots, n$ follows the linear structural equation model in Equations \ref{eq:multivariate_X_multiple_T} and \ref{eq:generalthetax}, and the following three sign conditions hold:
\begin{itemize}
    \item[] 1) $sign(\vec{\Delta}_x^\top \vec{\delta}_x) = sign(\Delta_{\theta} \delta_{\theta})$,
    \item[] 2) $sign(\Delta_{\theta} \Sigma_{\theta X} \Sigma_{XX}^{-1} \vec{\delta}_x) = sign(\Delta_{\theta} \delta_{\theta})$,
    \item[] 3) $sign( \delta_{\theta}) = sign( \delta_{\theta}- \Sigma_{\theta X} \Sigma_{XX}^{-1} \vec{\delta}_x)$ .
\end{itemize}
Then $\Delta_{\tau_x} \geq 0$, i.e., matching on $\mathbf{X}$ will have less bias than the  na\"ive DiD.
Additionally, $\hat\Delta_{\tau_x}$ is a consistent estimator for $\Delta_{\tau_x}$.
\label{theorem:matching_X}
\end{theorem}

The proof is provided in Appendix~\ref{appendix:proof_guidelineX}.
In general, while matching on $\mathbf{X}$ removes bias from the observed covariates, doing so does not guarantee overall bias reduction due to the ``edge cases'' mentioned in Section~\ref{subsection:takeaways}.
The three sign conditions in Theorem~\ref{theorem:matching_X} remove these edge cases.
The first sign condition says the pre-existing biases of $\theta$ and $\mathbf{X}$ are not in opposite directions.
The second sign condition ensures the imbalance of $ \delta_{\theta}$ is reduced, not increased, by the extra information we gain about $\theta$ by matching on a correlated $\mathbf{X}$.
In other words, the second sign condition does not allow the confounding effects of $\theta$ and $\mathbf{X}$ to go in the opposite direction of the correlation.
The third sign condition does not allow the additional reduction in bias gained from matching on a correlated $\mathbf{X}$ to over-correct the bias.
We leave the assessment of the plausability of these sign conditions to future empirical work. 
Importantly, these conditions are not necessary, in that there are many cases where they do not hold, but matching on $\mathbf{X}$ is still beneficial.

\subsection{Guidance for Matching on Pre-treatment Outcome(s)}
\label{subsection:rule_matching_Y}
While the guidance for matching on $\mathbf{X}$ is relatively straightforward, the same is not true for matching on the available pre-treatment outcome(s).
On one hand, matching on pre-treatment outcomes reduces the effect of imbalance on $\theta$ ($\delta_{\theta}$) by a factor proportional to the reliability.
On the other hand, matching on pre-treatment outcomes erodes the second ``difference'' in the DiD analysis, which can add bias to the overall estimate.
Since these quantities are consequences of how latent parameters relate to the outcome, a general data-driven guideline will have to also rely on some untestable assumptions and heuristics.
Here, we use our simplified model from above to provide such a guideline.

\begin{guideline}[Matching on $\mathbf{X}, \mathbf{Y^T}$ guideline] \ \\
Match on all available pre-treatment outcome(s) if the following inequality holds. 
\begin{equation}
 \hat r_{\theta}^T > 1 - \left|1 - \frac{ \hat{\bar{\beta}}_{\theta, \text{pre}}}{\hat\beta_{\theta, T}}\right|,
 \label{equation:guideline_Y_eq}
\end{equation} 
where the reliability and regression coefficients can be estimated as described below.

We can estimate the reduction in bias from matching additionally on $T$ pre-treatment outcomes (relative to matching only on $\mathbf{X}$) as
$$\hat\Delta_{\tau_{x,y}} := \left| |\hat\Delta_{\theta} \hat{\tilde{\delta}}_{\theta}| - |\hat \beta_{\theta ,T} \hat{\tilde{\delta}}_{\theta} (1 - \hat r_{\theta}^T) |\right|, $$
where $\hat\Delta_{\tau_{x,y}}$ is an estimate for $\Delta_{\tau_{x,y}}$, the additional bias reduction from matching on both observed covariates pre-treatment outcomes relative to matching on observed covariates only, formally defined as
\begin{equation}
\Delta_{\tau_{x,y}} := \left|E\big[\hat{\tau}_{gDiD}^{\mathbf{X}}\big] - \tau \right| - \left| E[\hat{\tau}_{gDiD}^{\mathbf{X}, \mathbf{Y^T}} ]  - \tau \right|.\label{eq:full_bias_eq}
\end{equation}

We describe how to estimate the components of the guideline below.
\label{guideline:matchboth}
\end{guideline}

 Guideline~\ref{guideline:matchboth} is motivated by Lemma~\ref{lemma:general_suffcond_match}, which provided conditions of when to match on the pre-period outcomes under the general multiple pre-period outcome case when there exist no observed covariates.
Theorem~\ref{theorem:bias_multiple_T}, under Assumption~\ref{assumption:multiple_time_periods}, shows that matching on observed covariates effectively residualizes out any effect on $\mathbf{X}$, reducing the covariate case to the no covariate case of Theorem~\ref{theorem:bias_multiple_T_noX}.
Therefore, we present Guideline~\ref{guideline:matchboth} as a general heuristic guideline that targets the main trade-off between reliability and the breakage in parallel trends.
We have shown that this residualization works when Assumption~\ref{assumption:multiple_time_periods} holds.
In Appendix~\ref{appendix:assump1} we show robustness results of Guideline~\ref{guideline:matchboth} to violations of Assumption~\ref{assumption:multiple_time_periods}, finding that even when Assumption~\ref{assumption:multiple_time_periods} does not hold, the guideline is correct except near the boundary of whether to match or not.

To estimate the quantities needed conduct the guideline check, we first residualize our outcomes for each time point (we require multiple pre-treatment observations, i.e., $T > 1$):
$$\tilde{Y}_{i,t} := Y_{i, t} - \hat{\vec{\beta}}_{x, t} \mathbf{X}_i , $$
where the $\hat{\vec{\beta}}_{x, t}$ are the regression coefficients for regressing the outcomes at timepoint $t$ onto the covariates.
In Appendix~\ref{appendix:proof_guidelineY} we show that 
\begin{equation}
\label{eq:tilde_Y}
\tilde{Y}_{i, t}  \xrightarrow{d}  \beta_{\theta, t}(\theta_i - \Sigma_{\theta X} \Sigma_{XX}^{-1} X_i) + \epsilon_{i,t} =  \beta_{\theta, t}\tilde{\theta}_i + N\left(0, \sigma_{E}^2\right),
\end{equation}
showing that this regression returns us to the no-covariate case of Section~\ref{subsection:nocov}, with new latent variable $\tilde{\theta}$ as shown in Theorem~\ref{theorem:bias_multiple_T}.

We then estimate the residual variance of the residualized model with
$$\hat{\sigma}_{E}^2  = \frac{1}{2} \widehat{ Var}\left(\tilde{Y}_{i, T - 1} - \tilde{Y}_{i, T - 2} \mid Z_i = 0 \right) .$$
If $\beta_{\theta, t} = \beta_{\theta, t'}$, then the variance of the differences allow for estimation of the residual variance. If these $\beta$ differ, this estimate will be biased upward, reducing the reliability term and thus shifting the recommendation towards not matching. Therefore, our matching guideline is conservative, i.e., when it tells to match it is correct despite the aforementioned stability assumption. On the other hand, if our guideline suggests to not match, it may still be beneficial to match if $\beta_{\theta, T-2} \neq \beta_{\theta, T-2}$.
Different time periods are possible here; we recommend the researcher select two periods that seem the most stable.
Other estimation approaches are possible here; see discussion and robustness results in Appendix~\ref{appendix:assump1}. 

We can use these residualized outcomes and estimated residual variance to obtain empirical estimates of the quantities given in Lemma~\ref{lemma:general_suffcond_match}'s Equation~\eqref{eq:guideline_two}: 
\begin{align*}
 \hat \beta_{\theta , t} &= \sqrt{ \widehat{Var}(\tilde{Y}_{i, t} \mid Z_i = 0) - \hat{\sigma}_{E}^2  }, \quad t = 0, 1, 2, \dots, T , \\
\hat r_{\theta}^T &= \frac{T \hat{\bar{\beta}}_{\theta, \text{pre}}^2 } {T\hat{\bar{\beta}}_{\theta, \text{pre}}^2+ \hat\sigma_E^2} \quad \mbox{ with } \quad
\hat{\bar{\beta}}_{\theta, \text{pre}}^2 = \frac{1}{T} \sum_{t = 0}^{T-1} \hat \beta_{\theta, t}^2.
\end{align*}

The additional quantities in the formula for bias reduction are estimated as 
\begin{align*}
\hat\Delta_{\theta} &= \hat \beta_{\theta, T}  - \hat{\bar{\beta}}_{\theta, \text{pre}} \quad \mbox{ with } \quad \hat{\bar{\beta}}_{\theta, \text{pre}} = \frac{1}{T} \sum_{t = 0}^{T-1} \hat \beta_{\theta, t}, \mbox{ and } \\
\hat{\tilde{\delta}}_{\theta} &= \frac{\hat{E}(\bar{\tilde{Y}}_{i, 0:(T-1)} \mid Z_i = 1) - \hat{E}(\bar{\tilde{Y}}_{i, 0:(T-1)} \mid Z_i = 0)}{\hat{\bar{\beta}}_{\theta, \text{pre}}}.
\end{align*}

The above estimation formulae may seem as if we are forgetting to account for $\tilde{\sigma}_{\theta}$. We show in Appendix~\ref{appendix:proof_guidelineY} that the final estimates do not require estimating $\tilde{\sigma}_{\theta}$ because $\tilde{\sigma}_{\theta}$ cancels out.
Equivalently, we can think of our latent, residualized $\tilde{\theta}$ as having unit variance.

The above guideline technically rely on a linear model and some degree of stability (here, between the final two pre-treatment periods).
Although these assumptions are unlikely to hold exactly in practice, they do provide a heuristic for deciding on the benefits of matching.
We leave the extension of our guideline when linearity and stable pre-treatment trends fail to future work. 
We also acknowledge that other approaches, e.g., ones that incorporate additional information or data, may lead to superior performance.
Alternative approaches for estimation are also possible (such as averaging coefficients over multiple time periods); we leave exploring the benefits and shortcomings of these to future work.
When we only have a single pre-period outcome ($T = 1$), we cannot estimate the reliability and thus implement the above guideline.
We can, however, instead directly assume different reliability values to calculate the guideline. See Appendix~\ref{appendix:guideline} for details of this sensitivity analysis approach.

We demonstrate estimating these quantities in our applied example in Section~\ref{section:application}.
We also formally summarize the assumptions that theoretically justify Guideline~\ref{guideline:matchboth} and show consistency in the above estimation procedure in the following theorem: 

\begin{theorem}[Conditions for consistently estimating Guideline~\ref{guideline:matchboth}]
Assume $T > 1$ and the general model of Equations~\ref{eq:multivariate_X_multiple_T} and \ref{eq:generalthetax} with $q=1$. 
Further suppose that we have stable effects of $\theta$ in the last two pre-treatment period, i.e., $\beta_{\theta, T-1} = \beta_{\theta, T-2}$. In this case, $\hat r_{\theta \mid x}^T$ is a consistent estimator for $r_{\theta \mid x}^T$, and $\frac{ \hat{\bar{\beta}}_{\theta, \text{pre}}}{\hat\beta_{\theta, T}}$ is a consistent estimator for $\frac{ \bar{\beta}_{\theta, \text{pre}}}{\beta_{\theta, T}}$.
\label{theorem:matching_Y}
\end{theorem}
See proof in Appendix~\ref{appendix:proof_guidelineY}.

To account for estimation error, we recommend using a case-wise bootstrap to assess the uncertainty of the guidelines with respect to measurement error and to obtain confidence intervals for the estimated parameters \citep{efron_boot}.
We illustrate this with our principal turnover example in Section~\ref{section:application}, below.
Our provided scripts implement this approach.

\section{Application - The Impacts of Principal Turnover}
\label{section:application}

As introduced in Section~\ref{section:empirical_example}, \citet{principal_turnover} are interested in determining whether the impact of principal turnover has any causal impact on student achievement a year after the principal has changed.
We follow the guidelines in Section~\ref{section:sensitivity_analysis} using the same data the authors used with a few adjustments.
We first include all seven standardized observed covariates as $\mathbf{X}$ (further details in Section~\ref{section:empirical_example}).\footnote{For time varying covariates we take the average of these covariates over the relevant years and treat them as time invariant covariates. We standardized with respect to the treatment group, i.e., we divide each $\mathbf{X}_i$ by the standard deviation of the treatment group's respective covariate.}
Although Figure~\ref{fig:math_results} shows the original matching DiD estimates by \citeauthor{principal_turnover}, we do not actually need to perform any matching to use our guidelines detailed above. 
% Nevertheless, we also empirically match to show how the treatment effects may vary for illustrative purposes (further detailed below).

This application has staggered adoption, in that we have a set of years, and in a given year some schools are treated (lost a principal) and other schools are not.
We, therefore, apply our guidelines to each year in turn, and then average the results, weighting by the number of treated units in a given year.

We also look at the individual year recommendations to assess the stability of the guidelines across time.
Our provided script automates both this aggregation and per-year analysis.

To quantify the reduction in bias from matching on $\mathbf{X}$ only, we estimate the bias contribution by following Guideline~\ref{guideline:matchX}.
More specifically, we estimate $\vec{\Delta}_x$ by obtaining the post-treatment and pre-treatment slopes of $\mathbf{X}$ through seven separate linear regression of $Y_{i, t}$ for $t = 0, 1, \dots, 6$ on $\mathbf{X}_i$ in the control group.
We also estimate $\vec{\delta}_x$ by taking the simple difference in means of $\mathbf{X}$ from the treatment and control group across the seven covariates.
We repeat this for years 2005 through 2016, and then weight across all years, weighting by number of treatment units in each year.

Final estimates are shown in the top left entry of Table~\ref{tab:application_results}. 
The estimates correspond (left to right) to school enrollment size, proportion of students with free lunch, proportion of Hispanic students, proportion of Black students, proportion of new-to-school teachers, principal experience, and average number of principal transition in the past five years.

The authors reported an estimated treatment effect of about $\hat\tau = -0.035$ \citep{principal_turnover}.
The first column of Table~\ref{tab:application_results} shows that matching on observed covariates $\mathbf{X}$ reduces bias by about $\hat\Delta_{\tau_x} = 0.012$.
The estimated breakage in parallel trends for $\mathbf{X}$, $\hat\Delta_x$ (estimated by taking the difference between the linear regression slope of $\mathbf{X}$ on the post-treatment outcome and the average slope of $\mathbf{X}$ across all the pre-treatment outcomes),  is high for some variables but low for others.
Altogether, matching on $\mathbf{X}$ appears to reduce the bias by 0.012, roughly one-third of the original authors' reported estimated treatment effect size, affirming the authors decision to match on $\mathbf{X}$.

%To illustrate this point, we compared na\"ive DiD to simple matching on these data.
%The weighted-average point estimate of the na\"ive DiD, obtained from doing a simple DiD on every year and then taking the weighted average of the resulting impact estimates, is $-0.073$.
%By contrast, the point estimate obtained from DiD after doing one-to-one matching on $\mathbf{X}$\footnote{To perform this matching step we use the \textit{MatchIt} package in R, which uses the nearest one-to-one matching based on estimated propensity scores \citep{MatchIt}.} is $-0.058$.
%Although we should not expect the difference (in this case 0.015) to be exactly $\hat\Delta_{\tau_x}$ because of imperfect matching and probable deviations from the linear model, in this case the difference is very close to our estimated bias difference. 

\begin{table}

\begin{tabular}{l|l|l} 
& Match on $\mathbf{X}$ & Match on $\mathbf{X}$ and $\mathbf{Y_{T}}$ \\ 
 \hline   
Estimated Parameters & $\hat{\vec{\delta}}_x = (-0.37, 0.034, -0.0029, 0.032, $ & $\hat r_{\theta}^T = 0.937; [0.93, 0.94]$  \\
 & $\qquad \qquad 0.013, -0.071, -0.0022)$ &  \\
 & $\hat\Delta_x = (-0.010, -0.23, 0.25, -0.12,$ &  $\hat{s} = \frac{\hat{\bar{\beta}}_{\theta, \text{pre}}}{\hat \beta_{\theta, T}} = 0.936; [0.91, 0.96]$ \\
 & $\qquad \qquad -0.19, -0.00073, -0.041)$ & $\hat{\tilde{\delta}}_{\theta} = -0.11$ \\
\hline
Match & Default to Yes & Yes, because $\hat r_{\theta}^T > 1 - \left|1 - \hat{s}\right|$\\
 & & [53\% bootstrap samples also agree] \\ 
\hline
Estimated Reduction in Bias & $\hat\Delta_{\tau_x} =  0.012; [0.009, 0.018]$  & $\hat\Delta_{\tau_Y} = 0.0027; [0.0006, 0.0062]$ \\ 
\end{tabular}
\caption{Results from applying guidelines in Section~\ref{section:sensitivity_analysis} on principal turnover empirical application.
The first row contains key parameters that summarize how the matching may help or not for matching on $\mathbf{X}$ (left) and additionally on $\mathbf{Y}_T$ (right).
The second row is a binary answer on whether the practitioner should match on the respective variables prior to their DiD analysis. The last row contains an estimate of the reduction in bias from the resulting matching. 95\% bootstrapped confidence intervals for key parameters are also provided in square brackets.} 
\label{tab:application_results}

\end{table}

To determine whether to additionally match on the six available pre-treatment outcomes ($T = 6$), we estimate the conditional reliability and breakage in parallel trends of our latent covariate (see Guideline~\ref{guideline:matchboth}).
For each possible treatment year, we begin by obtaining ``residualized'' responses $\tilde{Y}_{i, t}$ for all time periods through the residuals of the same seven separate linear regression of $Y_{i, t}$ on $\mathbf{X}_i$ within the control group used above in the match-on-$X$ evaluation.
Using the residualized outcomes as the new outcome, we then obtain an estimate of the variance of the noise in the outcome, $\hat \sigma_E^2$, by leveraging the information from only the last two pre-treatment outcomes and the fact that the error of the variance is assumed homoscedastic throughout the pre-treatment periods.
We next use the equations in Section~\ref{subsection:rule_matching_Y} to obtain estimates of the pre- and post-treatment $\theta$ coefficients by exploiting the fact that, e.g., the variance of the residualized pre-treatment outcome is a simple function of the pre-treatment slope and $\sigma_E^2$.
This gives (averaged across all treatment years) $\hat{\bar{\beta}}_{\theta, \text{pre}} = 0.54$ and $\hat \beta_{\theta, T} = 0.58$, with $\hat{s} = 0.936$, suggesting a slight break in parallel trends.

We also estimate the imbalance of $\theta$, getting $\hat{\tilde{\delta}}_{\theta} = -0.11$, using the fact that the difference in means of the treated and control group's residualized pre-treatment outcome is proportional to $\delta_{\theta}$ under our linear model.

Finally, we estimate the reliability $\hat r_{\theta}^T = 0.937$ by using a simple plug-in estimator, using the estimate of $\hat{\bar{\beta}}_{\theta, \text{pre}}^2$ and $\hat\sigma_E^2$.
For the individual years, the reliability values were generally above 0.92, with a low outlier of 0.86.
The $s$-values were generally near one, with a low of 0.77.
The second column of Table~\ref{tab:application_results} summarizes the key estimated parameters, and Appendix~\ref{Appendix:application_details} gives further details of the by-year results.

As our reliability is higher than our $s$-value, we recommend matching on the six available pre-treatment outcome in addition to the baseline covariates.
Similarly, 7 of the 12 years had a match recommendation (across bootstraps, the middle 95\% of number of years with a match recommendation was 6 through 10, suggesting stable match recommendations for at least half of all years).
The right hand side of Table~\ref{tab:application_results} shows an estimated additional reduction in bias, beyond matching on $X$, of around $0.0027$.
The small reduction in bias is because the breakage in parallel trends, $\hat{\bar{\beta}}_{\theta, \text{pre}}/\hat \beta_{\theta, T} = 0.936$, is relatively small (close to one) for $\tilde{\theta}$.
We, nevertheless, recommend to match because Guideline~\ref{guideline:matchboth} holds and because the estimated reliability ($\hat r_{\theta}^T$) is very high, meaning the risk of bias amplification from matching is low.
We additionally found that the decision to match (second row third column) remained ``Yes'' for over 50\% of all 1000 bootstrapped runs, when we perform a clustered bootstrapped sample by resampling each school's entire current and past outcome series.
Our provided script again automates all of these estimation steps.

Although the original authors \citet{principal_turnover} emphasised the importance of matching on different pre-treatment trends to account for the latent confounders, our results interestingly suggest that matching on the observed covariates was actually more important.
That being said, we note that matching on $\mathbf{X}$, insofar as it is correlated with the outcome, makes matching on the outcome less necessary.
Furthermore, our guideline is biased towards underestimating the benefit of matching on lagged outcome if the conditional relationship between outcome and latent confounder is not parallel pre-treatment (see Appendix~\ref{appendix:assump1}), which further supports the match recommendation here.

\section{Concluding Remarks}
\label{section:discussion}
We explore the bias of matching prior to a DiD analysis when parallel trends does not initially hold by mathematically characterizing the bias under a linear structural equation.
We verify that matching on observed covariates likely leads to a reduction in bias.
Further, we find that matching on pre-treatment outcomes exhibits a bias-bias trade off between recovering the latent confounder and undermining any pre-existing parallel trend.

We use our results to create guidelines for determining whether matching is recommended, and further provide an estimate of the resulting reduction in bias. We apply this strategy to a recent application of a DiD analysis involving matching to evaluate the impact of principal turnover on student achievement \citep{principal_turnover}. We find evidence that the authors' decision to match on all available pre-treatment outcomes  and especially on observed covariates $\mathbf{X}$  did in fact reduce bias and create a more credible causal estimate. 

Our work, however, is not comprehensive. First, our results are specific to the linear structural equation model. Although this is a useful starting point that allows for flexible exploration, it still contains strong parametric limitations.
For example, our setup does not explicitly allow for time varying covariates, instead focusing on time varying relationships between covariate and outcome (see Remark~\ref{remark:time_varying}).
We did, however, consider interaction effects in Appendix~\ref{Appendix:interaction}; these findings did not substantively differ from those found in Theorem~\ref{theorem:mainresults}.
% We also only consider an assignment mechanism where units self-select into treatment or control based on observed and unobserved confounders as opposed to selection based on the actual pre-period outcome as detailed in Remark~\ref{remark:assignment_mech}.

We believe our framing could be connected to time-varying contexts, however.
For example, in the simple pre-post case, let $(\theta_{i,t}, X_{i,t})$ be evolving through an AR(1) process, i.e., $X_{i, t + 1} = \rho^X X_{i,t} + \kappa_{i,t}^X$, where $\kappa_{i,t}^X$ is an independent Gaussian random variable, $|\rho^X| \leq 1$, and $\theta_{i, t+1}$ is defined similarly.
Then our model in Equation~\eqref{eq:simpleresponsemodel} can be re-parameterized to account for this time varying covariate as:
$$Y_{i,0}(0) = \beta_{0, 0} +\beta_{\theta, 0} \theta_{i, 0}  + \beta_{x, 0} X_{i, 0} + \epsilon_{i, 0},$$
$$Y_{i,1}(0) = \beta_{0, 1} +\beta_{\theta, 1} \theta_{i, 1}  + \beta_{x, 1} X_{i, 1} + \epsilon_{i, 1} = \beta_{0, 1} + \tilde{\beta}_{\theta, 1} \theta_{i, 0} + \tilde{\beta}_{x, 1} X_{i, 0} + \tilde{\epsilon}_{i, 1},$$
where $\tilde{\beta}_{x, 1} =  \rho^X \beta_{x, 1}$, $\tilde{\beta}_{\theta, 1}$ is defined similarly, and $\tilde{\epsilon}_{i, 1} = \epsilon_{i, 1} + \beta_{x, 1}\kappa_{i, 1}^X + \beta_{\theta, 1}\kappa_{i, 1}^{\theta}$.
This suggests that Equation~\eqref{eq:simpleresponsemodel} could account for time varying confounders.

The tests for whether to match, discussed in Section~\ref{section:sensitivity_analysis}, are based on our linear model and additional stability assumptions.
We acknowledge that there are potentially other interesting directions and more robust ways to determine whether matching is suitable for a specific application; this could be fruitful area for future work.
More fully assessing how rigorous our proposed guidelines are in the face of model misspecification is also an interesting question. Additionally, we also only consider an assignment mechanism where units self-select into treatment or control based on observed and unobserved confounders as opposed to selection based on the actual pre-period outcome as detailed in Remark~\ref{remark:assignment_mech}.

Lastly, we only consider the theoretical bias under a perfect matching scheme. In most applications, perfect matches do not exist and in some cases having even an approximately good match is difficult. An interesting future direction would be to quantify the bias under imperfect matching, perhaps even accounting for systematic imperfect matching based on covariates.

\newpage

\bibliography{bibtex_file, my, imai, DiD} 
\bibliographystyle{pa}

\newpage

\appendix

\section{Guideline in two time period setting - sensitivity analysis}
\label{appendix:guideline}
Section~\ref{section:sensitivity_analysis} provides a heuristic guideline for determining whether to additionally match on lagged outcomes.
That guideline uses an estimation strategy that relies on having multiple pre-period outcomes ($T > 1$), so as to be able to estimate properties of the latent structure.
In particular, we use the multiple periods to estimate the reliability of the outcome as a measure of the latent covariate.
Classic DiD settings, however, only have a single pre-treatment period, making this approach not tenable.
As an alternative, we can simply directly assume a value of the reliability $r_{\theta}$ to obtain our guideline recommendation, and then explore a range of hypothetical reliability values to assess how consistent the guideline is, given the observable characteristics of the data. 

The process is similar to that presented in Section~\ref{subsection:rule_matching_Y}, but with the order of estimation slightly altered. First, given an assumed value for $r_{\theta}$, estimate  $\sigma_E^2$ as:
$$\hat \sigma_E^2 = \widehat{Var}(\tilde{Y}_{i, 0} \mid Z_i =0) \times (1- r_{\theta}),$$
where $T = 1$.
This is consistent because 
$$r_{\theta} = 1 - \frac{\sigma_E^2}{Var(\tilde{Y}_{i, 0} \mid Z_i =0)} . $$
Once we have an estimate of $\sigma_E^2$, the guidelines remain exactly the same.
For example, we have
\begin{align*}
 \hat \beta_{\theta , t} &= \sqrt{ \widehat{Var}(\tilde{Y}_{i, t} \mid Z_i = 0) - \hat{\sigma}_{E}^2  } = \sqrt{ \widehat{Var}(\tilde{Y}_{i, t} \mid Z_i = 0) r_\theta } , \quad  t = 0, 1 ,\\
\hat\Delta_{\theta} &= \hat \beta_{\theta, 1}  - \hat\beta_{\theta, 0} \mbox{ and } \\
\hat{\tilde{\delta}}_{\theta} &= \frac{\hat{E}(\tilde{Y}_{i, T - 1} \mid Z_i = 1) - \hat{E}(\tilde{Y}_{i, T - 1} \mid Z_i = 0)}{\hat\beta_{\theta, 0}} .
\end{align*}

To determine whether to match additionally on the pre-treatment outcome, check Lemma~\ref{lemma:suffcond_match} with the estimated values and the given $r_{\theta}$ as before.
To finally get the estimated reduction in bias, compute
$$\hat\Delta_{\tau_{x,y}} := \left| |\hat\Delta_{\theta} \hat{\tilde{\delta}}_{\theta}| - |\hat \beta_{\theta ,1} \hat{\tilde{\delta}}_{\theta} (1 -  r_{\theta}) |\right|, $$
using the above values and the given $r_{\theta}$.

\section{Robustness of Guideline~\ref{guideline:matchboth}}
\label{appendix:assump1}

The idea behind Guideline~\ref{guideline:matchboth} is that we can residualize out the relationship of the observed covariates to the outcome, and then assess whether the pre-treatment residualized outcomes offer enough information on the latent covariate $\theta$ that it is worth matching on them.
Technically, the guidelines are based on an assumption of parallel trends holding for multiple time periods before treatment (Assumption~\ref{assumption:multiple_time_periods}), but we suggest it can be used even when that is not believed to be the case.
We explore this suggestion in this section via a small simulation study that suggests that when Assumption~\ref{assumption:multiple_time_periods} does not hold, the guideline tends towards not recommending matching, but that the general principle of larger breaks in parallel trends and higher levels of reliability both correspond with greater likelhood of a match recommendation.
In other words, we find that in the models we have explored,  Guideline~\ref{guideline:matchboth} gives the correct decision to match nearly 100\% of the time except for cases where matching or not both result in similar biases regardless of Assumption~\ref{assumption:multiple_time_periods}.

There are two areas of possible concern.
The first is whether the guideline is even true when Assumption~\ref{assumption:multiple_time_periods} is violated.
The second is what impact there is on the estimation of the guideline when the weaker version of Assumption~\ref{assumption:multiple_time_periods}, i.e., that $\beta_{\theta, T-1} = \beta_{\theta, T-2}$, is violated (see Theorem~\ref{theorem:matching_Y}).

To assess robustness to these assumptions, we conduct a variety of simulations with different types of violation of Assumption 1, and see how our guideline tends to compares to the ground truth bias that we can directly calculate using our overall Theorem~\ref{theorem:most_general}.

In our simulation settings, we assume four pre-treatment periods $T = 4$.
For our base simulation we set the post-slope $\beta_{\theta, T} = 0.8$ and pre-slopes $\beta_{\theta,t} = 0.0, 0.2, 0.4, 0.6$ for $t = 0, 1, 2, 3$, thus violating the stability assumption and Assumption~\ref{assumption:multiple_time_periods} with a similar magnitude of the violation of the parallel trends as we have in the post-period.
0.25 0.50 0.60 0.35
0.40 0.30 0.05 0.25
We include two observed covariates, where, for our base simulation, $\vec{\beta}_{x, t} = (0.25, 0.40), (0.50, 0.30), (0.60, 0.05)$, $(0.35, 0.25)$, and $(0.70, 0.50)$ for $t = 0, 1, 2, 3, 4$.
We initially allow $\text{Cor}(\mathbf{X}, \theta) = (0.3, 0.6)$ so that matching on $X$ impacts balance in $\theta$.
We also initially set the correlation of the two $\mathbf{X}$ to 0.5. 
All our covariates have unit variance, and the mean of the treatment group is 1 and the control group is 0.
We then explore a range of $\sigma_E^2$ values to see how overall reliability impacts our results.

For each scenario, we conduct 1,000 trials where we repeatedly generate a synthetic dataset according to our model and then estimate our guideline based on those data.
We compare these 1,000 estimates of guidance and bias to the the true guidance and bias given the parameters.
We plot the results in Figure~\ref{fig:robustness}, with how often Guideline~\ref{guideline:matchboth} recommends one to match (left figure) and the average estimated reduction in bias (right figure).
The $x$-axis shows different levels of residual noise, with higher noise corresponding to smaller reliability and reduced benefit of matching.
Consequently, at a certain threshold of $\sigma_E^2$, the bias from matching will exceed that of not matching. We label this turnover point, calculated with our general theorem, with a dotted vertical red line. 

We correctly classify to match (left of red dotted line) or to not match (right of red dotted line) nearly 100\% of the time as long as one is not in a scenario where matching makes little difference.
The right plot of Figure~\ref{fig:robustness} shows the estimated reduction in bias $\hat\Delta_{\tau_{x, y}}$ is approximately a constant factor away from the true $\Delta_{\tau_{x, y}}$. For example, when $\sigma_E = 0.30$, the true reduction in bias is $0.13$ while our estimated reduction in bias is $0.081$. When $\sigma_E = 1.20$, a setting where matching would actually result in a greater average bias (making the difference in bias negative), the true \textit{gain} in bias from matching is $0.056$ while the estimated gain in bias is $0.074$.

We then explored how this pattern would shift with different aspects of the parallel trends assumption and different correlations between covariates.
In particular, we extended our base scenario into a multifactor experiment with three factors.
The first is whether we impose parallel trends on theta, impose the same coefficient only for the last two pre-treatment time periods, or let theta fully vary.
The second is whether we set all the correlations between our covariates to 0, just set the correlation of the observed covariates to 0, just set the correlation of the covariates and theta to 0, or let all be non-zero.
All combinations allow exploring how the structure of our covariates impacts our guidelines.
The third is whether we have parallel trends on the observed covariates, or not.

We plot the results in Figure~\ref{fig:multifactorA}, with the $x$-axis again showing different levels of residual noise.
The $y$-axis now shows the true additional reduction in bias from matching, and each scenario is a dot colored by whether Guideline~\ref{guideline:matchboth} generally gives the correct advice.
That advice should be to match when above the zero line and to not match when below it.
When parallel trends holds pre-treatment, our advice is correct.
When the assumption does not hold, we generally see that we only give the wrong advice when the benefit of matching is slight.
In other words, we tend to be overly pessimistic about matching, only recommending matching when there are substantial benefits.

Overall, we believe our results suggest our proposed decision rule on when to match or not is very robust to deviations of the stability assumption. To further explore our results, we plot, on Figure~\ref{fig:multifactorB}, the difference between the average estimated bias reduction from the guideline minus the true bias reduction calculated from our general theorem which does not require the stability assumptions.
Overall, the estimated bias is all negative, again showing how various violations of stability cause us to underestimate the benefits of matching.

To further understand why our guideline is robust to violations in the stability assumption, we also show the average estimated parameters from following Guideline~\ref{guideline:matchboth} in Table~\ref{tab:robustness} for our core scenario.
In particular, Table~\ref{tab:robustness} shows that our estimated parameters are not far off from the truth despite a violation of the assumption. Therefore, we believe that our heuristic guideline is still useful to guide when to match or not even when the stability assumption of $\theta$ is violated.  

% \cmntM{Unfortunately, when the narrow assumption holds, the error in the estimated bias reduction looks to be about the same as if it does not hold--but the estimated sigma pre is estimated perfectly.  This suggests that our ability to estimate sigma pre is not the issue here---or that something is still wrong with the simulations or calculations.  I can't figure out what is going on, but these results do not add up to a coherent whole.}

% \cmntM{Given how the narrow assumption doesn't really help us, this paragraph is less compelling.}
Lastly, one can calculate the exact theoretical bias of $\hat\sigma_E^2$ under the violation of our assumption. More specifically, we estimate $\hat\sigma_E^2$ by 
$$\hat{\sigma}_{E}^2  = \frac{1}{2} \widehat{ Var}\left(\tilde{Y}_{i, T - 1} - \tilde{Y}_{i, T - 2} \mid Z_i = 0 \right).$$
\noindent We show that
$$\widehat{ Var}\left(\tilde{Y}_{i, T - 1} - \tilde{Y}_{i, T - 2} \mid Z_i = 0 \right) = (\beta_{\theta, T-1} - \beta_{\theta, T-2})^2 \tilde{\sigma}_{\theta}^2 + 2 \sigma_E^2.$$
Therefore, if $\beta_{\theta, T-1} \neq \beta_{\theta, T-2}$ we always overestimate $\sigma_E^2$, giving a more conservative guideline. In other words, if Guideline~\ref{guideline:matchboth} suggests to match, it should always be correct (under our linear model). If Guideline~\ref{guideline:matchboth} suggests to not match, it may be still favorable to match. This is also reflected in the left plot of Figure~\ref{fig:robustness} since we over-conservatively state to not match (most of the time) near the red dotted-line (from the left) even when one should match.

\begin{figure}[t!]
\begin{center}
\includegraphics[width=12cm, height = 5cm]{"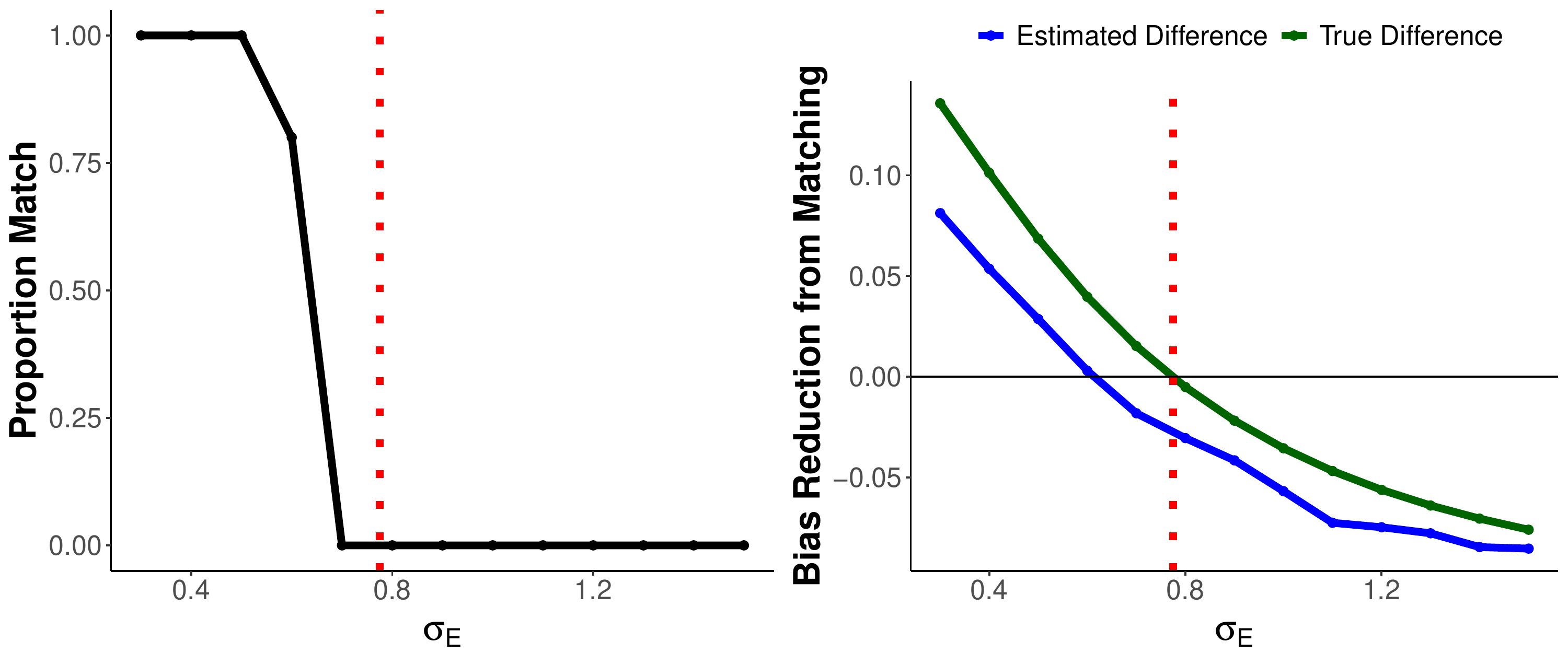"}
\caption{Guideline~\ref{guideline:matchboth} performance under model misspecification: all parameters match those in the model presented on Figure~\ref{fig:bias_complications} of the main paper except we violate Assumption~\ref{assumption:multiple_time_periods} with two observed covariates  Red dotted line represents the turning point on when one should match or not. We have 4 pre-period time points with $\beta_{\theta, t} = 0.0, 0.2, 0.4, 0.6, \beta_{\theta, T} = 0.8$ (a violation of the stability of $\theta$ slopes). We also break stable trends for $\mathbf{X}$ similarly. Left plot shows if we naively followed Guideline~\ref{guideline:matchboth} (pretending Lemma~\ref{lemma:general_suffcond_match} holds without Assumption~\ref{assumption:multiple_time_periods} and also pretending $\beta_{\theta, T-1} = \beta_{\theta, T-2}$ is true) we still correctly classify when to match or not mostly all the time except when near the decision boundary in red. Right plot shows both the true difference in bias of matching only on $\mathbf{X}$ and matching on both observed covariates and pre-period outcome DiD (in green) and the estimated difference, $\hat\Delta_{\tau_{x, y}}$, obtained from Guideline~\ref{guideline:matchboth} (in blue). }
\label{fig:robustness}%
\end{center}
\end{figure}

\begin{figure}[t!]
\begin{center}
\includegraphics[width=12cm, height = 5cm]{"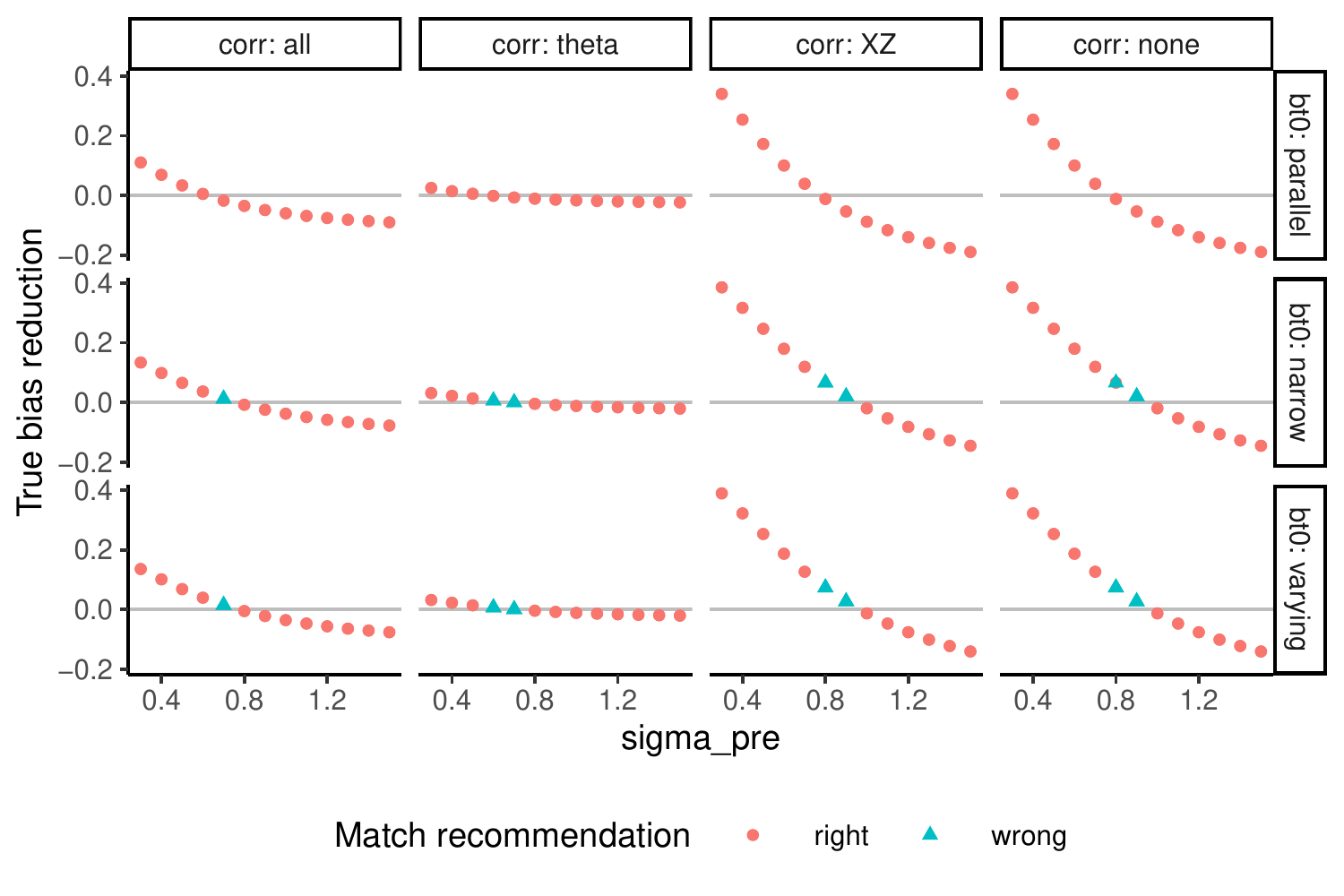"}
\caption{True bias reduction and match recommendation across multifactor simulation scenarios. The columns correspond to the correlation structure of the covariates, with all correlated, the observed correlated with $\theta$ but not themselves, the observed not correlated with $\theta$ but themselves, and all independent. The rows correspond to how $\beta_\theta$ varies: stable, stable in the last two time periods, and fully varying. The final factor of varying coefficients for the observed covariates, or not, does not change the simulation results.}
\label{fig:multifactorA}%
\end{center}
\end{figure} 

\begin{figure}[t!]
\begin{center}
\includegraphics[width=12cm, height = 5cm]{"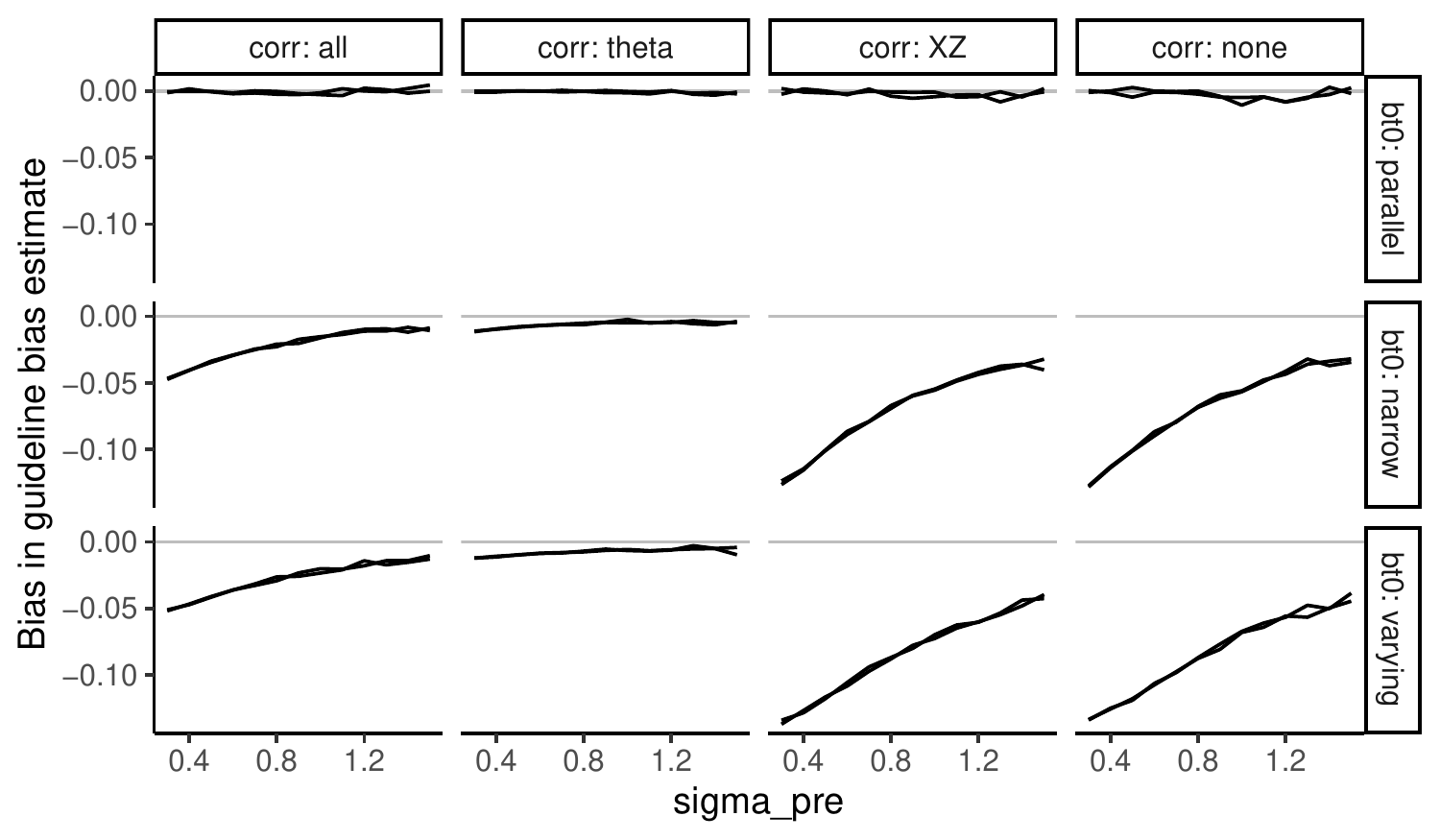"}
\caption{Bias in estimating bias reduction of matching on both covariates and lagged outcomes across multifactor simulation scenarios. All lines are on or below 0, showing that the benefits of matching tend to be underestimated when stability does not hold.}
\label{fig:multifactorB}%
\end{center}
\end{figure}

\begin{table}
\centering
\begin{tabular}{c|c|c|c|c} 
 $\sigma_E$ & $\hat \sigma_E^2$& $\hat{\bar{\beta}}_{\theta, \text{pre}}$ ($\bar{\beta}_{\theta, \text{pre}} = 0.30$) & $\hat \beta_{\theta, T}$ ($\beta_{\theta, T} = 0.80$) & $\hat \delta_{\theta}$ ($\delta_{\theta} = 1$)  \\ 
 \hline   
0.300 &  0.318 & 0.433 & 0.989 & 1.01 \\
0.400 & 0.418 & 0.431 & 0.981 & 1.00 \\
 0.500 &  0.513 & 0.432& 0.983 & 1.00 \\
0.600 & 0.612 & 0.433 & 0.982 & 1.00 \\
0.700&   0.709  & 0.434 & 0.981 & 1.00 \\
0.800 &  0.802& 0.433 & 0.981 & 1.00 \\
0.900&  0.901 & 0.432 & 0.981 & 1.00 \\
1.00 &  1.01  & 0.430 & 0.974 & 1.00 \\
1.10 &  1.11  & 0.431 & 0.970 & 1.00 \\
1.20 &  1.21  & 0.420 & 0.983 & 1.00 \\
1.30 &  1.31  & 0.425 & 0.982 & 1.00 \\
1.40 &  1.41  & 0.428 & 0.972 & 1.00 \\
1.50 &  1.51  & 0.423 & 0.982 & 1.00 \\
\end{tabular}
\caption{Parameter estimates from Figure~\ref{fig:robustness}. The first column represents the different noise level used in the $x$-axis in Figure~\ref{fig:robustness}. Columns 2-4 shows the average estimated parameters ($\hat\sigma_E^2, \hat{\bar{\beta}}_{\theta, \text{pre}}, \hat\beta_{\theta, T}, \hat\delta$) applying Guideline~\ref{guideline:matchboth} when the stability assumption, i.e., $\beta_{\theta, T-1} = \beta_{\theta, T-2}$, and Assumption~\ref{assumption:multiple_time_periods} are violated. We report up to three significant digits.}
\label{tab:robustness}
\end{table}

\section{Additional Details of Application}
\label{Appendix:application_details}
We here give additional details of how the original authors \citeauthor{principal_turnover} performed the matching and analysis steps.
We also give the year-by-year guideline recommendations, analyzing each cohort of schools that had a principal turnover in that given year in turn.

In Missouri 96.5\% of the schools experienced at least one principal turnover between 2001-2015 with 50\% of schools experiencing more than three principal turnovers \citep{principal_turnover}. Therefore, most schools actually received the treatment more than once. To address this challenge, \citeauthor{principal_turnover} treat the principal transition as the unit of observation and arrange the data to have one observation per event per time. This is often known as the ``stacking method'' due to ``stacking'' the data \citep{stacking}.
For example, if a school was treated in 2010 and 2012 and untreated in all other years between 2001-2015, then this school would appear in fifteen rows as two treated observations and thirteen control observations.
Each row would have columns for the lagged outcomes relative to that year. 

The authors then construct a matched comparison group that experience the same downward trend in student achievement, has similar demographic characteristics, and has similar history of teacher and principal turnover using propensity score matching. In particular, they match on (a) current and lagged (up to 5 years) school achievement levels in math, (b) current and lagged proportion of new-to-school teachers, (c) binary indicators for principal transitions in each of the prior 5 years, (d) current and lagged principal experience, and (e) four school demographics variables. To perform the matching step, the authors estimate the propensity scores of all stacked events using a logistic regression with the above predictors (see \citet{principal_turnover} for more details). They then use the constructed matching weights (inverse propensity score weighting) as weights in the subsequent regression analysis to construct a comparable control set.
More specifically, the authors run a regression (using the propensity score weights) with the interaction between time and treatment (both as categorical variables), where time takes a value of 0 when it is the final year before the principal departs. 
For example, if a treated school A had a principal depart in 2012, then school $A$ would appear with treatment value 1 at $t = 0$ while the control units would be any schools that did not have principals departing in 2012. 

Given this model, the DiD estimate is the estimated coefficient of the interaction between treatment and time $t = 1$, where the baseline category is when $t = 0$. The authors further account for school and time fixed effects, the four school demographics listed above, and five lagged outcomes. Finally, the standard errors are clustered by school.
Equation 1 in \citep{principal_turnover} further details the exact regression model. 

 Extending Section~\ref{section:application} of the main paper, we next show in Table~\ref{tab:yearly_table} the year by year results of Table~\ref{tab:application_results}.
Overall, the estimates appear relatively stable across years, with similarly high reliabilities and low breakages in parallel trends near $1.0$ (with 2005 being a notable exception).
We note that middle  95\% of number of years with a match recommendation is 6-10 years out of the 12 years.
\begin{table}
\centering
\begin{tabular}{ccllll} 
 Year &  Match on $\mathbf{Y_{T}}$ & Bias Reduction ($X$)  & Bias Reduction ($Y$) & Reliability  &  Slopes  \\ 
  &  & $\hat\Delta_{\tau_x}$ & $\hat\Delta_{\tau_{Y}}$ & $\hat r_{\theta}^T$ &  $\hat s$ \\ 
 \hline   
2005 &  Yes &  $0.0056$ & $0.0033$ & $0.90$ &  $0.76$\\
2006 &  No & $0.014$ & NA & $0.88$ &  $1.05$\\
2007 &   No & $0.023$ & NA & $0.85$ & $1.02$\\
2008 &  No & $0.028$ & NA & $0.95$ & $1.02$\\
2009 &  Yes & $0.0064$ & $0.0034$ & $0.96$ & $0.91$\\
2010 &  Yes & $0.00084$ & $0.0028$ & $0.93$ & $0.88$\\
2011 &   Yes & $0.022$ & $0.0087$ & $0.95$ & $0.89$\\
2012 & Yes & $0.0048$ & $0.0036$ & $0.96$ & $0.90$\\
2013 & Yes & $0.0093$ & $0.0061$ & $0.96$ & $0.88$\\
2014 &  Yes & $0.018$ & $0.0064$ & $ 0.97$ & $0.93$\\
2015 &  No & $0.0045$ & NA & $0.97$ & $0.97$\\
2016 & No & $0.0060$ & NA & $0.95$ & $0.99$\\
\end{tabular}
\caption{ Detailed year-by-year results from applying guidelines in Section~\ref{section:sensitivity_analysis} on principal turnover empirical application. Table~\ref{tab:yearly_table} presents the year-by-year analysis of applying our Guidelines in Section~\ref{section:sensitivity_analysis} to the principal turnover application. Column 1 shows the year of the post-treatment period. Column 2 presents the match recommendation from Guideline~\ref{guideline:matchboth}. Column 3-4 shows the estimated reduction in bias from matching only on $X$ and then additionally on the pre-treatment outcomes, respectively. The last two columns shows the estimated reliability ($\hat r_{\theta}^T$) and the estimated breakage in parallel slope ($ \hat s := \hat{\bar{\beta}}_{\theta, \text{pre}}/ \hat\beta_{\theta , T}$), respectively.}
\label{tab:yearly_table}
\end{table}

\section{Proofs of Primary Theorems}
\label{appendix:proofs}

\subsection{Proof of Theorem~\ref{theorem:mainresults}}
\label{appendix:proof_mainresult}
The proofs for the bias of the na\"ive DiD and the difference in means estimator are straight forward, thus we omit them. We also leave out all time intercept term $\beta_{0, t}$ when evaluating the relevant expectations because they cancel out. Before proving the bias of the various matching estimators, we will use a well known and useful fact of the multivariate normal distribution.

\begin{lemma}[Multivariate Normal Conditional Distribution]
\label{lemma:multivariate}
Suppose random variables $(X,Y)$ jointly follow a normal distribution: 

$$\begin{pmatrix} X \\ Y \end{pmatrix} \sim N\left(\begin{pmatrix} \mu_{1} \\ \mu_{2}  \end{pmatrix} ,  \begin{pmatrix} V_{11} & V_{12} \\  V_{21} & V_{22}  \end{pmatrix}\right),$$
\noindent then
$$E(Y \mid X) = \mu_2 + V_{21}V_{11}^{-1}(X - \mu_1).$$
\end{lemma}

\subsubsection{Matching on covariates}
$$E\big[\hat{\tau}_{DiD}^X\big]- \tau =  \Delta_{\theta}\big[ \delta_{\theta}  - \rho \frac{\sigma_{\theta}}{\sigma_x}\delta_x \big]$$

\begin{proof}
\begin{align*}
E\big[\hat{\tau}_{DiD}^X\big] - \tau &= E(Y_{i, 1}\mid Z_{i} = 1) -  E(Y_{i, 0}\mid Z_{i} = 1) - \tau \\
&\qquad - ( E_{x \mid Z_i = 1}[E(Y_{i, 1}\mid Z_{i} = 0, X_i = x)] - E_{x \mid Z_i = 1}[E(Y_{i, 0}\mid Z_{i} = 0, X_i = x)])  \\
&=  \mu_{\theta, 1}(\beta_{\theta, 1} - \beta_{\theta, 0}) +  \mu_{x, 1}(\beta_{x, 1} - \beta_{x, 0})  \\
& \qquad - ( E_{x \mid Z_i = 1}(\beta_{\theta, 1}(\mu_{\theta, 0} + \rho \frac{\sigma_{\theta}}{\sigma_{x}}(x - \mu_{x, 0})) + \beta_{x, 1}x ) \\
& \qquad - E_{x \mid Z_i = 1}(\beta_{\theta, 0}(\mu_{\theta, 0} + \rho \frac{\sigma_{\theta}}{\sigma_{x}}(x - \mu_{x, 0})) + \beta_{x, 0}x )) \\
&=  (\mu_{\theta, 1}  - \mu_{\theta, 0})(\beta_{\theta, 1} - \beta_{\theta, 0}) - \rho \frac{\sigma_{\theta}}{\sigma_x}(\mu_{x, 1} - \mu_{x, 0})(\beta_{\theta, 1} - \beta_{\theta, 0})
\end{align*}
\end{proof}
\noindent The fourth line follows because $E(\theta_i \mid Z_i = z, X_i = x) = \mu_{\theta, z} + \rho \frac{\sigma_{\theta}}{\sigma_{x}}(x - \mu_{x, z})$ from Lemma~\ref{lemma:multivariate}. 
    
\subsubsection{Matching on covariates and past outcome}
$$E\big[\hat{\tau}_{DiD}^{X, Y_{0}}\big]- \tau =  \beta_{\theta, 1}\bigg[(1 -  r_{ \theta \mid x})(\delta_{\theta} - \rho \frac{\sigma_{\theta}}{\sigma_x}\delta_x  )\bigg]$$

\begin{proof}
\begin{align*}
E\big[\hat{\tau}_{DiD}^{X, Y_{0}}\big] - \tau &= E(Y_{i, 1}\mid Z_{i} = 1) -  E_{(x,y) \mid Z_i = 1}[E(Y_{i, 1}\mid Z_{i} = 0, X_i = x, Y_{i, 0} = y)] ) - \tau \\
&= \beta_{\theta, 1}(\mu_{\theta, 1} - E_{(x,y) \mid Z_i = 1}[E(\theta_i \mid Z_i = 0,X_i = x, Y_{i, 0} = y)])  \\
\end{align*}
The key is that we have the following multivariate Gaussian,
$$\begin{pmatrix} \theta_i \\ X_i \\ Y_{i, 0} \end{pmatrix}  \mid Z_i = 0 \sim N\left( \begin{pmatrix} \mu_{\theta, 0} \\ \mu_{x, 0}\\ \beta_{\theta, 0}\mu_{\theta, 0} + \beta_{x, 0}\mu_{x, 0} \end{pmatrix}, \Sigma \right), $$
where $\Sigma = $
$$\begin{pmatrix} \sigma_{\theta}^2 & \sigma_{\theta} \sigma_x \rho & \beta_{\theta, 0}\sigma_{\theta}^2 + \beta_{x, 0}\sigma_{\theta}\sigma_x\rho  \\  \sigma_{\theta} \sigma_x \rho  & \sigma_x^2 & \beta_{x, 0}\sigma_x^2 + \beta_{\theta, 0}\sigma_{\theta} \sigma_x \rho  \\   \beta_{\theta, 0}\sigma_{\theta}^2 + \beta_{x, 0}\sigma_{\theta}\sigma_x\rho &  \beta_{x, 0}\sigma_x^2 + \beta_{\theta, 0}\sigma_{\theta} \sigma_x \rho    & (\beta_{\theta, 0})^2 \sigma_{\theta}^2 + (\beta_{x, 0})^2\sigma_x^2 + \sigma_{E}^2 + 2 \beta_{\theta, 0}\beta_{x, 0}\sigma_{\theta}\sigma_x\rho \end{pmatrix}.$$
Next, we apply Lemma~\ref{lemma:multivariate} to obtain $E(\theta_{i} \mid Z_{i} = 0, Y_{i, 0} = y, X_i = x) = $
\begin{align*}
&\mu_{\theta, 0} + \begin{pmatrix} \sigma_{\theta} \sigma_x \rho & \beta_{\theta, 0}\sigma_{\theta}^2 + \beta_{x, 0}\sigma_{\theta}\sigma_x\rho \end{pmatrix} \\
&\begin{pmatrix} \sigma_x^2 & \beta_{x, 0}\sigma_x^2 + \beta_{\theta, 0}\sigma_{\theta} \sigma_x \rho \\ \beta_{x, 0}\sigma_x^2 + \beta_{\theta, 0}\sigma_{\theta} \sigma_x \rho & (\beta_{\theta, 0})^2 \sigma_{\theta}^2 + (\beta_{x, 0})^2\sigma_x^2 + \sigma_{E}^2 + 2 \beta_{\theta, 0}\beta_{x, 0}\sigma_{\theta}\sigma_x\rho 
\end{pmatrix}^{-1} \begin{pmatrix} x - \mu_{x, 0} \\ y - \beta_{\theta, 0}\mu_{\theta, 0} - \beta_{x, 0} \mu_{x, 0} \end{pmatrix}  \\
&= \mu_{\theta, 0} +  \frac{1}{det} \begin{pmatrix} m_1 & m_2 \end{pmatrix} \begin{pmatrix} x - \mu_{x, 0} \\ y - \beta_{\theta, 0}\sigma_{\theta}^2 + \beta_{x, 0} \sigma_{\theta}\sigma_x\rho \end{pmatrix}  \\
&=  \mu_{\theta, 0}  + \frac{m_1}{det}(x - \mu_{x, 0}) + \frac{m_2}{det}(y -  \beta_{\theta, 0}\mu_{\theta, 0} - \beta_{x, 0} \mu_{x, 0}), 
\end{align*}
where we define the following to avoid long algebraic expressions\footnote{Any time we define notations to avoid long algebraic expressions,  we define it only for the respective section. In other words, $m_1, m_2, det$ are defined only for Appendix~\ref{appendix:proof_mainresult} and we reuse this notation for the following sections but with different expressions.}: $m_1 = -\beta_{x, 0}\beta_{\theta, 0}\sigma_{\theta}^2\sigma_x^2(1 - \rho^2) + \rho \sigma_{\theta} \sigma_x\sigma_{E}^2$, $m_2 = \beta_{\theta, 0}\sigma_{\theta}^2\sigma_x^2 (1 - \rho^2)$, and $det = \sigma_x^2((\beta_{\theta, 0})^2 \sigma_{\theta}^2(1-\rho^2) + \sigma_{E}^2)$. Finally, the bias is, 
\begin{align*}
E\big[\hat{\tau}_{DiD}^{X, Y_{0}}\big]  - \tau &=  \beta_{\theta, 1}(\mu_{\theta, 1} - E_{(x,y) \mid Z_i = 1}[E(\theta_i \mid Z_i = 0,X_i = x, Y_{i, 0} = y))]  \\
&= \beta_{\theta, 1}(\mu_{\theta, 1} - (\mu_{\theta, 0}  + \frac{m_1}{det}(\mu_{x, 1} - \mu_{x, 0}) + \frac{m_2}{det}( \beta_{\theta, 0}\mu_{\theta, 1} + \beta_{x, 0} \mu_{x, 1}   -  \beta_{\theta, 0}\mu_{\theta, 0} - \beta_{x, 0} \mu_{x, 0}))) 
\end{align*}
\noindent The rest follows from simple algebraic manipulation.
\end{proof}

\subsection{Proof of Multiple covariates and time points (Section~\ref{section:generalization})}
\label{Appendix:generalization_proofs}
We first prove Theorem~\ref{theorem:most_general}. We omit the proof of the bias of the na\"ive DiD estimator and matching on only $\mathbf{X}$ as the results come directly from elementary calculation or applying Lemma~\ref{lemma:multivariate} to find $E_{\mathbf{x} \mid Z_i = 1}[E(\boldsymbol{\theta}_i \mid \mathbf{X}_{i} = \mathbf{x}, Z_i = 0)]$. We now prove the bias of matching on both observed covariate and the $T$ pre-treatment outcomes. 

\begin{proof} 
To compute this bias, we must find $E_{(\mathbf{x}, \mathbf{y^T}) \mid Z_i = 1}[E(\boldsymbol{\theta_i} \mid Z_{i} = 0, \mathbf{X}_{i} = \mathbf{x}, \mathbf{Y_i^T} = \mathbf{y^t})]$. Similar to the above proofs, the key is to realize that: 
\begin{equation*}
\begin{pmatrix} \boldsymbol{\theta}_i \\ \mathbf{X}_{i} \\ \mathbf{Y}_i^{T} \end{pmatrix} \mid Z_i = 0 \sim N\left(\begin{pmatrix} \vec{\mu}_{\theta, 0} \\ \vec{\mu}_{x, 0} \\ (\vec{\beta}_{\theta, 0})^\top \vec{\mu}_{\theta, 0} + (\vec{\beta}_{x,0})^\top \vec{\mu}_{x, 0} \\ \vdots \\(\vec{\beta}_{\theta, T-1})^\top \vec{\mu}_{\theta, 0} + (\vec{\beta}_{x,T-1})^\top \vec{\mu}_{x, 0} \end{pmatrix} , \begin{pmatrix} \Sigma_{\theta \theta} & \Sigma_{\theta X} & \Sigma_{\theta Y_{T}} \\ \Sigma_{X \theta} & \Sigma_{XX} & \Sigma_{X Y_{T}} \\ \Sigma_{Y_{T} \theta} & \Sigma_{Y_{T} X} & \Sigma_{Y_{T} Y_{T}}  \end{pmatrix} \right), 
\end{equation*}
We again apply Lemma~\ref{lemma:multivariate} to obtain the desired result.
\end{proof}

We now prove Theorem~\ref{theorem:bias_multiple_T_noX}. 

\begin{proof}
 The key is to utilize the general result in Theorem~\ref{theorem:most_general} and directly compute $(\Sigma_{Y_{T} Y_{T}})^{-1}$ assuming no $X$. We have that 
$$\Sigma_{Y_{T} Y_{T}} = \sigma_E^2 I + a'a, \quad a:= \begin{pmatrix}
    \beta_{\theta, 0} \sigma_{\theta} & , \dots , & \beta_{\theta, T-1} \sigma_{\theta},
\end{pmatrix}$$
where we purposefully write the covariance matrix in this form to apply the well known Sherman-Morrison inverse formula. Consequently, we have that
$$(\Sigma_{Y_{T} Y_{T}})^{-1} = \frac{1}{\sigma_E^2} I - \frac{a'a / (\sigma_E^2)^2}{1 + aa'/\sigma_E^2}.$$
The denominator simplifies to the following
\begin{align*}
    1 + aa'/\sigma_E^2 &= 1 +  \frac{\sum_{t = 0}^{T-1} \beta_{\theta, t}^2 \sigma_{\theta}^2}{\sigma_E^2} \\
    &= \frac{\sum_{t = 0}^{T-1} \beta_{\theta, t}^2 \sigma_{\theta}^2 + \sigma_E^2}{\sigma_E^2}
\end{align*}
Let $l = \sum_{t = 0}^{T-1} \beta_{\theta, t}^2 \sigma_{\theta}^2 + \sigma_E^2$.
Then we have that
$$(\Sigma_{Y_{T} Y_{T}})^{-1} = \frac{l}{l \sigma_E^2} I - \frac{a'a}{\sigma_E^2 l}$$
Writing this out as matrix we have that 
$$(\Sigma_{Y_{T} Y_{T}})^{-1} = \frac{1}{ l \sigma_E^2} \begin{pmatrix} l - \beta_{\theta, 0}^2 \sigma_{\theta}^2 &  & &  - \beta_{\theta, i}\beta_{\theta, j} \sigma_{\theta}^2\\
 & l - \beta_{\theta, 1}^2 \sigma_{\theta}^2 &  &  \\
  & \ddots & \\
- \beta_{\theta, i}\beta_{\theta, j} \sigma_{\theta}^2   & & & l - \beta_{\theta, T-1}^2 \sigma_{\theta}^2    
\end{pmatrix}$$

Next, we have that
$\Sigma_{\theta Y_{T}} = \begin{pmatrix}
    \beta_{\theta, 0} \sigma_{\theta}^2 & , \dots , & \beta_{\theta, T-1}\sigma_{\theta}^2
\end{pmatrix}$. Therefore, we have that
\begin{align*}
  \Sigma_{\theta Y_{T}} (\Sigma_{Y_{T} Y_{T}})^{-1}  &= \begin{pmatrix}
    \frac{\beta_{\theta, 0}\sigma_{\theta}^2 (1-l)}{l \sigma_E^2} & , \dots , & \frac{\beta_{\theta, T-1}\sigma_{\theta}^2 (1-l)}{l \sigma_E^2}
\end{pmatrix} \\
&= \begin{pmatrix}
    \frac{\beta_{\theta, 0}\sigma_{\theta}^2 }{l} & , \dots , & \frac{\beta_{\theta, T-1}\sigma_{\theta}^2}{l}
\end{pmatrix}
\end{align*}
Finally we have that 
$$
 \Sigma_{\theta Y_{T}} (\Sigma_{Y_{T} Y_{T}})^{-1} \begin{pmatrix}
 \vec{\beta}_{\theta, 0 }^\top \vec{\delta}_{\theta}  \\ \vdots \\ \vec{\beta}_{\theta, T-1}^\top \vec{\delta}_{\theta}    
 \end{pmatrix} = \delta_{\theta} \frac{\sum_{t = 0}^{T-1} \beta_{\theta, t}^2 \sigma_{\theta}^2} {\sum_{t = 0}^{T-1} \beta_{\theta, t}^2 \sigma_{\theta}^2 + \sigma_E^2}
$$
Putting this into the final formula in the last line of Theorem~\ref{theorem:most_general} gives the desired claim. 
\end{proof}

We now prove Theorem~\ref{theorem:bias_multiple_T}. We first state two useful Lemmas about matrix inversions:
\begin{lemma}[Matrix Inversion]
\label{lemma:matrix_inversion}
Suppose we have a matrix $\mathbf{Q}_{(T+1) \times (T+1)}$ such that:

$$\mathbf{Q} = \begin{pmatrix} x & y& y & \dots & y  \\ y & x & & & & \\ y & &  x & &  \\ \vdots & & & \ddots & \\ y & y & \dots & & x \end{pmatrix},$$ 
where $x,y \in \mathbb{R}$ such that $(x-y) \neq 0$ and $(x+Ty) \neq 0$.
Given such $\mathbf{Q}$, we have that,

$$\mathbf{Q}^{-1} = \begin{pmatrix} \frac{x+ (T-1)y}{(x-y)(x+Ty)} & -\frac{y}{(x-y)(x+Ty)}& -\frac{y}{(x-y)(x+Ty)} & \dots & -\frac{y}{(x-y)(x+Ty)}  \\ -\frac{y}{(x-y)(x+Ty)} & \frac{x+ (T-1)y}{(x-y)(x+Ty)} & & & & \\ -\frac{y}{(x-y)(x+Ty)} & &  \ddots & &  \\ \vdots & & & \ddots & \\ -\frac{y}{(x-y)(x+Ty)} & -\frac{y}{(x-y)(x+Ty)} & \dots & & \frac{x+ (T-1)y}{(x-y)(x+Ty)} \end{pmatrix}$$
\end{lemma}

\begin{lemma}[Block Matrix Inversion]
\label{lemma:block_matrix_inversion}
Suppose we have a block matrix $\mathbf{P} =  \begin{pmatrix} \mathbf{A} & \mathbf{B} \\ \mathbf{C} & \mathbf{D} \\ \end{pmatrix}$ then
$\mathbf{P}^{-1} =  \begin{pmatrix} \mathbf{A}^{-1} + \mathbf{A}^{-1}\mathbf{B}(\mathbf{D} - \mathbf{C}\mathbf{A}^{-1}\mathbf{B})^{-1}\mathbf{C}\mathbf{A}^{-1} & -\mathbf{A}^{-1}\mathbf{B}(\mathbf{D} - \mathbf{C}\mathbf{A}^{-1}\mathbf{B})^{-1} \\ -(\mathbf{D} - \mathbf{C}\mathbf{A}^{-1}\mathbf{B})^{-1}\mathbf{C}\mathbf{A}^{-1} & (\mathbf{D} - \mathbf{C}\mathbf{A}^{-1}\mathbf{B})^{-1} \\ \end{pmatrix}$
\end{lemma}
\noindent The proofs of these Lemmas are omitted because these are well established properties of matrix inverses. 

\begin{proof} (of Theorem~\ref{theorem:bias_multiple_T})
We first assume that $\Sigma_{\theta X} = 0$, i.e., $\theta$ and $\mathbf{X}$ are uncorrelated. We prove the case when $\Sigma_{\theta X} \neq 0$ later. We remind readers that we are under Assumption~\ref{assumption:multiple_time_periods} that states that all slopes of $\theta$ and $\mathbf{X}$ are the same in the pre-treatment periods.
Also, $\theta$ is univariate.
Consequently, without loss of generality we write all the pre-treatment slopes of $\theta$ and $\mathbf{X}$ as $\beta_{\theta, 0}$ and $\vec{\beta}_{x, 0}$, respectively. 

We start the proof by writing the bias as, 
\begin{align*}
E\big[\hat{\tau}_{gDiD}^{\mathbf{X}, \mathbf{Y}^{T}}\big]  - \tau  &= \beta_{\theta, T}(\mu_{\theta, 1} - E_{(\mathbf{x}, \mathbf{y^T}) \mid Z_i = 1}[E(\theta_i \mid Z_i = 0, \mathbf{X_i} = \mathbf{x}, \mathbf{Y_i^T} = \mathbf{y^T})]) \\
\end{align*}
Focusing on the inner expectation $E(\theta_i \mid Z_i = 0, \mathbf{X_i} = \mathbf{x}, \mathbf{Y_i^T} = \mathbf{y^t})$, we notice the following multivariate Gaussian distribution using the fact that $\theta$ and $\mathbf{X}$ are uncorrelated,
$$\begin{pmatrix} \theta_i \\ \mathbf{X_i} \\ Y_{i, T - 1}\\ Y_{i, T - 2} \\ \vdots \\ Y_{i, 0} \end{pmatrix}  \mid Z_i = 0 \sim N\left( \begin{pmatrix} \mu_{\theta, 0} \\ \vec{\mu}_{x, 0}\\ \beta_{\theta, 0}\mu_{\theta, 0} + \vec{\beta}_{x,  0}^\top \vec{\mu}_{x, 0} \\ \beta_{\theta, 0}\mu_{\theta, 0} + \vec{\beta}_{x,  0}^\top \vec{\mu}_{x, 0} \\ \vdots \\ \beta_{\theta, 0}\mu_{\theta, 0} + \vec{\beta}_{x,  0}^\top \vec{\mu}_{x, 0} \end{pmatrix}, \Sigma \right),$$
where 
\begin{align*}
 \Sigma &= \begin{pmatrix} \sigma_{\theta}^2 & \mathbf{B}_{1 \times (T+p)} \\ \mathbf{B}_{(T+p) \times 1}^\top & \mathbf{D}_{(T+p) \times (T+p)} \end{pmatrix}   \\
 \mathbf{B} &= \begin{pmatrix} \vec{0}_{1 \times p} & \beta_{\theta, 0}\sigma_{\theta}^2 & \beta_{\theta, 0}\sigma_{\theta}^2  & \dots & \beta_{\theta, 0}\sigma_{\theta}^2  \end{pmatrix} \\
 \mathbf{D} &= \begin{pmatrix} \Sigma_{XX} & \mathbf{C}_{p \times T} \\ \mathbf{C}_{T \times p}^\top & \mathbf{A}_{T \times T} \end{pmatrix} \\
 \mathbf{C}^\top &= Cov(\mathbf{Y_i^T}, \mathbf{X_i}) = \begin{pmatrix}  & (\vec{\beta}_{x, 0}^\top \Sigma_{XX})_{1 \times p} &  \\  & (\vec{\beta}_{x, 0}^\top \Sigma_{XX})_{1 \times p} &  \\  & \vdots & \\  & (\vec{\beta}_{x, 0}^\top \Sigma_{XX})_{1 \times p} &  \end{pmatrix} \\ 
 \mathbf{A} &= Var(\mathbf{Y_i^T}), 
\end{align*}
and $\mathbf{A}$ has the matrix structure of matrix $\mathbf{Q}$ in Lemma~\ref{lemma:matrix_inversion} with $x = (\beta_{\theta, 0})^2 \sigma_{\theta}^2 + \vec{\beta}_{x,0}^\top \Sigma_{XX}\vec{\beta}_{x,0} + \sigma_E^2$ and $y = x - \sigma_E^2$ by definition of $\mathbf{Y_i^T}$ in our linear structural model. 

To calculate $E(\theta_i \mid Z_i = 0, \mathbf{X_i} = \mathbf{x}, \mathbf{Y_i^T} = \mathbf{y^T})$, we need to invert matrix $\mathbf{D}$. This can be done by first applying Lemma~\ref{lemma:block_matrix_inversion} to matrix $\mathbf{D}$ giving us,
$$\mathbf{D}^{-1} = \begin{pmatrix} \Sigma_{XX}^{-1}  + \Sigma_{XX}^{-1} \mathbf{C}(\mathbf{A} - \mathbf{C}^\top \Sigma_{XX}^{-1} \mathbf{C} )^{-1}\mathbf{C}^\top \Sigma_{XX}^{-1}  & -\Sigma_{XX}^{-1} \mathbf{C} (\mathbf{A} - \mathbf{C}^\top \Sigma_{XX}^{-1} \mathbf{C} )^{-1} \\ -(\mathbf{A} - \mathbf{C}^\top \Sigma_{XX}^{-1} \mathbf{C} )^{-1}\mathbf{C}^\top \Sigma_{XX}^{-1} & (\mathbf{A} - \mathbf{C}^\top \Sigma_{XX}^{-1} \mathbf{C} )^{-1} \end{pmatrix}$$
Next, we simplify $\mathbf{M} := (\mathbf{A} - \mathbf{C}^\top \Sigma_{XX}^{-1} \mathbf{C} )^{-1}$. The key is that 
$$\mathbf{C}^\top \Sigma_{XX}^{-1} \mathbf{C} = \begin{pmatrix} \vec{\beta}_{x, 0}^\top \Sigma_{XX} \vec{\beta}_{x, 0} & \dots & \vec{\beta}_{x, 0}^\top \Sigma_{XX} \vec{\beta}_{x, 0} \\ \vec{\beta}_{x, 0}^\top \Sigma_{XX} \vec{\beta}_{x, 0} & \dots & \vec{\beta}_{x, 0}^\top \Sigma_{XX} \vec{\beta}_{x, 0} \\ \vdots & \vdots & \vdots \\ \vec{\beta}_{x, 0}^\top \Sigma_{XX} \vec{\beta}_{x, 0} & \dots & \vec{\beta}_{x, 0}^\top \Sigma_{XX} \vec{\beta}_{x, 0}  \end{pmatrix}_{T \times T}.$$
Thus, $\mathbf{M}$ is again the same matrix structure as matrix $\mathbf{Q}$ in Lemma~\ref{lemma:matrix_inversion} with $x' = (\beta_{\theta, 0})^2 \sigma_{\theta}^2 + \sigma_{E}^2$ and $y' = (\beta_{\theta, 0})^2 \sigma_{\theta}^2$. Consequently, we can apply Lemma~\ref{lemma:matrix_inversion} to get $\mathbf{M}^{-1}$, which is again the same form as the matrix $\mathbf{Q}$ in Lemma~\ref{lemma:matrix_inversion} with $x = \frac{ (T - 1) (\beta_{\theta, 0})^2 \sigma_{\theta}^2 + \sigma_{E}^2 }{\sigma_{E}^2 \{ T(\beta_{\theta, 0})^2 \sigma_{\theta}^2 + \sigma_{E}^2 \}}$ and  $y = -\frac{ (\beta_{\theta, 0})^2 \sigma_{\theta}^2 }{\sigma_{E}^2 \{ T(\beta_{\theta, 0})^2 \sigma_{\theta}^2 + \sigma_{E}^2 \}} $. Putting everything together we have that:
$$\mathbf{D}^{-1} = \begin{pmatrix} a  & 
\begin{pmatrix} \vec{\beta}_{x,0} & \vec{\beta}_{x,0} & \dots & \vec{\beta}_{x,0}
\end{pmatrix}_{p \times T} \mathbf{M}^{-1} \\ 
 \mathbf{M}^{-1}\begin{pmatrix} \vec{\beta}_{x,0}^\top \\ \vec{\beta}_{x,0}^\top \\ \vdots \\ \vec{\beta}_{x,0}^\top
\end{pmatrix}_{T \times p} & \mathbf{M}^{-1},
\end{pmatrix}$$
where $a$ is some matrix that will be irrelevant to further calculations. Returning back to calculating $E(\theta_i \mid Z_i = 0, \mathbf{X_i} = \mathbf{x}, \mathbf{Y_i^T} = \mathbf{y^t})$, we have that
$$E(\theta_i \mid Z_i = 0, \mathbf{X_i} = \mathbf{x}, \mathbf{Y_i^T} = \mathbf{y^t})) = \mu_{\theta, 0} + \mathbf{B} \mathbf{D}^{-1} \begin{pmatrix}
\mathbf{x} - \vec{\mu}_{x, 0} \\ y_{T - 1} - \beta_{\theta, 0} \mu_{\theta, 0} - \vec{\beta}_{x, 0}^\top \vec{\mu}_{x, 0} \\ \vdots \\ y_{0} - \beta_{\theta, 0} \mu_{\theta, 0} -  \vec{\beta}_{x, 0}^\top \vec{\mu}_{x, 0}
\end{pmatrix}.$$
We also have from simple matrix multiplication that
$$ \mathbf{B}\mathbf{D}^{-1} = \begin{pmatrix} - \frac{\beta_{\theta, 0} \beta_{x_1, 0} \sigma_{\theta}^2 T}{T(\beta_{\theta, 0})^2 \sigma_{\theta}^2 + \sigma_{E}^2} & - \frac{\beta_{\theta, 0} \beta_{x_2, 0} \sigma_{\theta}^2 T}{T(\beta_{\theta, 0})^2 \sigma_{\theta}^2 + \sigma_{E}^2} & \dots & - \frac{\beta_{\theta, 0} \beta_{x_p, 0} \sigma_{\theta}^2 T}{T(\beta_{\theta, 0})^2 \sigma_{\theta}^2 + \sigma_{E}^2} & \frac{\beta_{\theta, 0} \sigma_{\theta}^2}{T(\beta_{\theta, 0})^2 \sigma_{\theta}^2 + \sigma_{E}^2} & \dots & \frac{\beta_{\theta, 0} \sigma_{\theta}^2}{T(\beta_{\theta, 0})^2 \sigma_{\theta}^2 + \sigma_{E}^2}
\end{pmatrix}, $$
where $\vec{\beta}_{x ,0}^\top =  \begin{pmatrix} \beta_{x_1, 0} & \beta_{x_2, 0} & \dots & \beta_{x_p, 0}\end{pmatrix}$. We then finally take the outer expectation to obtain,
\begin{align*}
E_{(\mathbf{x}, \mathbf{y^T}) \mid Z_i = 1}[E(\theta_i \mid Z_i = 0,  \mathbf{X_i} = \mathbf{x}, \mathbf{Y_i^T} = \mathbf{y^t}))] = \mu_{\theta, 0} + \frac{T(\beta_{\theta, 0})^2 \sigma_{\theta}^2 }{T(\beta_{\theta, 0})^2 \sigma_{\theta}^2 + \sigma_{E}^2} \delta_{\theta}\\
\end{align*}

We now return to the more general case of arbitrarily correlated $\theta$ and $\mathbf{X}$, i.e., we let $\Sigma_{\theta X} \neq 0$.
First define a $\tilde{\theta}$ as our latent covariate after partialing out what can be explained by $\mathbf{X}$:
\begin{equation*}
    \tilde{\theta}_i = \theta_i - \Sigma_{\theta X} \Sigma_{XX}^{-1} \mathbf{X_i}.
\end{equation*}
This new latent confounder is purposefully constructed to be uncorrelated with $\mathbf{X}$. In other words, we have the following facts about $\tilde{\theta}$,
\begin{align*}
Cov(\tilde{\theta}_i, \mathbf{X}_i \mid Z_i = z) &= \Sigma_{\theta X} - \Sigma_{\theta X} = 0 \\
\tilde{\mu}_{\theta, z}  &= E(\tilde{\theta}_i \mid Z_i = z) = \mu_{\theta, z} - \Sigma_{\theta X}\Sigma_{XX}^{-1} \vec{\mu}_{x, z} \\
\tilde{\sigma}_{\theta}^2  &= Var(\tilde{\theta}_i \mid Z_i = z) = \sigma_{\theta}^2 - \Sigma_{\theta X}\Sigma_{XX}^{-1} \Sigma_{X \theta} \\
Cov(\tilde{\theta}_i, Y_{i, t} \mid Z_i = z) &= \beta_{\theta ,t}\tilde{\sigma}_{\theta}^2
\end{align*}
We can then rewrite our new latent confounder by,
\begin{equation}
    Y_{i, t} = \beta_{0,t} + \beta_{\theta, t} \tilde{\theta}_i + \vec{\beta}_{x, t}^\top \mathbf{X}_i +  \beta_{\theta, t}(\theta_i - \tilde{\theta}_i) +  \epsilon_{i ,t} 
\label{eq:tilde_theta_outcome}
\end{equation}
Then, using Equation~\eqref{eq:tilde_theta_outcome}, we can write the bias (in this more general setting) as,
\begin{equation}
E\big[\hat{\tau}_{gDiD}^{\mathbf{X}, \mathbf{Y}^{T}}\big]  - \tau  = \beta_{\theta, T}(\mu_{\theta, 1} - E_{(\mathbf{x}, \mathbf{y^T}) \mid Z_i = 1}[E(\tilde{\theta_i} \mid Z_i = 0, \mathbf{X_i} = \mathbf{x}, \mathbf{Y_i^T} = \mathbf{y^t})] - \Sigma_{\theta X} \Sigma_{XX}^{-1} \vec{\mu}_{x, 1}).
\label{eq:bias_both_proof}
\end{equation}
It remains to find $E(\tilde{\theta_i} \mid Z_i = 0, \mathbf{X_i} = \mathbf{x}, \mathbf{Y_i^T} = \mathbf{y^t})$. However, the key is that $Cov(\tilde{\theta}_i, \mathbf{X}_i \mid Z_i = 0) = \Sigma_{\theta X} - \Sigma_{\theta X} = 0$ (by construction). Thus, we are back in the uncorrelated case. More formally, we have the following multivariate normal distribution, 

$$\begin{pmatrix} \tilde{\theta}_i \\ \mathbf{X_i} \\ Y_{i, T - 1}\\ Y_{i, T - 2} \\ \vdots \\ Y_{i, 0} \end{pmatrix}  \mid Z_i = 0 \sim N\left( \begin{pmatrix} \tilde{\mu}_{\theta ,0} \\ \vec{\mu}_{x, 0}\\ \beta_{\theta, 0}\mu_{\theta, 0} + \vec{\beta}_{x,  0}^\top \vec{\mu}_{x, 0} \\ \beta_{\theta, 0}\mu_{\theta, 0} + \vec{\beta}_{x,  0}^\top \vec{\mu}_{x, 0} \\ \vdots \\ \beta_{\theta, 0}\mu_{\theta, 0} + \vec{\beta}_{x,  0}^\top \vec{\mu}_{x, 0} \end{pmatrix}, \Sigma \right),$$
where 
\begin{align*}
 \Sigma &= \begin{pmatrix} \tilde{\sigma}_{\theta}^2 & \mathbf{B}_{1 \times (T+p)} \\ \mathbf{B}_{(T+p) \times 1}^\top & \mathbf{D}_{(T+p) \times (T+p)} \end{pmatrix}   \\
 \mathbf{B} &= \begin{pmatrix} \vec{0}_{1 \times p} & \beta_{\theta, 0}\tilde{\sigma}_{\theta}^2 & \beta_{\theta, 0}\tilde{\sigma}_{\theta}^2  & \dots & \beta_{\theta, 0} \tilde{\sigma}_{\theta}^2 \end{pmatrix}, \\
\end{align*}
and $\mathbf{C}$ is defined similarly as above and $\mathbf{A}$ is again the same matrix structure as the matrix $\mathbf{Q}$ in Lemma~\ref{lemma:matrix_inversion} with $x = (\beta_{\theta, 0})^2 \tilde{\sigma}_{\theta}^2 + \vec{\beta}_{x,0}^\top \Sigma_{XX}\vec{\beta}_{x,0} + \sigma_E^2$ and $y = x - \sigma_E^2$. Finally, following the same steps as the above proof for the uncorrelated case, we realize that 
$$E(\tilde{\theta_i} \mid Z_i = 0, \mathbf{X_i} = \mathbf{x}, \mathbf{Y_i^T} = \mathbf{y^t}) = \tilde{\mu}_{\theta ,0} + \frac{T(\beta_{\theta, 0})^2 \tilde{\sigma}_{\theta}^2 }{T(\beta_{\theta, 0})^2 \tilde{\sigma}_{\theta}^2 + \sigma_{E}^2}(\tilde{\mu}_{\theta, 1} - \tilde{\mu}_{\theta, 0})$$
Putting this back into Equation~\eqref{eq:bias_both_proof} gives the desired result since $\beta_{\theta, 0} = \beta_{\theta , T-1}$ under Assumption~\ref{assumption:multiple_time_periods}. 
\end{proof}

\section{Matching only on Pre-treatment Outcome}
\label{Appendix:only_pre_treatment}
\textbf{}

In this section we explore the matching estimator that matches only on the pre-treatment outcome $Y_{0}$. We show that it does not offer more insights than that already offered by Theorem~\ref{theorem:mainresults}. We first introduce the estimator that matches only on pre-treatment outcome and state our results under the same univariate confounder setting in Section~\ref{section:main_results}.
\begin{align*}
E\big[\hat{\tau}_{DiD}^{Y_{0}}\big] &= E(Y_{i, 1}\mid Z_{i} = 1) -   E(Y_{i, 0}\mid Z_{i} = 1)   \\
 &- ( E_{y \mid Z_i = 1}[E(Y_{i, 1}\mid Z_{i} = 0, Y_{i, 0}= y)] -  E_{y \mid Z_i = 1}[E(Y_{i, 0}\mid Z_{i} = 0, Y_{i, 0} = y)]) \\
 &= E(Y_{i, 1}\mid Z_{i} = 1) -  E_{y \mid Z_i = 1}[E(Y_{i, 1}\mid Z_{i} = 0, Y_{i, 0}= y)]
\end{align*}

We now provide the bias result for this matching estimator that matches only on the pre-treatment outcome under the same linear model described in Section~\ref{section:setup}.
\begin{theorem}[Bias of DiD and Matching only on Pre-treatment Outcome]
\label{theom:only_match_Y}
If $(Z_i, X_i, \theta_i, Y_{i,t})$ are independently and identically drawn from the data generating process as shown in Equations~\ref{eq:simplethetax}-~\ref{eq:simpleresponsemodel}, then the bias of matching only on the pre-treatment outcome is the following,
\begin{align*}
E\big[\hat{\tau}_{DiD}^{Y_{0}}\big]  - \tau &= \beta_{\theta, 1}\delta_{\theta}(1 - r_{\theta}) + \beta_{x, 1} \delta_x (1 - r_{x}) - \beta_{\theta, 1} \frac{\beta_{x, 0}}{\beta_{\theta, 0}} r_{\theta} \delta_x - \beta_{x, 1}\frac{\beta_{\theta, 0}}{\beta_{x, 0}} r_{x} \delta_{\theta}  ,\\
\end{align*}
where 
\begin{align*}
    r_{\theta} &= \frac{(\beta_{\theta, 0})^2\sigma_{\theta}^2 + \beta_{x, 0}\beta_{\theta, 0}\sigma_{\theta}\sigma_x\rho }{(\beta_{\theta, 0})^2\sigma_{\theta}^2 + (\beta_{x, 0})^2\sigma_x^2 + 2\beta_{x, 0}\beta_{\theta, 0}\sigma_{\theta}\sigma_x\rho + \sigma_{E}^2 } \\
    r_{x} &= \frac{(\beta_{x, 0})^2\sigma_x^2 + \beta_{x, 0}\beta_{\theta, 0}\sigma_{\theta}\sigma_x\rho }{(\beta_{\theta, 0})^2\sigma_{\theta}^2 + (\beta_{x, 0})^2\sigma_x^2 + 2\beta_{x, 0}\beta_{\theta, 0}\sigma_{\theta}\sigma_x\rho + \sigma_{E}^2}. \\
\end{align*}
\end{theorem}
In the case where we do not match on the observed covariate $X$, both $X$ and $\theta$ are treated as ``unobserved'' since we do not directly match on $X$. Therefore, the pre-treatment outcome aims to proxy both of these confounders. Consequently, matching on $Y_{0}$ partially recovers both $\theta$ and $X$ proportional to their respective contribution to the reliability of $Y_{0}$ as represented by the $r_{\theta}, r_{x}$, respectively. $r_{\theta}, r_{x}$ terms approximately captures the proportion of variance of $Y_{0}$ each variable captures respectively. If $\theta$ accounts for more variance of $Y_{0}$, then matching on $Y_{0}$ reduces the confounding effect of $\theta$ through a larger term of $r_{\theta}$. Finally, even if both $X$ and $\theta$ do not have time varying effects on the outcome ($\Delta_{\theta} = \Delta_x = 0$), then matching on the pre-treatment outcome does not result in zero bias since it again erodes the second difference in the DiD estimator as shown by the bias expression above. 

The main difference in this matching estimator compared to the matching estimator that matches on both $X$ and $Y_{0}$ is the third and fourth terms in the above bias expression. Unlike the case when we are able to match on $X$ and perfectly get rid of the effect of $X$, matching only on $Y_0$ allows the confounders' imbalance ($\delta_x, \delta_{\theta}$) to continue making an impact on the match, leading to the third and fourth terms. To illustrate this point, we consider the following hypothetical scenario. Suppose that $\theta$ accounted for all the total variance of $Y_{0}$, i.e.,  $r_{\theta} = 1, r_{x} = \sigma_E = 0$. Therefore, we should hope to perfectly recover $\theta$ when matching for $Y_{0}$. Further suppose that $\mu_{x, 1} = 5$ and $\mu_{x, 0} = 3$ while $\beta_{x, 1} = \beta_{\theta, 1} = 1$. Consequently, $Y_{i, 0} \mid (Z_i = 0) = \theta_i + 3$ for the control while $Y_{i, 0} \mid (Z_i = 1) = \theta_i + 5$ for the treatment (ignoring the intercept). If we observe a treatment unit with $Y_{i, 0} \mid (Z_i = 1) = 6$, thus $\theta_i = 1$, and find a match $j$ from the control unit $Y_{j, 0} \mid (Z_j = 0) = 6$, then $\theta_j = 3 \neq \theta_i$. We see that the difference between $\theta_i$ and $\theta_j$ is proportional (in this case exactly equal since all slopes are one) to $\delta_x$. Even if $\theta$ accounted for the total variance of the pre-treatment outcome, matching on the pre-treatment outcome does not recover $\theta$ due to the imbalance of $X$. Therefore, the third term in the bias expression captures this shift of $\theta$ due to the imbalance of $X$. The fourth term is the symmetrical shift for $X$ due to the imbalance of $\theta$.

In summary, the main insight from matching on just the pre-treatment outcome is similar to that presented in Section~\ref{section:main_results}. On one hand matching on the pre-treatment helps correct for the bias contribution of both $\theta$ and $X$ through the reliability terms of the respective variable. On the other hand, if parallel trends originally held matching on the pre-treatment outcome biases the DiD estimator. From the point of the practitioners, it is insensible to not match on the observed covariate $X$ while matching on the pre-treatment outcome as the pre-treatment outcome now has to proxy both $X$ and $\theta$. Putting all this together, we chose to not present this matching estimator as it is both undesirable and does not lead to further insight than that already provided in Section~\ref{section:main_results}. We now state the proof of Theorem~\ref{theom:only_match_Y}.

\begin{proof}
The proof is very similar to the proof of matching on both $X$ and pre-treatment outcome bias. We first have that:
\begin{align*}
E[\hat{\tau}_{DiD}^{Y_{0}}]  - \tau  &= E(Y_{i, 1}\mid Z_{i} = 1) -  E_{y \mid Z_i = 1}[E(Y_{i, 1}\mid Z_{i} = 0, Y_{i, 0}= y)] ) - \tau \\
&= \beta_{\theta, 1}(\mu_{\theta, 1} - E_{y \mid Z_i = 1}[E(\theta_i \mid Z_i = 0, Y_{i, 0} = y)])  \\
&+ \beta_{x, 1}(\mu_{x, 1} - E_{y \mid Z_i = 1}[E(X_i \mid Z_i = 0, Y_{i, 0} = y)])
\end{align*}
We then characterize the following multivariate distribution
$$\begin{pmatrix} \theta_i \\ Y_{i, 0} \end{pmatrix}  \mid Z_i = 0 \sim N\left( \begin{pmatrix} \mu_{\theta, 0} \\ \beta_{\theta, 0}\mu_{\theta, 0} + \beta_{x, 0}\mu_{x, 0}  \end{pmatrix}, \Sigma \right),$$ where
$$\Sigma = \begin{pmatrix} \sigma_{\theta}^2 & \beta_{\theta, 0}\sigma_{\theta}^2 + \beta_{x, 0}\sigma_{\theta}\sigma_x\rho  \\  \beta_{\theta, 0}\sigma_{\theta}^2 + \beta_{x, 0}\sigma_{\theta}\sigma_x\rho   & (\beta_{\theta, 0})^2 \sigma_{\theta}^2 + (\beta_{x, 0})^2\sigma_x^2 + \sigma_{E}^2 + 2 \beta_{\theta, 0}\beta_{x, 0}\sigma_{\theta}\sigma_x\rho \end{pmatrix}.$$

Now we can apply the Lemma~\ref{lemma:multivariate} to obtain $$E(\theta_i \mid Z_i = 0, Y_{i, 0} = y) = \mu_{\theta, 0} + \frac{\beta_{\theta, 0}\sigma_{\theta}^2 + \beta_{x, 0}\sigma_{\theta}\sigma_x\rho}{d}(y - \beta_{\theta, 0}\mu_{\theta, 0} + \beta_{x, 0}\mu_{x, 0}),$$
where $d$ is defined shortly after. Symmetrically we have that 
$$E(X_{i} \mid Z_{i} = 0, Y_{i, 0} = y) = \mu_{x, 0} + \frac{\beta_{x, 0}\sigma_x^2 + \beta_{\theta, 0}\sigma_{\theta}\sigma_x\rho}{d}(y - \beta_{\theta, 0}\mu_{\theta, 0} + \beta_{x, 0}\mu_{x, 0}).$$
To avoid long algebraic expressions, we define the following variables:
\begin{align*}
    c_1&= \beta_{\theta, 0}\sigma_{\theta}^2 +\beta_{x, 0}\sigma_{\theta}\sigma_x\rho \\
    c_2&= \beta_{x, 0}\sigma_{x}^2 + \beta_{\theta, 0}\sigma_{\theta}\sigma_x\rho \\
    d&= (\beta_{\theta, 0})^2\sigma_{\theta}^2 + (\beta_{x, 0})^2\sigma_x^2 + \sigma_{E}^2 + 2\beta_{x, 0}\beta_{\theta, 0}\sigma_{\theta}\sigma_x\rho. \\
\end{align*}
Finishing off the proof: 
\begin{align*}
E[\hat{\tau}_{DiD}^{Y_{0}}] - \tau  &= \beta_{\theta, 1}(\mu_{\theta, 1} - E_{y \mid Z_i = 1}[E(\theta_i \mid Z_i = 0, Y_{i, 0} = y)]) + \beta_{x, 1}(\mu_{x, 1} - E_{y \mid Z_i = 1}[E(X_i \mid Z_i = 0, Y_{i, 0} = y)])\\ 
&= \beta_{\theta, 1}(\mu_{\theta, 1} - \mu_{\theta, 0} - \frac{\beta_{\theta, 0}\sigma_{\theta}^2 + \beta_{x, 0}\sigma_{\theta}\sigma_x\rho}{d}(\beta_{\theta, 0}(\mu_{\theta, 1} - \mu_{\theta, 0}) + \beta_{x, 0}(\mu_{x, 1} - \mu_{x, 0}) )) \\
&+ \beta_{x, 1}(\mu_{x, 1} -   \mu_{x, 0} - \frac{\beta_{x, 0}\sigma_x^2 + \beta_{\theta, 0}\sigma_{\theta}\sigma_x\rho}{d}(\beta_{\theta, 0}(\mu_{\theta, 1} - \mu_{\theta, 0}) + \beta_{x, 0}(\mu_{x, 1} - \mu_{x, 0}) )) \\
&= \beta_{\theta, 1}(\mu_{\theta, 1} - \mu_{\theta, 0})(1 - \frac{\beta_{\theta, 0} c_1}{d}) + \beta_{x, 1}(\mu_{x, 1} - \mu_{x, 0})(1 - \frac{\beta_{x, 0} c_2}{d}) \\
&- \frac{\beta_{\theta, 1}c_1\beta_{x, 0}}{d}(\mu_{x, 1} - \mu_{x, 0}) - \frac{\beta_{x, 1}c_2\beta_{\theta, 0}}{d}(\mu_{\theta, 1} - \mu_{\theta, 0})
\end{align*}

The final result follows after algebraic simplification. 
\end{proof}

\section{Extending to Include Interactions}
\label{Appendix:interaction}
In this section, we analyze the case when we add an interaction between $\theta$ and $X$.
The idea of an interacted model is allow for the relationship between the latent confounder and outcome to change for different levels of the observed covariate $X$.
The interacted model is
\begin{equation}
Y_{i,t}(0) = \beta_{0, t} +\beta_{\theta, t} \theta_{i}  + \beta_{x, t} X_{i} + \beta_{\theta x, t} X_i \theta_i+ \epsilon_{i, t},
\label{eq:interaction}
\end{equation}
with $Y_{i,t}(1), Y_{i,t}$ defined as they were in Equation~\eqref{eq:simpleresponsemodel}.

In this section, we show the bias results of the na\"ive DiD and the two main matching estimators when $\rho = 0$ for this model.
We focus on the uncorrelated case to clearly show how the interaction term does not lead to any new insights than those provided in Section~\ref{section:main_results}.
Adding a correlation between $\theta$ and $X$ only contributes by further offsetting the bias of $\theta$ with the correlation, as described in Section~\ref{subsection:full_general_case}. 

Letting the interaction imbalance and time varying effects be
\begin{align*}
\delta_{\theta x} & := \mu_{\theta, 1}\mu_{x, 1} - \mu_{\theta, 0}\mu_{x, 0} \\
\Delta_{\theta x} & := \beta_{\theta x, 1} - \beta_{\theta x, 0}, 
\end{align*}
we then have the following theorem.

\begin{theorem}[Bias of DiD and Matching DiD estimators with interactions]
If $(Z_i, X_i, \theta_i, Y_{i,t})$ are independently and identically drawn from the data generating process as shown in Equations~\ref{eq:simplethetax} \& \ref{eq:interaction} with $\rho = 0$, then the bias of our estimators are the following:
\begin{align*}
E\big[\hat{\tau}_{DiD}\big]- \tau &= \Delta_{\theta}\delta_{\theta} + \Delta_x\delta_x + \Delta_{\theta x} \delta_{\theta x}\\
E\big[\hat{\tau}_{DiD}^X\big]- \tau &=  \Delta_{\theta}\delta_{\theta} + \Delta_{\theta x} \mu_{x,1} \delta_{\theta}\\
E\big[\hat{\tau}_{DiD}^{X, Y_{0}}\big]- \tau &=  \beta_{\theta, 1}\delta_{\theta}(1 - \tilde{r}_{\theta \mid x})  + \beta_{\theta x, 1} \mu_{x, 1} \delta_{\theta} (1 - \tilde{r}_{\theta \mid x}) - (\beta_{\theta,1} + \beta_{\theta x,1}\mu_{x,1})L ,
\end{align*}
where 
\begin{align*}
    \tilde{r}_{\theta \mid x} &= \frac{ (\beta_{\theta, 0})^2 \sigma_{\theta}^2 + \beta_{\theta, 0}\beta_{\theta x, 0} \mu_{x, 0}\sigma_{\theta}^2}
    {(\beta_{\theta, 0})^2 \sigma_{\theta}^2 + 2\beta_{\theta, 0}\beta_{\theta x, 0} \mu_{x, 0}\sigma_{\theta}^2 + \beta_{\theta x , 0}^2(\sigma_{\theta}^2 \sigma_x^2 + \sigma_{\theta}^2 \mu_{x, 0}^2) + \sigma_E^2} \\
    L &= \frac{(\beta_{\theta x, 0} \mu_{x,0} + \beta_{\theta, 0})(\delta_{\theta x} - \delta_x\mu_{\theta, 0})\beta_{\theta x, 0} \sigma_{\theta}^2}
    {(\beta_{\theta, 0})^2 \sigma_{\theta}^2 + 2\beta_{\theta, 0}\beta_{\theta x, 0} \mu_{x, 0}\sigma_{\theta}^2 + \beta_{\theta x , 0}^2(\sigma_{\theta}^2 \sigma_x^2 + \sigma_{\theta}^2 \mu_{x, 0}^2) + \sigma_E^2}
\end{align*}
\label{theorem:results_interaction}
\end{theorem}

\begin{proof}
We derive each estimator separately, except for the na\"ive DiD bias as it is straight forward. 

\noindent \textbf{Proof for matching on observed covariates:}

The first two terms of the DiD estimator are unchanged by the matching process: 
\begin{align*}
E(Y_{i, 1}\mid Z_{i} = 1) -  E(Y_{i, 0}\mid Z_{i} = 1) &= \mu_{\theta, 1}(\beta_{\theta, 1} - \beta_{\theta, 0}) + \mu_{x, 1}(\beta_{x, 1} - \beta_{x, 0}) +\mu_{\theta, 1}\mu_{x, 1}(\beta_{\theta x, 1} - \beta_{\theta x, 0}) 
\end{align*}
Now focusing on last two term's inner expectation:
\begin{align*}
&(E(Y_{i, 1}\mid Z_{i} = 0, X_i = x) -  E(Y_{i, 0}\mid Z_{i} = 0, X_i = x))  \\
& \quad = \mu_{\theta, 0}(\beta_{\theta, 1} - \beta_{\theta, 0}) + x(\beta_{x, 1} - \beta_{x, 0}) +\mu_{\theta, 0} x(\beta_{\theta x, 1} - \beta_{\theta x, 0}) 
\end{align*}
Taking expectation of $x$ over the distribution of $Z_i = 1$ gives us the desired result. 

\noindent \textbf{Proof for matching on observed covariates and pre-treatment outcome:}

\begin{align*}
E\big[\hat{\tau}_{DiD}^{X, Y_{0}}\big] - \tau &= E(Y_{i, 1}\mid Z_{i} = 1) -  E_{(x,y) \mid Z_i = 1}[E(Y_{i, 1}\mid Z_{i} = 0, X_i = x, Y_{i, 0} = y)] - \tau \\
&= \beta_{\theta, 1}(\mu_{\theta, 1} - E_{(x, y) \mid Z_i = 1}[E(\theta_i \mid Z_i = 0, X_i = x, Y_{i, 0} = y))] \\
&+ \beta_{\theta x, 1}\mu_{x, 1}(\mu_{\theta, 1} -E_{(x, y) \mid Z_i = 1}[E(\theta_i \mid Z_i = 0, X_i = x, Y_{i, 0} = y))]
\end{align*}

We first establish the following multivariate normal distribution,

$$\begin{pmatrix} \theta_i \\ X_i \\ Y_{i, 0} \end{pmatrix}  \mid Z_i = 0 \sim N\left( \begin{pmatrix} \mu_{\theta, 0} \\ \mu_{x, 0}\\ \beta_{\theta, 0}\mu_{\theta, 0} + \beta_{x, 0} \mu_{x, 0} + \beta_{\theta x, 0} \mu_{\theta, 0} \mu_{x, 0}  \end{pmatrix},  \begin{pmatrix} \sigma_{\theta}^2 &  0 & c_1 \\  0  & \sigma_x^2 & c_2 \\   c_1 &   c_2   & d \end{pmatrix}\right) , $$
where we again define the following variables to avoid long algebraic expressions
\begin{align*}
Var(Y_{i, 0} \mid Z_i =0) & := d = (\beta_{\theta, 0})^2 \sigma_{\theta}^2 + (\beta_{x, 0})^2 \sigma_x^2 + (\beta_{\theta x, 0})^2 (\sigma_{\theta}^2\sigma_x^2 + \sigma_x^2(\mu_{\theta, 0})^2 + \sigma_{\theta}^2(\mu_{x, 0})^2) \\
&+ 2\beta_{\theta, 0} \beta_{\theta x, 0} \mu_{x, 0} \sigma_{\theta}^2 + 2\beta_{x, 0} \beta_{\theta x, 0} \mu_{\theta, 0} \sigma_x^2 + \sigma_{E}^2 \\
Cov(\theta_i,Y_{i, 0} \mid Z_i = 0) & := c_1 =  \beta_{\theta, 0}\sigma_{\theta}^2 +  \beta_{\theta x, 0} \mu_{x, 0} \sigma_{\theta}^2 \\
Cov(X_i, Y_{i, 0} \mid Z_i = 0) &:= c_2 =  \beta_{x, 0} \sigma_x^2 + \beta_{\theta x, 0} \mu_{\theta, 0} \sigma_x^2 
\end{align*}

Using Lemma~\ref{lemma:multivariate}: $E_{(x,y) \mid Z_i = 1}[(E(\theta_i \mid Z_i = 0, X_i = x, Y_{i, 0} = y))] = \mu_{\theta, 0} + \frac{1}{det}( \delta_x(\beta_{x, 0} \sigma_x^2 c_1 -c_1c_2) + \delta_{\theta} \beta_{\theta, 0} \sigma_x^2 c_1 + \beta_{\theta x, 0} \sigma_x^2 c_1 (\mu_{\theta, 1} \mu_{x, 1} - \mu_{\theta, 0}\mu_{x, 0}) )$, where $det = \sigma_x^2 d - c_2^2$.

Putting this together we have:

\begin{align*}
E\big[\hat{\tau}_{DiD}^{X, Y_{0}}\big] &= \beta_{\theta, 1}(\mu_{\theta, 1} -  (\mu_{\theta, 0} + \frac{1}{det}( \delta_x(\beta_{x, 0} \sigma_x^2 c_1 -c_1c_2) + \delta_{\theta} \beta_{\theta, 0} \sigma_x^2 c_1 + \beta_{\theta x, 0}\sigma_x^2 c_1 (\mu_{\theta, 1} \mu_{x, 1} - \mu_{\theta, 0}\mu_{x, 0}) )) ) \\
&+ \beta_{\theta x, 1}\mu_{x, 1}(\mu_{\theta, 1} -  (\mu_{\theta, 0} + \frac{1}{det}( \delta_x(\beta_{x, 0} \sigma_x^2 c_1 -c_1c_2) + \delta_{\theta} \beta_{\theta, 0} \sigma_x^2 c_1 + \beta_{\theta x, 0}\sigma_x^2 c_1 (\mu_{\theta, 1} \mu_{x, 1} - \mu_{\theta, 0}\mu_{x, 0}) )) ) 
\end{align*}

The result follows from algebraic simplifications. 
\end{proof}

\subsection{Interpretation of Bias Results}
We first focus on the bias of matching on $X$ only.
Section~\ref{section:main_results} shows matching on the observed covariates helped mitigate the bias contribution of $X$ but does not fully reduce the bias to zero due to the latent confounder.
The interaction case is similar, except matching on $X$ partially removes the bias contribution of the interaction by allowing the bias contribution of the interaction term to come through the imbalance of $\theta$ (as opposed to both the imbalance of $\theta$ and $X$ via $\delta_{\theta x}$).
Contrary to what one may have expected, matching on $X$ when $X$ and $\theta$ are interacted does not improve the plausibility of the conditional parallel trend (conditional on $X$) compared to the no-interaction case.
Similar to the results in Section~\ref{section:main_results}, the conditional parallel trends (conditional on $X$) only holds when the unconditional parallel trends originally held due to the latent confounder $\theta$.

When matching additionally on the pre-treatment outcome, we again have a similar story to the no-interaction results presented in Section~\ref{section:main_results}.
First, $\tilde{r}_{\theta \mid x}$ acts as the new ``reliability'' when we have interactions.
In this case, matching on the pre-treatment outcome reduces the bias by both main effect and interaction effect of $\theta$ through this new reliability term.
The only difference is the ``leftover'' terms $(\beta_{\theta,1} + \beta_{\theta x,1}\mu_{x,1})L$ in the bias expression.
These terms appear because matching on $X$ does not cleanly get rid of all contribution of $X$ due to the interaction similar to the results presented in Appendix~\ref{Appendix:only_pre_treatment}.
Therefore, we see $\delta_x$ and $\delta_{\theta x}$ continuously making an impact on the bias by acting as a hindrance to how well we recover $\theta$.
This ``leftover'' effect of $X$ is captured by the $L$ term and will be exactly zero if the interaction effect did not exist, i.e., if $\beta_{\theta x, 0} = 0$.
In summary, because the main ideas presented in Section~\ref{section:main_results} still hold in the interaction case, we choose to leave this interaction case in the Appendix.

\section{Proofs of Guidelines}
\label{appendix:sens}
\subsection{Proof of Theorem~\ref{theorem:matching_X}}
\label{appendix:proof_guidelineX}
\begin{proof}
We prove the first part of Theorem~\ref{theorem:matching_X}, i.e., that the three sign conditions guarantee that $\Delta_{\tau_x} > 0$. Although Section~\ref{subsection:rule_matching_X} states the results assuming $\theta$ is univariate, the result holds generally for multivariate $\boldsymbol{\theta}$, thus we prove it in this more general case. In other words, we want to show that given our three sign conditions we have the following,
$$
\Delta_{\tau_x} \geq 0  \Longleftrightarrow |\vec{\Delta}_{\theta}^\top \vec{\delta}_{\theta} + \vec{\Delta}_x^\top\vec{\delta}_x| \geq |\vec{\Delta}_{\theta}^\top [\vec{\delta}_{\theta} -  \Sigma_{\theta X} \Sigma_{XX}^{-1} \vec{\delta}_x] |
$$
We have two cases.
The first case is if $\vec{\Delta}_{\theta}^\top \vec{\delta}_{\theta} > 0$. We can then drop the absolute value for the left side because of the first sign condition. $\vec{\Delta}_{\theta}^\top \vec{\delta}_{\theta} > 0$ can hold in two ways. The first is when $\vec{\Delta}_{\theta}^\top >0, \vec{\delta}_{\theta} >0$. In this case, we can drop the absolute value for the right side because of the third sign condition, leading to the following inequality
$$|\vec{\Delta}_{\theta}^\top \vec{\delta}_{\theta} + \vec{\Delta}_x^\top\vec{\delta}_x| \geq |\vec{\Delta}_{\theta}^\top [\vec{\delta}_{\theta} -  \Sigma_{\theta X} \Sigma_{XX}^{-1} \vec{\delta}_x] |  \Longleftrightarrow  \vec{\Delta}_x^\top\vec{\delta}_x \geq \vec{\Delta}_{\theta}^\top [- \Sigma_{\theta X} \Sigma_{XX}^{-1} \vec{\delta}_x]   $$
By the second sign condition, the left hand side is positive while the right hand side is negative, thus we are done. $\vec{\Delta}_{\theta}^\top \vec{\delta}_{\theta} > 0$ holds also when $\vec{\Delta}_{\theta}^\top < 0, \vec{\delta}_{\theta} < 0$. In this case, we can similarly drop the absolute value for the right side because of the third sign condition, leading to the same steps as above, thus proving the first case. 

The second case is if $\vec{\Delta}_{\theta}^\top \vec{\delta}_{\theta} < 0$. Using the first sign condition, we can drop the left hand side's absolute value after multiplying by negative one. Following the same steps as above, $\vec{\Delta}_{\theta}^\top \vec{\delta}_{\theta} < 0$ when $\vec{\Delta}_{\theta}^\top < 0, \vec{\delta}_{\theta} >0$. In this case, we can similarly drop the absolute value for the right side after multiplying by negative one because of the third sign condition, leading to the following inequality
$$|\vec{\Delta}_{\theta}^\top \vec{\delta}_{\theta} + \vec{\Delta}_x^\top\vec{\delta}_x| \geq |\vec{\Delta}_{\theta}^\top [\vec{\delta}_{\theta} -  \Sigma_{\theta X} \Sigma_{XX}^{-1} \vec{\delta}_x] |  \Longleftrightarrow  -\vec{\Delta}_x^\top\vec{\delta}_x \geq \vec{\Delta}_{\theta}^\top [\Sigma_{\theta X} \Sigma_{XX}^{-1} \vec{\delta}_x]   $$
This again follows from the second sign condition since the left hand side is positive while the right hand side is negative. In the other case when $\vec{\Delta}_{\theta}^\top > 0, \vec{\delta}_{\theta} < 0$ the proof is symmetrical.
\end{proof}

\begin{proof}
We now prove the second part of Theorem~\ref{theorem:matching_X}. In other words, we want to show that our estimator consistently estimates:
$$\Delta_{\tau_x} = |\vec{\delta}_x^\top( \vec{\Delta}_x +  \Delta_{\theta}\Sigma_{\theta X} \Sigma_{XX}^{-1})|,$$
where $\theta$ is univariate and $\Sigma_{\theta X}$ is an arbitrary covariate matrix. We are also able to simplify $\Delta_{\tau_x}$ (without some absolute values) because of the three sign conditions. 

First we note that estimating $\vec{\delta}_x$ with $\hat{\vec{\delta}}_x$ is an unbiased and consistent estimator for $\vec{\delta}_x$. It remains to show that $(\hat{\vec{\beta}}_{x, T} - \frac{\sum_{t = 0}^{T - 1} \hat{\vec{\beta}}_{x, t}}{T})$ consistently estimates $\vec{\Delta}_x +  \Delta_{\theta}\Sigma_{\theta X} \Sigma_{XX}^{-1}$. 

Because $\hat{\vec{\beta}}_{x, t}$ is the resulting $p$-dimensional regression coefficient from a linear regression of $Y_{i,t}$ on $\mathbf{X}_i$ within the control group, we have from elementary results of linear regression that $\hat{\vec{\beta}}_{x, t}$ converges to the coefficient of $\mathbf{X}_i$ in $E(Y_{i, t} \mid Z_i = 0, \mathbf{X}_i)$. It is simple to show that (ignoring intercept), 
\begin{align*}
E(Y_{i, t} \mid Z_i = 0, \mathbf{X}_i) &= \beta_{\theta, t}(\mu_{\theta, 0} + \Sigma_{\theta X} \Sigma_{XX}^{-1} (\mathbf{X}_i - \vec{\mu}_{x, 0})) + \vec{\beta}_{x, t}^{\top} \mathbf{X}_i \\
&= \beta_{\theta, t}(\mu_{\theta, 0} - \Sigma_{\theta X}\Sigma_{XX}^{-1} \vec{\mu}_{x, 0}) + (\vec{\beta}_{x, t}^\top + \beta_{\theta ,t}\Sigma_{\theta X} \Sigma_{XX}^{-1}) \mathbf{X}_i,
\end{align*}
which holds because 
$$E(\theta_i \mid Z_i =0, \mathbf{X}_i)  = \mu_{\theta, 0} + \Sigma_{\theta X} \Sigma_{XX}^{-1} (\mathbf{X}_i - \vec{\mu}_{x, 0}).$$
Therefore, we have that $\hat{\vec{\beta}}_{x, t}$ converges to $(\vec{\beta}_{x, t}^\top + \beta_{\theta , t}\Sigma_{\theta X} \Sigma_{XX}^{-1})$. Consequently, $(\hat{\vec{\beta}}_{x, T} - \frac{\sum_{t = 0}^{T - 1} \hat{\vec{\beta}}_{x, t}}{T})$ converges to $\vec{\Delta}_x +  \Delta_{\theta}\Sigma_{\theta X} \Sigma_{XX}^{-1}$ as desired. 
\end{proof}

\subsection{Proof of Lemma~\ref{lemma:general_suffcond_match} and Theorem~\ref{theorem:matching_Y}}
\label{appendix:proof_guidelineY}
Lemma~\ref{lemma:general_suffcond_match} directly follows from Lemma~\ref{lemma:suffcond_match}. More formally, we have that 
$$ \Delta_{\tau_{x,y}} = |\Delta_{\theta} \delta_{\theta}| - |\beta_{\theta ,T} \delta_{\theta} (1 - \hat r_{\theta \mid x}^T)|$$
We then directly use the proof for Lemma~\ref{lemma:suffcond_match}.

We now prove Theorem~\ref{theorem:matching_Y}. In other words, we want to show that 
$$\hat r_{\theta}^T \rightarrow r_{\theta}^T, \quad \frac{\hat{\bar{\beta}}_{\theta, \text{pre}}}{{\hat\beta_{\theta, T}}} \rightarrow\frac{ \bar{\beta}_{\theta, \text{pre}}}{\beta_{\theta, T}}$$

\begin{proof}
First, we start with what $\tilde{Y}_{i,t}$ converges to. As shown in Appendix~\ref{appendix:proof_guidelineX}, the regression coefficient of regressing $Y_{i, t}$ on $\mathbf{X}_i$ within the control (or treatment) group converges to $\vec{\beta}_{x, t}^\top + \beta_{\theta , t}\Sigma_{\theta X} \Sigma_{XX}^{-1}$, or a shifted estimand that contains the explanatory aspect of $\theta$ on $Y$ obtainable via $\mathbf{X}$. 

If we residualize using this model, we get that $\tilde{Y}_{i,t}$ converges to (by Slutsky's Theorem)
$$\tilde{Y}_{i, t}  \xrightarrow{d}  \beta_{\theta, t}(\theta_i - \Sigma_{\theta X} \Sigma_{XX}^{-1} X_i) + \epsilon_{i,t} =  \beta_{\theta, 0}\tilde{\theta}_i + N\left(0, \sigma_{E}^2\right).$$
Consequently, instead of the old latent confounder $\theta_i$, we have a new latent confounder $\tilde{\theta}_i = \theta_i - \Sigma_{\theta X} \Sigma_{XX}^{-1} X_i$.

We next show the consistency of $\hat{\sigma}_{E}^2 $. First, with a slight abuse of notation we define the population version of $\tilde{Y}_{i, t}$ as
$$\tilde{Y}_{i, t} = \beta_{\theta, t}\tilde{\theta}_i + \epsilon_{i, t}.$$
We then have that
$$Var(\tilde{Y}_{i, T - 1} -  \tilde{Y}_{i, T - 2}) = (\beta_{\theta, T-1} - \beta_{\theta, T-2})^2\tilde{\sigma}_{\theta}^2 + 2\sigma_E^2.$$
Since we assume $\beta_{\theta, T-1} = \beta_{\theta, T-2}$ we have that $\hat \sigma_E^2$ is a consistent estimator using the continuous mapping theorem and Slutsky's Theorem. 

Now, since $Var(\tilde{Y}_{i, t} \mid Z_i = z) = \beta_{\theta, t}^2 \tilde{\sigma}_{\theta}^2 + \sigma_E^2$ for every $t$, we have that 
$$\hat \beta_{\theta, t} =  \sqrt{ \widehat{Var}(\tilde{Y}_{i, t} \mid Z_i = 0) - \hat{\sigma}_{E}^2  } \xrightarrow{p} \beta_{\theta , t}\tilde{\sigma}_{\theta}.$$
Although this is not consistent for $\beta_{\theta, t}$ by a scale factor of $\tilde{\sigma}_{\theta}$, we show later that this is irrelevant. Putting this together we have, by Slutsky's theorem, that
$$\frac{ \hat{\bar{\beta}}_{\theta, \text{pre}}}{\hat\beta_{\theta, T}} \xrightarrow{p} \frac{\hat{\bar{\beta}}_{\theta, \text{pre}} \tilde{\sigma}_{\theta}}{\beta_{\theta, T} \tilde{\sigma}_{\theta}} = \frac{\bar{\beta}_{\theta, \text{pre}}}{\beta_{\theta, T} }, $$
giving us the desired consistency for $\frac{\bar{\beta}_{\theta, \text{pre}} }{\beta_{\theta, T} }$. We now turn to how we can estimate $\tilde{\delta}_{\theta}$. We have that, 
$$E( \overline{Y}_{i, 0:(T-1)} \mid Z_i = 1) - E(\overline{Y}_{i, 0:(T-1)} \mid Z_i = 0) = \bar{\beta}_{\theta, \text{pre}} \delta_{\theta}.$$
Therefore, a plug-in estimator is,
$$\hat{\tilde{\delta}}_{\theta} =\frac{\hat{E}(\bar{\tilde{Y}}_{i, 0:(T-1)} \mid Z_i = 1) - \hat{E}(\bar{\tilde{Y}}_{i, 0:(T-1)} \mid Z_i = 0)}{\hat{\bar{\beta}}_{\theta, \text{pre}}}$$
Lastly, $r_{\theta}^T = \frac{T(\bar{\beta}_{\theta, \text{pre}})^2\tilde{\sigma}_{\theta}^2}{T(\bar{\beta}_{\theta, \text{pre}})^2\tilde{\sigma}_{\theta}^2 + \sigma_{E}^2} $ can also be consistently estimated via the plug-in estimator, 
$$ \hat r_{\theta}^T = \frac{T \hat{\bar{\beta}}_{\theta, \text{pre}}^2 } {T \hat{\bar{\beta}}_{\theta, \text{pre}}^2 + \hat\sigma_E^2}
 \xrightarrow{p} \frac{T(\hat{\bar{\beta}}_{\theta, \text{pre}})^2\tilde{\sigma}_{\theta}^2}{T(\hat{\bar{\beta}}_{\theta, \text{pre}})^2\tilde{\sigma}_{\theta}^2 + \sigma_{E}^2} = r_{\theta}^T, $$
proving we have a consistent estimator of the reliability.

Lastly, we remark that although we took the positive square root when obtaining $\hat\beta_{\theta, t}$, even if we took $\hat\beta_{\theta, t} = -\sqrt{ \widehat{Var}(\tilde{Y}_{i, t} \mid Z_i = 0) - \hat{\sigma}_{E}^2  }$ the bias estimates will remain identical since $\hat\beta_{\theta, t}$ gets squared, i.e., the sign does not matter. 
\end{proof}

\end{document}